\documentclass[11pt]{article}
\usepackage{blindtext}
\usepackage[usenames,dvipsnames,svgnames,table]{xcolor}
\usepackage{enumitem}
\usepackage{caption}
\usepackage{placeins}
\usepackage{graphicx}

\usepackage{fancybox}
\usepackage{hyperref}

\usepackage{amsmath,amsfonts,amsthm,amssymb,xcolor}
\usepackage{bm} 

\usepackage{nicefrac}

\usepackage{float}

\newtheorem{theorem}{Theorem}[section]

\newtheorem{corollary}[theorem]{Corollary}
\newtheorem{lemma}[theorem]{Lemma}

\newtheorem{claim}[theorem]{Claim}

\theoremstyle{definition}
\newtheorem{definition}[theorem]{Definition}

\newtheorem*{theorem*}{Theorem}
\newtheorem*{corollary*}{Corollary}
\newtheorem*{conjecture*}{Conjecture}
\newtheorem*{lemma*}{Lemma}
\newtheorem*{thm*}{Theorem}
\newtheorem*{prop*}{Proposition}
\newtheorem*{obs*}{Observation}
\newtheorem*{definition*}{Definition}

\newtheorem*{remark*}{Remark}
\newtheorem*{rec*}{Recommendation}

\newenvironment{fminipage}%
  {\begin{Sbox}\begin{minipage}}%
  {\end{minipage}\end{Sbox}\fbox{\TheSbox}}

\def\floor#1{\left\lfloor #1 \right\rfloor}

\def\abs#1{\left|#1  \right|}

\def\norm#1{\left\| #1 \right\|}

\def\calH{\mathcal{H}}

\def\aa{\pmb{\mathit{a}}}
\newcommand\bb{\boldsymbol{\mathit{b}}}
\newcommand\cc{\boldsymbol{\mathit{c}}}
\newcommand\dd{\boldsymbol{\mathit{d}}}
\newcommand\ee{\boldsymbol{\mathit{e}}}
\newcommand\ff{\boldsymbol{\mathit{f}}}

\renewcommand\ss{\boldsymbol{\mathit{s}}}

\newcommand\uu{\boldsymbol{\mathit{u}}}

\newcommand\xx{\boldsymbol{\mathit{x}}}

\newcommand\veczero{\boldsymbol{0}}
\newcommand\vecone{\boldsymbol{1}}

\renewcommand\AA{\boldsymbol{\mathit{A}}}

\newcommand\II{\boldsymbol{\mathit{I}}}

\newcommand\AAtil{\boldsymbol{\mathit{\tilde{A}}}}

\newcommand\Rtil{\mathit{\tilde{R}}}

\newcommand\bbtil{\boldsymbol{\tilde{\mathit{b}}}}
\newcommand\cctil{\boldsymbol{\tilde{\mathit{c}}}}

\newcommand\xxtil{\boldsymbol{\tilde{\mathit{x}}}}

\newcommand\Otil{\widetilde{O}}

\newcommand\R{\mathbb{R}}

\DeclareMathOperator{\nnz}{nnz}

\DeclareMathOperator*{\poly}{poly}

\def\eps{\epsilon}

\usepackage{fullpage}
\pdfoutput=1

\usepackage[letterpaper,
left=1in,right=1in,top=1in,bottom=1in]{geometry}

\usepackage[
  backend=biber,
  backref=true,
  backrefstyle=none,
  date=year,
  doi=false,
  giveninits=true,
  hyperref=true,
  maxbibnames=10,
  sortcites=false,
  style=alphabetic,
  url=false, 
]{biblatex}
\usepackage{amsmath}
\usepackage{amssymb}
\usepackage{amsthm}
\usepackage{graphicx}
\usepackage{subfigure}
\usepackage{float}
\usepackage{color}
\usepackage{pifont}
\usepackage[ruled, linesnumbered]{algorithm2e}
\usepackage{booktabs}
\usepackage{multirow}
\usepackage{threeparttable}
\usepackage{arydshln}
\usepackage{multirow}
\usepackage{threeparttable}
\usepackage[toc,page]{appendix}
\usepackage{bbm}
\usepackage{mathrsfs} 
\usepackage{tablefootnote}

\newcommand{\todo}[1]{{\bf \color{red} TODO: #1}}
\newcommand{\todolow}[1]{{\bf \color{o range} TODOLOW: #1}}

\newcommand{\rasmus}[1]{{\bf \color{olive} Rasmus: #1}}
\newcommand{\ming}[1]{{\bf \color{blue} [Ming: #1]}}
\newcommand{\peng}[1]{{\bf \color{orange}[Peng: #1]}}

\renewcommand\todo[1]{}
\renewcommand{\todolow}[1]{}
\renewcommand{\rasmus}[1]{}
\renewcommand{\ming}[1]{}
\renewcommand{\peng}[1]{}

\newcommand\Z{\mathbb{Z}}

\begin{document}

\title{Two-Commodity Flow is Equivalent to Linear Programming \\under
  Nearly-Linear Time Reductions}

\author{
 Ming Ding\\ 
  \texttt{ming.ding@inf.ethz.ch}\\
  Department of Computer Science\\
  ETH Zurich
  \and
 Rasmus Kyng\thanks{The research leading to these results has received funding from grant no. 200021 204787 of the Swiss National Science Foundation.}\\ 
  \texttt{kyng@inf.ethz.ch}\\
  Department of Computer Science\\
  ETH Zurich
  \and
  Peng Zhang\\
  \texttt{pz149@rutgers.edu}\\
  Department of Computer Science\\
  Rutgers University
}

\date{}

\clearpage\maketitle


\thispagestyle{empty}

\begin{abstract}
  We give a nearly-linear time reduction that encodes any linear
  program as a 2-commodity flow problem with only a small
  blow-up in size.
  Under mild assumptions similar to those employed by
  modern fast solvers for linear programs, our reduction causes only a
  polylogarithmic multiplicative increase in the size of the program, and runs in
  nearly-linear time.
  Our reduction applies to high-accuracy approximation algorithms and
  exact algorithms.
  Given an approximate solution to the 2-commodity flow problem, we
  can extract a solution to the linear program in linear time with
  only a polynomial factor increase in the error. 
  This implies that 
  any algorithm that solves the 2-commodity flow problem can solve linear
  programs in essentially the same time.
Given a directed graph with edge capacities and two source-sink pairs, the goal of the 2-commodity flow problem is to maximize the sum of the flows routed between the two source-sink pairs subject to edge capacities and flow conservation.
A 2-commodity flow problem can be formulated as a linear program,
which can be solved to high accuracy in almost the current matrix
multiplication time (Cohen-Lee-Song JACM'21).
Our reduction shows
that linear programs can be approximately solved, to high accuracy,
using 2-commodity flow as well.

Our proof follows the outline of Itai’s polynomial-time reduction of a linear program to a 2-commodity flow problem (JACM’78).
Itai's reduction shows that exactly solving 2-commodity flow and
exactly solving linear programming are polynomial-time equivalent. We
improve Itai’s reduction to nearly preserve the problem representation size
in each step.
In addition, we establish an error bound for approximately solving each intermediate problem in the reduction, and show that the accumulated error is polynomially bounded. We remark that our reduction does not run in strongly polynomial time and that it is open whether 2-commodity flow and linear programming are equivalent in strongly polynomial time.
\end{abstract}

\newpage

\sloppy

\pagenumbering{arabic} 

\section{Introduction}

Multi-commodity maximum flow is a very well-studied problem, which can
be formulated as a linear program.
In this paper, we show that general linear programs can be very
efficiently encoded as a multi-commodity maximum flow programs.
Many variants of multi-commodity flow problems exist. We consider one
of the simplest directed variants, 2-commodity maximum through-put
flow.
Given a directed graph with edge capacities and two source-sink pairs,
this problem requires us to maximize the sum of the flows routed
between the two source-sink pairs, while satisfying capacity
constraints and flow conservation at the remaining nodes.
In the rest of the paper, we will simply refer to this as
\emph{the 2-commodity flow problem}.
We abbreviate this problem as 2CF.
Our goal is to relate the hardness of solving 2CF to that of solving
linear programs (LPs).
2-commodity flow is easily expressed as a linear program, so it is
clearly no harder than solving LPs.
We show that the 2-commodity flow problem can encode a linear
program with only a polylogarithmic blow-up in size, when the program
has polynomially bounded integer entries and polynomially bounded solution norm.
Our reduction runs in nearly-linear time. Given an approximate
solution to the 2-commodity flow problem, we can recover, in linear time, an
approximate solution to the linear program with only a polynomial
factor increase in the error.
Our reduction also shows that an exact solution to the flow problem yields an
exact solution to the linear program.

Multi-commodity flow problems are extremely well-studied and have been
the subject of numerous surveys \cite{K78, AMO93, OMV00, BKV09, W18},
in part because a large number of problems can be expressed as
variants of multi-commodity flow.
Our result shows a very strong form of this observation: In fact,
general linear programs can be expressed
as 2-commodity flow problems with essentially the same size.
Early in the study of these problems, before a polynomial-time
algorithm for linear programming was known, it was shown that the
\emph{undirected} 2-commodity flow problem can be solved in polynomial
time \cite{H63}. In fact, it can be reduced to two undirected single commodity
maximum flow problems.
In contrast, directed 2-commodity flow problems were seemingly harder,
despite the discovery of non-trivial algorithms for some special cases
\cite{E76,E78}.

\paragraph{Searching for multi-commodity flow solvers.}
Alon Itai \cite{I78} proved a polynomial-time reduction from
linear programming to 2-commodity flow, before a polynomial-time
algorithm for linear programming was known.
For decades, the only major progress on solving multi-commodity flow
came from improvements to general linear program solvers
\cite{K80,K84,R88,V89}.
Leighton et al. \cite{LMPSTT95} showed that undirected capacitated
$k$-commodity flow in a graph with $m$ edges and $n$ vertices
can be approximately solved in $\Otil(kmn)$ time,
completely routing all demands with  $1+\epsilon$
times the optimal congestion, 
albeit with a poor dependence
on the error parameter $\epsilon$.
This beats solve-times for linear programming in sparse
graphs for small $k$, even with today's LP solvers that run in current
matrix multiplication time, albeit with much worse error.
This result spurred a number of follow-up works with improvements for
low-accuracy algorithms \cite{GK07, F00, M10}.
Later, breakthroughs in achieving almost- and nearly-linear time
algorithms for undirected single-commodity maximum flow also lead to
faster algorithms for undirected $k$-commodity flow \cite{KLOS14,
  S13,P16}, culminating in Sherman's introduction of area-convexity to
build a $\Otil(mk\epsilon^{-1})$ time algorithm for approximate undirected
$k$-commodity flow \cite{S17}.

\paragraph{Solving single-commodity flow problems.}
Single commodity flow problems have been an area of tremendous success
for the development of graph algorithms, starting with an era of algorithms
influenced by early results on maximum flow and minimum cut
\cite{FF56} and later the development of
powerful combinatorial algorithms for maximum
flow with polynomially bounded edge capacities \cite{D70,ET75,GR98}.
Later, a breakthrough nearly-linear time algorithm for electrical
flows by Spielman and Teng \cite{ST04} lead to
the \emph{Laplacian paradigm}.
A long line of work explored direct improvements and simplifications
of this result \cite{KMP10,KMP11,KOSZ13,PS14,KS16,JS21}.
This also motivated a new line of research on undirected maximum flow
\cite{CKMST11,LRS13,KLOS14,S13}, which in turn lead to faster algorithms
for directed maximum flow and minimum cost flow \cite{M13,M16,LS20,KLS20,vdBLLSSSW21,GLP21} building on
powerful tools using mixed-$\ell_2,\ell_p$-norm minimizing flows \cite{KPSW19} and
inverse-maintenance ideas \cite{CGHPS20}.
Certain developments are particularly relevant to our result: For a
graph $G = (V,E)$ these works established high-accuracy algorithms with
$\Otil(|E|)$ running time for computing electrical flow \cite{ST04} and $O(|E|^{4/3})$
running time for unit capacity directed maximum flow \cite{M13,KLS20},
and  $\Otil(\min(|E|^{1.497},|E|+|V|^{1.5}))$ running time for
directed maximum flow with general capacities
\cite{GLP21,vdBLLSSSW21}.

\paragraph{Solving general linear programs.}
As described in the previous paragraphs, there has been tremendous
success in developing fast algorithms for single-commodity flow
problems and undirected multi-commodity flow problems, albeit
in the latter case only in the low-accuracy regime
(as the algorithm running times depend polynomially on the error parameter).
In contrast, the best known algorithms for directed multi-commodity
flow simply treat the problem as a general linear program, and use a
solver for general linear programs.

The fastest known solvers for general linear programs are based on
interior point methods \cite{K84}, and in particular central path
methods \cite{R88}.
Recently, there has been significant progress on solvers for general linear
programs, but the running time required to solve a linear program with
roughly $n$ variables and $\Otil(n)$ constraints
(assuming polynomially bounded entries and polytope radius) is stuck at
the $\Otil(n^{2.372\ldots})$, the running time provided by LP solvers
that run in current matrix multiplication time \cite{CLS21}.
Note that this running time is in the RealRAM model, and this
algorithm cannot be translated to fixed point arithmetic with
polylogarithmic bit complexity per number without additional
assumptions on the input, as we describe the paragraphs on numerical
stability below.
To compare these running times with those for single-commodity maximum
flow algorithms on a graph with $|V|$ vertices and $|E|$ edges,
observe that in a sparse graph with $|E| = \Otil(|V|)$, by
writing the maximum flow problem as a linear program, we can solve it
using general linear program solvers and obtain a running time of
$\Otil(|V|^{2.372\ldots})$, while the state-of-the art maximum flow
solver obtains a running time of $\Otil(|V|^{1.497})$ on such a sparse
graph.
On dense graphs with $|E| = \Theta(|V|^2)$, the gap is smaller but
still substantial: The running time is $\Otil(|V|^{2})$ using maximum
flow algorithms vs.  $\Otil(|V|^{2.372\ldots})$ using general LP
algorithms.

\paragraph{How hard is it solve multi-commodity flow?}
The many successes in developing high-accuracy algorithms for
single-commodity flow problems highlight an important open question:
Can multi-commodity flow be solved to high accuracy faster than
general linear programs?
We rule out this possibility, by proving that any linear program
(assuming it is polynomially bounded and has integer entries)
can be encoded as a multi-commodity flow problem in
nearly-linear time.
This implies that any improvement
in the running time of (high-accuracy) algorithms for sparse
multi-commodity flow problems would directly translate to a faster algorithm
for solving sparse linear programs to high accuracy, with only a polylogarithmic 
increase in running time.

Previous work by Kyng and Zhang \cite{KZ20} had shown that fast
algorithms for multi-commodity flow were unlikely to arise from
combining interior point methods with special-purpose linear equation
solvers.
Concretely, they showed that the linear equations that arise in
interior point methods for multi-commodity flow are as hard to solve
as arbitrary linear equations.
This ruled out algorithms following the pattern of the known fast
algorithms for high-accuracy single-commodity flow problems.
However, it left open the broader question if some other family of
algorithms could succeed.
We now show that, in general, a separation between multi-commodity flow
and linear programming is not possible.

\subsection{Background: Numerical stability of linear program solvers and reductions}

Current research on fast algorithms for solving linear programs generally
relies on assuming bounds on (1) the size of the program entries and (2)
the norm of all feasible solutions.
Generally, algorithm running time depends logarithmically on these
quantities, and hence to make these factors negligible, entry size and
feasible solution norms are assumed to be polynomially bounded, for
example in \cite{CLS21}.
We will refer to a linear program satisfying these assumptions as
\emph{polynomially bounded}.
More precisely, we say a linear program
    with $N$ non-zero coefficients is \emph{polynomially bounded} if it has coefficients in the
    range $[-X,X]$ and $\norm{\xx}_1 \le R$ for all feasible $\xx$ 
(i.e. the polytope of feasible solutions has radius of $R$ in $\ell_1$
norm), and $X,R \leq O(N^c)$ for some constant $c$.
In fact, if there exists a feasible solution $\xx$ satisfying $\norm{\xx}_1 \le R$, 
then we can add a constraint $\norm{\xx}_1 \le R$ to the LP (which can
be rewritten as linear inequality constraints)
so that in the new LP, all feasible solutions have $\ell_1$ norm at most $R$.
This only increases the number of nonzeros in the LP by at most a constant factor.

\paragraph{Interior point methods and reductions with fixed point arithmetic.}
Modern fast interior point methods for linear programming, such as \cite{CLS21}, are 
analyzed in the RealRAM model.
In order to implement these algorithms using \emph{fixed point arithmetic} with
polylogarithmic bit complexity per number, instead of RealRAM, additional assumptions are
required. For example, this class of algorithms relies on computing matrix
inverses, and these must be approximately representable using
polylogarithmic bit complexity per entry. This is not possible, if the
inverses have exponentially large entries, which may occur even in
polynomially bounded linear programs.
For example, consider a linear program feasibility problem $\AA \xx \leq
\bb, \xx \geq \veczero$, with constraint matrix $\AA \in \R^{2n \times 2n} $ given by 
\[
  \AA(i,j) =
  \begin{cases}
    1 & \text{ if } i  < n \text{ and } i = j  \\
    -2 & \text{ if } i < n \text{ and } i + 1 = j \\
     2 & \text{ if } i \geq n \text{ and } i = j  \\
    -1 & \text{ if } i \geq n \text{ and } i + 1 = j \\
    0  & \text{o.w.}
  \end{cases}
\]
Such a linear program is polynomially bounded for many choices of
$\bb$, e.g. $\bb = \ee_{2n}$.
Unfortunately, for the vector $\xx \in \R^{2n}$ given by
\[
  \xx(i) =
  \begin{cases}
    2^{-i-1} & \text{ if } i  \leq n \\
       0  & \text{o.w.}
  \end{cases}
\]
we have $\AA\xx = 2^{-n} \ee_n$, and from this one can see that $\AA^{-1}$ must have entries of size
at least $\Omega(2^n/n)$. This will cause algorithms such as \cite{CLS21} to
perform intermediate calculations with $n$ bits per number, increasing
the running time by a factor roughly $n$.

Modern interior point methods can be translated to fixed precision
arithmetic with various different assumptions leading to
different per entry bit complexity (see \cite{CLS21} for a discussion
of one standard sufficient condition).
Furthermore, if the problem has polynomially bounded condition number
(when appropriately defined), then we expect that polylogarithmic bit
complexity per entry should suffice, at least for highly accurate
approximate solutions, by relying on fast stable
numerical linear algebra \cite{DDH07}, although we are not aware of a
complete analysis of this translation.

If a linear program with integer entries is solved to sufficiently small
additive error, the approximate solution can be converted into an
exact solution, e.g. see \cite{R88, LS14, CLS21} for a discussion of the necessary
precision and for a further discussion of numerical stability
properties of interior point methods.
Polynomially bounded linear programs with integer coefficients may
still require exponentially small additive error for this rounding to
succeed.

We give a reduction from general linear programming to 2-commodity
flow, and like \cite{CLS21}, we assume the program is polynomially
bounded to carry out the reduction.
We also use an additional assumption, namely that the linear program is
written using integral entries\footnote{W.l.o.g., by scaling, this is the same as
  assuming the program is written with polynomially bounded fixed
  precision numbers of the form
  $k/D$ where $k$ is an integer, and $D$ is an integral denominator shared
  across all entries, and both $k$ and $D$ are polynomially bounded.}.
We do not make additional assumptions about polynomially bounded
condition number of problem. This means we can apply our reduction to programs such as the one
above, despite \cite{CLS21} not obtaining a reasonable running time on
such programs using fixed point arithmetic.

Our analysis of our reduction uses the RealRAM model like
\cite{CLS21} and other modern interior point method analysis, however,
it should be straightforward to translate our reduction and error
analysis to fixed point arithmetic with polylogarithmically many bits,
because all our mappings are simple linear transformations, and we
never need to compute or apply a matrix inverse.

\paragraph{Rounding linear programs to have integer entries.}
It is possible to give some fairly general and natural sufficient conditions
for when a polynomially bounded linear program can be rounded to have
integral entries, one example of this is having a polynomially bounded
\emph{Renegar's condition number}.
 Renegar introduced this condition number for linear programs in
 \cite{R95}.
For a given linear program, suppose that perturbing the entries of the
program by at most $\delta$ each does not change the feasibility of the 
the linear program, and let $\delta^*$ be the largest such $\delta$.
Let $U$ denote the maximum absolute value of entries in the linear program.
Then $\kappa = U/\delta^*$ is Renegar's condition number for the linear program.

Suppose we are given a polynomially bounded linear program $\max\{\cc^\top \xx: \AA \xx \le \bb, \xx \ge \bf{0} \}$ (also referred to as $(\AA,\bb,\cc)$),
with polytope radius at most $R$, and Renegar's condition number
$\kappa$ also bounded by a polynomial.
We wish to compute a vector $\xx \ge \bf{0}$ with an $\eps$ additive error on each constraint and in the optimal value.
We can reduce this problem for instance $(\AA,\bb,\cc)$ to a polynomially bounded linear program instance
with integral input numbers. 
Specifically, we round the entries of $\AA$ down to $\AAtil$ and those of $\bb, \cc$ up to $\bbtil, \cctil$
all by at most $\min\{\frac{\eps}{3R}, \frac{U}{\kappa R}\}$
such that each entry of $\AAtil, \bbtil, \cctil$ only needs a logarithmic number of bits.
Suppose $\xxtil$ is a solution to $(\AAtil, \bbtil, \cctil)$ with $\frac{\eps}{3}$ additive error on each 
constraint and in the optimal value. 
Then, 
\begin{align*}
  & \AA \xxtil = \AAtil \xxtil + (\AA - \AAtil) \xxtil
\le  \bb +\eps \bf{1}  \\
& \cc^\top \xxtil \ge \cctil^\top \xxtil^* - \frac{2\eps}{3}
\end{align*}
where $\bf{1}$ is the all-one vector, $\xxtil^*$ is an optimal solution to $(\AAtil, \bbtil, \cctil)$.
In addition, the optimal value of $(\AAtil, \bbtil, \cctil)$ is greater than or equal to 
that of $(\AA,\bb,\cc)$.
So, $\xxtil$ is a solution to $(\AA,\bb,\cc)$ with $\eps$ additive
error as desired.
Since each entry of $\AAtil,\bbtil,\cctil$ has a logarithmic number of bits, we can scale all of them 
to polynomially bounded integers without changing the feasible set and the optimal solutions.
Thus we see that if we restrict ourselves to polynomially linear programs with
polynomially bounded Renegar's condition number, and we wish to solve
the program with small additive error, we can assume without loss of
generality that the program has integer coefficients.

\subsection{Previous work}
Our paper follows the proof by Itai \cite{I78} that linear programming
is polynomial-time reducible to 2-commodity flow.
However, it is also inspired by recent works on hardness for
structured linear equations \cite{KZ20} and packing/covering LPs
\cite{KWZ20}, which focused on obtaining nearly-linear time reductions
in somewhat related settings.
These works in turn were motivated by the last decade's substantial
progress on fine-grained complexity for a range of polynomial time
solvable problems, e.g. see \cite{WW18}.
Also notable is the result by Musco et al. \cite{MNSUW19} on hardness
for matrix spectrum approximation.

\subsection{Our contributions}
In this paper, we explore the hardness of 2-commodity maximum
throughput flow, which for brevity we refer to as the 2-commodity flow
problem or 2CF.
We relate the difficulty of 2CF to that of linear
programming (LP) by developing an extremely efficient reduction from
the former to the latter.
The main properties of our reduction are described by the informal
theorem statement below.
We give a formal statement of Theorem~\ref{thm:MainInformal} as
Theorem~\ref{thm: main} in Section~\ref{sect: main_results}, and we
state the proof in Section~\ref{sect: main}.

\begin{theorem}[Main Theorem (Informal)]
  \label{thm:MainInformal}
  Consider any polynomially bounded linear program $\max\{\cc^\top \xx: \AA \xx \le \bb, \xx \ge
  \bf{0} \}$ with integer coefficients and $N$ non-zero entries.
In nearly-linear time, we can convert this linear program to a
2-commodity flow problem which is feasible if and only if the original
program is.
The 2-commodity flow problem has $\Otil(N)$ edges and has polynomially
bounded integral edge capacities.
Furthermore, any solution to the 2-commodity flow instance
with at most
$\epsilon$ additive error on each constraint and value at most $\epsilon$  from the optimum
can be
converted to a solution to the original linear program with
additive error $\Otil(\poly (N)\eps)$ on each constraint and
similarly value within $\Otil(\poly (N)\eps)$ of the optimum.

This implies that, for any constant $a > 1$,
if any 2-commodity flow instance with polynomially bounded integer capacities
can be solved with $\epsilon$ additive error in time $\Otil\left(|E|^a\cdot
  \poly\log(1/\epsilon)\right)$, then any polynomially bounded linear
program can be solved to $\epsilon$ additive error in time $\Otil\left(N^a\cdot \poly\log(1/\epsilon)\right)$.
\
\end{theorem}

Note that in our definition any 2CF problem is already an LP, and so
no reduction in the other direction is necessary.
Our notion of approximate solutions to LPs and 2CF problems is also
such that treating a 2CF problem as an LP and solving it approximately
ensures that the 2CF is approximately solved w.r.t. to the our
approximate solution definition for 2CF.

We obtain Theorem~\ref{thm:MainInformal}, our main result, by making several improvements to Itai's reduction from LP to 2CF.
Firstly, while Itai produced a 2CF with the number of edges on
the order of $\Theta\left(N^2\log^2 X\right)$, we show that an improved
gadget can reduce this to $O\left(N\log X\right)$.
Thus, in the case of polynomially bounded linear programs, where
$\log X = O(\log N)$, we get an only polylogarithmic multiplicative increase in the
number of non-zero entries from
$N$ to $\Otil(N)$, whereas Itai had an increase in the number of
nonzeros by a factor $\Otil(N)$, from $N$ to $\Otil(N^2)$.

Secondly, Itai used very large graph edge capacities that
require $O\left(\left(N\log X\right)^{1.01}\right)$ many bits
\emph{per edge}, letting the capacities grow exponentially given an LP with
polynomially bounded entries.
We show that when the feasible polytope radius $R$ is bounded, we can
ensure capacities remain a polynomial function of the initial
parameters $N, R$, and $X$.
In the important case of polynomially bounded linear programs, this
means the capacities stay polynomially bounded.

Thirdly, while Itai only analyzed the chain of reductions under the case
with exact solutions, we generalize the analysis to the case with
approximate solutions by establishing an error analysis along the
chain of reductions. We show that the error only grows polynomially during the
reduction. 
Moreover, to simplify our error analysis, we observe that additional
structures can be established in many of Itai's reductions. For
instance, we propose the notion of a \textit{fixed flow network}, which consists of a subset of edges with equal lower and upper bound of capacity. It is a simplification of Itai's $(l,u)$ network with general capacity (both lower and upper bounds on the amount of flow).

\paragraph{Open problems.}
Our reductions do not suffice to prove that a strongly polynomial time algorithm
for 2-commodity flow would imply a strongly polynomial time algorithm
for linear programming.
In a similar vein, it is unclear if a more efficient reduction could
exist for the case of linear programs that are not polynomially
bounded.
We leave these as very interesting open problems.

Finally, our reductions do not preserve the ``shape'' of the linear
program, in particular, a dense linear program may be reduced to a
sparse 2-commodity flow problem with a similar number of edges as
there are non-zero entries in the original program.
It would be interesting to convert a dense linear program into a dense
2-commodity flow problem, e.g. to convert a linear program with $m$
constraints and $n$ variables (say, $m \leq n$) into a 2-commodity flow problem with
$\Otil(n)$ edges and $\Otil(m)$ vertices.

\subsection{Organization of the remaining paper}
In Section \ref{sect: preliminaries}, we introduce some general notation and definitions for the problems 
involved in the reduction from LP to 2CF.
In Section \ref{sect: main_results}, we state our main theorem and give a brief overview of the proof.
In Section \ref{sect: proof}, we provide proof details for each step along the chain. In each step, we describe a (nearly-)linear-time algorithm that reduces problem A to problem B, 
and a linear-time algorithm that maps a solution to B to a solution to A. In addition, we prove that the size of B is nearly-linear in that of A, and an approximate (or exact) solution to B can be mapped back to an approximate (or exact) solution to A with a polynomial blow-up in error parameters.
In Section \ref{sect: main}, we prove the main theorem by putting all intermediate steps together.


\section{Preliminaries}
\label{sect: preliminaries}

\subsection{Notation}
\paragraph{Matrices and vectors}
We use parentheses to denote entries of a matrix or a vector: 
Let $\AA(i,j)$ denote the $(i,j)$th entry of a matrix $\AA$, and let $\xx(i)$ denote the $i$th entry of a vector $\xx$. 
Given a matrix $\AA\in\R^{m\times n}$, we use $\aa_i^\top$ to denote the $i$th row of a matrix $\AA$ and $\nnz(\AA)$ to denote the number of nonzero entries of $\AA$. 
Without loss of generality, we assume that $\nnz(\AA)\geq \max\{m,n\}$. 
For any vector $\xx\in\R^n$, we define $\norm{\xx}_{\max}=\max_{i\in [n]} \abs{\xx(i)}$, $\norm{\xx}_1=\sum_{i\in[n]}\abs{\xx(i)}$.
For any matrix $\AA\in\R^{m\times n}$, we define $\norm{\AA}_{\max}=\max_{i,j}\abs{\AA(i,j)}$. 

We define a function $X$ that takes an arbitrary number of matrices
$\AA_1, \ldots, \AA_{k_1}$, vectors $\bb_1, \ldots, \bb_{k_2}$, 
and scalars $K_1, \ldots, K_{k_3}$
as arguments, and returns the maximum of $\norm{\cdot}_{\max}$ of all the arguments, i.e.,
\begin{multline*}
	X(\AA_1, \ldots, \AA_{k_1}, \bb_1, \ldots, \bb_{k_2}, K_1, \ldots, K_{k_3}) \\
	= \max \left\{ \norm{\AA_1}_{\max}, \ldots, \norm{\AA_{k_1}}_{\max}, 
    \norm{\bb_1}_{\max}, \ldots, \norm{\bb_{k_2}}_{\max}, 
    \abs{K_1}, \ldots, \abs{K_{k_3}} \right\}.
\end{multline*}

\subsection{Problem Definitions}

In this section, we formally define the problems that we use in the reduction. 
These problems fall into two categories: one is related to linear programming and linear equations, 
and the other is related to flow problems in graphs.
In addition, we define the errors for approximately solving these problems. 

\subsubsection{Linear Programming and Linear Equations with Positive Variables}

For the convenience of our reduction, 
we define linear programming as a ``decision" problem. 
We can solve the optimization problem 
$\max\{\cc^\top \xx: \AA \xx \le \bb, \xx \ge \bf{0} \}$
by binary searching its optimal value via the decision problem.

\begin{definition}[Linear programming (\textsc{lp})]
	\label{def: lp}
	Given a matrix $\AA \in \mathbb{Z}^{m \times n}$, vectors $\bb \in \mathbb{Z}^{m}$ and $\cc \in \mathbb{Z}^n$, 
     an integer $K$, and $R \geq \max\{1, \max\{\norm{\xx}_1: \AA \xx \le \bb, \xx \ge \bf{0} \}\}$, 
	 we refer to the \textsc{lp} problem for $(\AA,\bb,\cc,K,R)$ as the problem
	of finding a vector $\xx \in \mathbb{R}^n_{\ge 0}$ 
    satisfying
    $$\AA\xx\leq\bb \text{ and } \cc^\top\xx\geq K$$
	if such an $\xx$ exists and returning ``infeasible" otherwise.
\end{definition}

We will reduce linear programming to linear equations with nonnegative variables (LEN), 
and then to linear equations with nonnegative variables and small integer coefficients ($k$-LEN).

\begin{definition}[Linear Equations with Nonnegative Variables (\textsc{len})]
	\label{def: len}
	Given $\AA \in \mathbb{Z}^{m \times n},\bb \in \mathbb{Z}^{m}$,
	and $R \geq \max\{1, \max\{\norm{\xx}_1: \AA \xx = \bb, \xx \ge \bf{0} \} \}$,
	we refer to the \textsc{len} problem for $(\AA,\bb,R)$ as the problem of finding 
	a  vector $\xx \in \mathbb{R}_{\ge 0}^n$ 
    satisfying $\AA\xx=\bb$ if such an $\xx$ exists 
    and returning ``infeasible'' otherwise.
\end{definition}

\begin{definition}[$k$-LEN ($k$-\textsc{len})]
	\label{def: k-len}
	The $k$-\textsc{len} problem is an \textsc{len} problem $(\AA, \bb, R)$ 
	where the entries of $\AA$ are 
     integers in $[-k,k]$ for some given $k \in \mathbb{Z}_+$.
\end{definition}

We employ the following additive error notion. 
We append a letter ``A'' to each problem name to denote its approximation version, e.g., 
LP Approximate Problem is abbreviated to LPA.

\begin{definition}[Approximation Errors]
	We always require $\xx \ge \bf{0}$. In addition,
	\begin{enumerate}
		\item \textit{Error in objective}: $\cc^\top\xx\geq K$ is relaxed to $\cc^\top\xx\geq K-\eps$;
		\item \textit{Error in constraint}:
		\begin{itemize}
			\item The inequality constraint $\AA\xx\leq \bb$ is relaxed to $\AA\xx-\bb\leq\eps\vecone$, where $\vecone$ is the all-1 vector;
			\item The equality constraint $\AA\xx=\bb$ is relaxed to $\norm{\AA\xx-\bb}_{\infty}\leq\eps$.
		\end{itemize}
	\end{enumerate}		
\end{definition}

\subsubsection{Flow Problems}

A \emph{flow network} is a directed graph $G = (V,E)$, where $V$ is the 
set of vertices and $E \subset V \times V$ is the set of edges,
together with a vector of edge capacities $\uu \in \mathbb{Z}_{> 0}^{\abs{E}}$
that upper bound the amount of flow passing each edge.
A \emph{2-commodity flow network} is a flow network together with two 
source-sink pairs $s_i, t_i \in V$ for each commodity $i \in \{1,2\}$.

Given a 2-commodity flow network $(G = (V,E), \uu, s_1, t_1, s_2, t_2)$, 
a \emph{feasible 2-commodity flow} is a pair of flows $\ff_1, \ff_2 \in \mathbb{R}^{\abs{E}}_{\ge 0}$
that satisfies 
\begin{enumerate}
    \item capacity constraint: $\ff_1(e)+\ff_2(e)\leq \uu(e), ~\forall e\in E$, and
    \item conservation of flows: $\sum_{u: (u,v) \in E} \ff_i(u,v) = \sum_{w: (v,w) \in E} \ff_i(v,w), ~ \forall i \in \{1,2\}, v \in V \setminus \{s_i,t_i \}$\footnote{Note that for commodity $i$, this constraint includes the case of $v\in\{s_{\bar{i}}, t_{\bar{i}}\}, \bar{i}=\{1,2\}\backslash i$.}.
\end{enumerate}

Similar to the definition of LP, we define 2-commodity flow problem 
as a decision problem. We can solve a decision problem 
by solving the corresponding optimization problem.

\begin{definition}[2-Commodity Flow Problem (\textsc{2cf})]
	\label{def: 2cf}
Given a 2-commodity flow network $(G, \uu, s_1, t_1, s_2, t_2)$
together with $R \ge 0$,
we refer to the \textsc{2cf} problem for $(G, \uu, s_1, t_1, s_2, t_2, R)$
as the problem of finding a feasible 2-commodity flow $\ff_1, \ff_2\geq \bf 0$
satisfying $$F_1 + F_2 \ge R$$ 
if such flows exist 
and returning ``infeasible'' otherwise.
\end{definition}

To reduce \textsc{lp} to \textsc{2cf}, we need a sequence of variants of flow problems.

\begin{definition}[2-Commodity Flow with Required Flow Amount (\textsc{2cfr})]
	\label{def: 2cfr}
    Given a 2-commodity flow network $(G, \uu, s_1,t_1, s_2,t_2)$
    together with $R_1, R_2 \ge 0$, we refer to the \textsc{2cfr} for 
    $(G, \uu, s_1,t_1, s_2,t_2,\allowbreak R_1, R_2)$ as the problem of finding a feasible 
    2-commodity flow $\ff_1, \ff_2\geq \bf 0$ satisfying 
    $$F_1 \ge R_1, ~~F_2 \ge R_2$$
    if such flows exist and returning ``infeasible'' otherwise.
\end{definition}

\begin{definition}[Fixed flow constraints]
  Given a set $F \subseteq E$ in a 2-commodity flow network, we say the
  flows $\ff_1, \ff_2\geq \bf 0$ satisfy \emph{fixed flow constraints on $F$} if
  \[
  \ff_1(e) + \ff_2(e) = \uu(e), ~ \forall e \in F.
  \]
  Similarly, given a set $F \subseteq E$ in a 1-commodity flow
  network, we say the  flow $\ff\geq \bf 0$ satisfies \emph{fixed flow
  	constraints on $F$} if
  \[
  \ff(e) = \uu(e), ~ \forall e \in F.
  \]
\end{definition}

\begin{definition}[2-Commodity Fixed Flow Problem (\textsc{2cff})]
	\label{def: 2cff}
    Given a 2-commodity flow network $(G , \uu, s_1,t_1, s_2,t_2)$ together with
    a subset of edges $F \subseteq E$,
    we refer to the \textsc{2cff} problem for the
    tuple $(G, F, \uu, s_1,t_1, s_2,t_2)$ as the problem of finding a
    feasible 2-commodity flow $\ff_1, \ff_2 \ge \bf{0}$ which also
    satisfies the fixed flow constraints on $F$ if such flows exist
    and returning ``infeasible'' otherwise.
\end{definition}

\begin{definition}[Selective Fixed Flow Problem (\textsc{sff})]
	\label{def: sff}
	Given a 2-commodity network $(G, \uu, s_1,t_1,
        s_2,t_2)$ together with
        three edge sets $F, S_1, S_2  \subseteq E$,
	we refer to the \textsc{sff} problem for $(G, F, S_1,S_2, \uu, s_1,t_1, s_2,t_2)$
as the problem of finding a feasible 2-commodity flow $\ff_1, \ff_2 \ge \bf{0}$ such that 
for each $i \in \{1,2\}$, flow
$\ff_i(e) > 0$ only if $e \in S_i$,
and $\ff_1, \ff_2$ satisfy the fixed flow constraints on $F$, if such flows exist,
and returning ``infeasible'' otherwise.
\end{definition}

\begin{definition}[Fixed Homologous Flow Problem (\textsc{fhf})]
	\label{def: fhf}
Given a flow network with a single source-sink pair $(G, \uu, s, t)$
together with a collection of disjoint subsets of edges $\calH = \{H_1, \ldots, H_h\}$ and a subset of edges $F\subseteq E$ such that $F$ is disjoint from all the sets in $\calH$,
we refer to the \textsc{fhf} problem
for $(G, F,\calH,  \uu, s, t)$ as the problem of finding a feasible flow $\ff \ge \bf{0}$
such that 
\[
\ff(e_1) = 	\ff(e_2), ~~ \forall e_1, e_2 \in H_k, 1\le k\le h,  
\]
and $\ff$ satisfies the fixed flow constraints on $F$, if such flows exist,
and returning ``infeasible'' otherwise.
\end{definition}

\begin{definition}[Fixed Pair Homologous Flow Problem (\textsc{fphf})]
	\label{def: fphf}
An \textsc{fphf} is an \textsc{fhf} problem $(G, F, \calH, \uu, s, t)$
where every set in $\calH$ has size 2.
\end{definition}

Now, we define errors for the above flow problems. 

\begin{definition}[Approximation Errors]
We always require flows $\ff \ge \bf 0$ and $\ff_1, \ff_2 \ge \bf 0$. In addition,
\begin{enumerate}
	\item \textit{Error in congestion}: the capacity constraints are relaxed to: 
	\begin{equation*}
		\label{eq:error_in_congestion}
		\ff_1(e)+\ff_2(e)\leq \uu(e)+\eps, ~~\forall e\in E.
	\end{equation*}
	There are several variants corresponding to different flow problems.
	\begin{itemize}
		\item If $e\in F$ is a fixed-flow edge, the fixed-flow constraints are relaxed to
		\[\uu(e)-\eps\leq \ff_1(e)+\ff_2(e)\leq \uu(e)+\eps, ~~\forall e\in F\]
		\item If $G$ is a 1-commodity flow network, we replace $\ff_1(e)+\ff_2(e)$ by $\ff(e)$.
	\end{itemize}
	\item \textit{Error in demand}: the conservation of flows is relaxed to
	\begin{align}
		& \abs{\sum_{\begin{subarray}{c}
					u:
					(u,v)\in E
			\end{subarray}}\ff_i(u,v)-\sum_{\begin{subarray}{c}
					w:
					(v,w)\in E
			\end{subarray}}\ff_i(v,w)}\leq \epsilon, ~ \forall v\in V\backslash\{s_i,t_i\}, i\in\{1,2\}
			\label{eq:error_in_demand_1}
	\end{align}
	There are several variants of this constraint corresponding to different flow problems.
	\begin{itemize}
		\item If the problem is with flow requirement $F_i$, then besides Eq. \eqref{eq:error_in_demand_1}, we add demand constraints for $s_i$ and $t_i$ with respect to commodity $i$:
		\begin{align}
			& \abs{\sum_{w: (s_i,w)\in E} \ff_i(s_i,w)-F_i}\leq \eps, \qquad \abs{\sum_{u: (u,t_i)\in E} \ff_i(u,t_i)-F_i}\leq \eps, ~ i\in\{1,2\}
			\label{eqn:error_demand_s_t}
		\end{align}
		\item If $G$ is a 1-commodity flow network, Eq. \eqref{eq:error_in_demand_1} can be simplified as
		\begin{align*}
			& \abs{\sum_{\begin{subarray}{c}
						u:
						(u,v)\in E
				\end{subarray}}\ff(u,v)-\sum_{\begin{subarray}{c}
						w:
						(v,w)\in E
				\end{subarray}}\ff(v,w)}\leq \epsilon, ~ \forall v\in V\backslash\{s,t\}
		\end{align*}
		
	\end{itemize}
	\item \textit{Error in type}: the selective constraints are relaxed to \[\ff_{\bar{i}}(e) \leq \eps, ~~ \forall e\in S_{i},~ \bar{i}=\{1,2\}\backslash i.\]
	\item \textit{Error in (pair) homology}: the (pair) homologous constraints are relaxed to \[\abs{\ff(e_1)-\ff(e_2)}\leq \epsilon, ~ \forall e_1, e_2\in H_k, H_k\in\mathcal{H}.\]
\end{enumerate}
\end{definition}



\section{Main results}
\label{sect: main_results}

\begin{theorem}
	\label{thm: main}
	Given an \textsc{lpa} instance $(\AA,\bb,\cc,K,R,\epsilon^{lp})$ where $\AA\in\Z^{m\times n}, \bb\in Z^m, \cc\in \Z^n, K\in\Z, \epsilon^{lp} \ge 0$ and $\AA$ has $\nnz{(\AA)}$ nonzero entries, we can reduce it to a \textsc{2cfa} instance $(G=(V,E),\uu, s_1, t_1, s_2, t_2, R^{2cf}, \epsilon^{2cf})$ in time $O(\nnz(\AA)\log X)$ where $X = X(\AA,\bb,\cc,K)$, such that 
	\begin{align*}
		& |V|, |E|=O(\nnz(\AA)\log X), \\
		& \norm{\uu}_{\max}, R^{2cf}=O(\nnz^3(\AA) RX^{2} \log^2 X), \\
		& \epsilon^{2cf}=\Omega\left(\frac{1}{\nnz^{7}(\AA)R X^{3} \log^6 X }\right)\epsilon^{lp}.
	\end{align*} 
	If the \textsc{lp} instance $(\AA,\bb,\cc,K,R)$ has a solution, then the \textsc{2cf} instance $(G^{2cf},\uu^{2cf}, s_1, t_1, s_2, t_2, R^{2cf})$ has a solution.
	Furthermore, if $\ff^{2cf}$ is a solution to the \textsc{2cfa}
        (\textsc{2cf}) instance, then in time $O(\nnz(\AA)\log X)$,
        we can compute a solution $\xx$ to the \textsc{lpa} (\textsc{lp}, respectively) instance, where the exact case holds when $\eps^{2cf}=\eps^{lp}=0$.
\end{theorem}

Our main theorem immediately implies the following corollary.

\begin{corollary}
	\label{col}
	If we can solve any \textsc{2cfa} instance $(G = (V,E), \uu, s_1,t_1,s_2,t_2, R^{2cf}, \epsilon)$ 
	in time $O\left(\abs{E}^c \poly\log\left( \frac{\norm{\uu}_1}{\epsilon}\right)\right)$ for some small constant $c\geq 1$, 
	then we can solve any \textsc{lpa} instance $(\AA,\bb,\cc,K,R,\epsilon)$ 
	in time $O\left(\nnz^c(\AA) \poly\log\left(\frac{\nnz(\AA)RX(\AA,\bb,\cc,K)}{\epsilon}\right)\right)$.
\end{corollary}

\subsection{Overview of our proof}

In this section, we will explain how to 
reduce an LP instance to a 2-commodity flow (\textsc{2cf}) instance 
by a chain of efficient reductions.
In each step, we reduce a decision problem A to a decision problem B.
We guarantee that
(1) the reduction runs in nearly-linear time\footnote{Linear in the size of problem $A$,
poly-logarithmic in the maximum magnitude of all the numbers that describe $A$,
the feasible set radius, and the inverse of the error parameter if 
an approximate solution is allowed.},
(2) the size of B is nearly-linear in the size of A, and 
(3) that A is feasible implies that B is feasible, and an approximate solution to B 
can be turned to an approximate solution to A with only a polynomial blow-up in 
error parameters, in linear time.

We follow the outline of Itai's reduction~\cite{I78}.
A summary of the problem notation used in the reduction from \textsc{lp}  to \textsc{2cf} is given in Table \ref{tbl: notations}.
Itai first reduced an \textsc{lp} instance to a 1-\textsc{len} instance (linear equations 
with nonnegative variables and $\pm 1$ coefficients).
A 1-\textsc{len} instance can be case as a single-commodity flow problem
subject to additional homologous constraints and fixed flow constraints (i.e., \textsc{fhf}).
Then, Itai dropped these additional constraints step by step,
via introducing a second commodity of flow and 
imposing lower bound requirements on the total amount of flows routed 
between the source-sink pairs.
However, in the worst case, Itai's reduction from 1-\textsc{len} to \textsc{fhf} enlarges the 
problem size quadratically and is thus inefficient. 
One of our main contributions is to improve this step so that 
the problem representation size is preserved along the reduction chain.

Our second main contribution is an upper bound on the errors 
accumulated during the process of mapping an approximate solution to the \textsc{2cf}
instance to an approximate solution to the \textsc{lp} instance.
We show that the error only grows by polynomial factors.
Itai only considered exact solutions between these two instances.

We will explain the reductions based on the exact versions of the problems.
At the end of this section, we will discuss some intuitions of our error analysis.

\begin{table}[t]
	\renewcommand\arraystretch{1.4}
	\centering 
	\scriptsize
	\caption{A summary of notation used in the reduction from \textsc{lp} (\textsc{lpa}) to \textsc{2cf} (\textsc{2cfa}, respectively). The column ``Input'' and ``Output'' are shared for both exact and approximate problems. The column ``Error'' is only for approximate problems. }  
	\label{tbl: notations}  

	\begin{tabular}{|r|c|c|r|c|}  
		\hline  
		Exact problems & Input & Output & Approximate problems& Error \tablefootnote{Error parameters will be defined in Section~\ref{sect: proof}.} \\ \hline
		\textsc{lp} (Def.  \ref{def: lp})& $\AA,\bb, \cc,K, R$ & $\xx$ & \textsc{lpa} (Def.  \ref{def: lpa}) & $\epsilon^{lp}$\\ \hline
		\textsc{len} (Def.  \ref{def: len}) & $\AAtil, \tilde{\bb}, \tilde{R}$ & $\tilde{\xx}$ &	\textsc{lena} (Def.  \ref{def: lena-new})& $\epsilon^{le}$\\ \hline
		\textsc{2-len} (Def.  \ref{def: k-len}) & $\bar{\AA}, \bar{\bb}, \bar{R}$ & $\bar{\xx}$& \textsc{2-lena} (Def.  \ref{def: k-lena-new}) &$\epsilon^{2le}$\\ \hline
		\textsc{1-len} (Def.  \ref{def: k-len})& $\hat{\AA},\hat{\bb}, \hat{R}$ & $\hat{\xx}$ &\textsc{1-lena} (Def.  \ref{def: k-lena-new}) &$\epsilon^{1le}$\\ \hline
		\textsc{fhf} (Def.  \ref{def: fhf}) & $G^h, F^h, \mathcal{H}^h=\{H_1,\cdots, H_h\}, \uu^h, s,t$ & $\ff^h$ & \textsc{fhfa} (Def.  \ref{def: fhfa-new}) &$\epsilon^h$ \\ \hline
		\textsc{fphf} (Def.  \ref{def: fphf}) & $G^p, F^p, \mathcal{H}^p=\{H_1,\cdots, H_p\}, \uu^p, s, t$ & $\ff^p$ & \textsc{fphfa} (Def.  \ref{def: fphfa-new}) & $\epsilon^p$ \\ \hline
		\textsc{sff} (Def.  \ref{def: sff}) & $G^s, F^s,  S_1, S_2, \uu^s, s_1, t_1, s_2, t_2$ & $\ff^s$ & \textsc{sffa} (Def.  \ref{def: sffa-new}) &$\epsilon^s$\\ \hline
		\textsc{2cff} (Def.  \ref{def: 2cff}) & $G^f, F^f, \uu^f, s_1, t_1, s_2, t_2$ & $\ff^f$ & \textsc{2cffa} (Def.  \ref{def: 2cffa-new}) &$\epsilon^f$\\ \hline
		\textsc{2cfr} (Def.  \ref{def: 2cfr}) & $G^r, \uu^r, \bar{s}_1, \bar{t}_1, \bar{s}_2, \bar{t}_2, R_1, R_2$ & $\ff^r$ & \textsc{2cfra} (Def.  \ref{def: 2cfra-new}) &$\epsilon^r$\\ \hline
		\textsc{2cf} (Def.  \ref{def: 2cf}) & $G^{2cf}, \uu^{2cf}, \bar{\bar{s}}_1, \bar{t}_1, \bar{\bar{s}}_2, \bar{t}_2, R^{2cf}$ & $\ff^{2cf}$ & \textsc{2cfa} (Def.  \ref{def: 2cfa}) &$\epsilon^{2cf}$ \\ \hline
	\end{tabular}
\end{table}

\subsubsection{Reducing Linear Programming to Linear Equations with Nonnegative Variables and $\pm 1$ Coefficients}

Given an \textsc{lp} instance $(\AA, \bb, \cc, K, R)$ where 
$R \ge \max\{1, \max\{\norm{x}_1: \AA\xx \le \bb, \xx \ge \bf{0} \}\}$,
we want to compute a vector $\xx \ge \bf{0}$ satisfying 
\[
	\AA \xx \le \bb, \cc^\top \xx \ge K	
\]
or to correctly declare infeasible.
We introduce slack variables $\ss, \alpha \ge \bf{0}$
and turn the above inequalities to equalities:
\[
\AA\xx + \ss = \bb, \cc^\top \xx - \alpha = K	
\] 
which is an \textsc{len} instance $(\AAtil, \bbtil, \Rtil)$.
Comparing to Itai's proof, we need to track two additional parameters: the polytope radius $R$ and  
the maximum magnitude of the input entries $X$.

We then reduce the \textsc{len} instance to linear equations with $\pm 2$ coefficients (\textsc{2-len})
by bitwise decomposition. 
For each bit, we introduce a carry term.
Different from Itai's reduction, we impose an upper bound for each carry variable. 
We show that this upper bound does not change problem feasibility 
and it guarantees that the polytope radius only increases polynomially.
The following example demonstrates this process.
\[5x_1+3x_2-7x_3=-1\]
\[\Downarrow\]
\[(2^0+2^2)x_1+(2^0+2^1)x_2-(2^0+2^1+2^2)x_3=-2^0\]
\[\Downarrow\]
\[(x_1+x_2-x_3)2^0+(x_2-x_3)2^1+(x_1-x_3)2^2=-1\cdot2^0\]
It can be decomposed to 3 linear equations, together with carry terms $(c_i-d_i)$, where $c_i, d_i\geq 0$:
\begin{align*}
	x_1+x_2-x_3-2(c_0-d_0)&=-1\\
	x_2-x_3+(c_0-d_0)-2(c_1-d_1)&=0\\
	x_1-x_3+(c_1-d_1)&=0
\end{align*}

Next, we reduce the \textsc{2-len} instance to a \textsc{1-len} instance
by replacing each $\pm 2$ coefficient  variable 
with two new equal-valued variables.

All the above three reduction steps run in nearly-linear time, 
and the problem sizes increase nearly-linearly.

\subsubsection{Reducing Linear Equations with Nonnegative Variables and $\pm 1$ Coefficients to Fixed Homologous Flow Problem}

One of our main contributions is a linear-time reduction from 
\textsc{1-len} to \textsc{fhf} (single-commodity fixed homologous flow problem).
Our reduction is similar to Itai's reduction, but more efficient.

Itai observed that a single linear equation $\aa^\top \xx = b$ with $\pm 1$ coefficients 
can be represented as a fixed homologous flow network $G$.
The network $G$ has a source vertex $s$, a sink vertex $t$, 
and two additional vertices $J^+$ and $J^-$.
Each variable $\xx(i)$ with coefficient $\aa(i) \neq 0$ corresponds to an edge:
There is an edge from $s$ to $J^+$ if $\aa(i) = 1$,
and there is an edge from $s$ to $J^-$ if $\aa(i) = -1$.
The amount of flow passing this edge encodes the value of $\xx(i)$. 
Thus, the difference between the total amount of flow entering $J^+$
and that entering $J^-$ equals to $\aa^\top \xx$.
To force $\aa^\top \xx = b$, we add two edges $e_1, e_2$ from $J^+$ to $t$
and one edge $e_3$ from $J^-$ to $t$;
we require $e_1$ and $e_3$ to be homologous
and require $e_2$ to be a fixed flow with value $b$.

We can generalize this construction to encode a system of linear equations (see Figure \ref{fig: homo} in Section \ref{sect: 1lena-fhfa}).
Specifically, we create a gadget as above for each individual equation,
and then glue all the source (sink) vertices for all the equations together as the source (sink, respectively) 
of the graph. 
To guarantee the variable values to be consistent in these equations, 
we require the edges corresponding to the same variable $\xx(i)$ in different equations to be homologous.
We can check that the number of the vertices is linear in the number of equations;
the number of the edges and the total size of the homologous sets are both linear 
in the number of nonzero coefficients of the linear equation system.

\subsubsection{Dropping the Homologous and Fixed Flow Constraints}

To reduce \textsc{fhf} to \textsc{2cf} (2-commodity flow problem), 
we need to drop the homologous and fixed flow constraints.
The reduction has three main steps.

\paragraph{Reducing \textsc{FHF} to \textsc{SFF}.}

Given an \textsc{fhf} instance, we can 
reduce it to a fixed homologous flow instance in which each homologous edge set 
has size 2 (i.e., \textsc{fphf}).
To drop the homologous requirement in \textsc{fphf}, 
we introduce a second commodity of flow with source-sink pair $(s_2, t_2)$, 
and for each edge, we carefully select the type(s) of flow that can pass through this edge.
Specifically, given two homologous edges $(v,w)$ and $(y,z)$,
we construct a constant-sized gadget (see Figure \ref{fig: phomo-slu2cf} in Section \ref{sect: fphfa-sffa}):
We introduce new vertices $vw, vw', yz, yz'$,  
construct a directed path $P: s_2 \rightarrow vw \rightarrow
vw' \rightarrow yz \rightarrow yz' \rightarrow t_2$,
and add edges $(v,vw), (vw',w)$ and $(y,yz), (yz',z)$.
Now, there is a directed path $P_{vw}: v \rightarrow vw \rightarrow vw' \rightarrow w$
and a directed path $P_{yz}: y \rightarrow yz \rightarrow yz' \rightarrow z$.
Paths $P$ and $P_{vw}$ ($P_{yz}$) share an edge $e_{vw} = (vw,vw')$ ($e_{yz} = (yz,yz')$, respectively).
We select $e_{vw}$ and $e_{yz}$
for both flow $\ff_1$ and $\ff_2$, 
select the rest of the edges along $P$ for only $\ff_2$, 
and select the rest of the edges along $P_{vw}, P_{yz}$ for only $\ff_1$.
By this construction, in this gadget, we have $\ff_2(e_{vw})=\ff_2(e_{yz})$ being the amount of flow routed in $P$, $\ff_1(e_{vw})$ and $\ff_1(e_{yz})$ being the amount of flow routed in $P_{vw}$ and $P_{yz}$, respectively.
Next, we choose $e_{vw}$ and $e_{yz}$ to be fixed flow edges
with equal capacity; this guarantees the same amount of $\ff_1$ is routed through $P_{vw}$ and $P_{yz}$.
The new graph is an \textsc{sff} instance.

\paragraph{Reducing \textsc{SFF} to \textsc{2CFF}.}
Next, we will drop the selective requirement of the \textsc{sff} instance.
For each edge $(x,y)$ selected for flow $i$, 
we construct a constant-sized gadget (see Figure \ref{fig: slu2cf-lu2cf} in Section \ref{sect: sffa-2cffa}):
We introduce two vertices $xy, xy'$, construct a direct path $s_i \rightarrow xy' \rightarrow xy \rightarrow t_i$,
and add edge $(x, xy)$ and $(xy', y)$.
This gadget simulates a directed path from $x$ to $y$ for flow $\ff_i$, 
and guarantees no directed path from $x$ to $y$ for flow $\ff_{\bar{i}}$
so that $\ff_{\bar{i}}$ cannot be routed from $x$ to $y$.
We get a \textsc{2cff} instance.

\paragraph{Reducing \textsc{2CFF} to \textsc{2CF}.}
It remains to drop the fixed flow constraints. 
The gadget we will use is similar to that used in the last step.
We first introduce new sources $\bar{s}_1, \bar{s}_2$ and sinks $\bar{t}_1, \bar{t}_2$.
Then, for each edge $(x,y)$ with capacity $u$, we construct a constant-sized gadget (see Figure \ref{fig: lu2cf-2cfr} in Section \ref{sect: 2cffa-2cfra}).
We introduce two vertices $xy, xy'$, 
add edges $(\bar{s}_1, xy'), (\bar{s}_2, xy'), (xy, \bar{t}_1), (xy, \bar{t}_2), (xy', xy)$,
and $(x,xy), (xy',y)$. This simulates a directed path from $x$ to $y$ that both flow $\ff_1$ and $\ff_2$ can pass through.
We let $(xy', xy)$ have capacity $u$ if $(x,y)$ is a fixed flow edge
and $2u$ otherwise; we let all the other edges have capacity $u$. 
Assume all the edges incident to the sources and the sinks are saturated, 
then the total amount of flows routed from $x$ to $y$ in this gadget must be $u$ if $(x,y)$ is a fixed flow edge
and no larger than $u$ otherwise.
Moreover, since the original sources and sinks are no longer sources and sinks now, we have 
to satisfy the conservation of flows at these vertices. 
For each $i \in \{1,2\}$, we create a similar gadget involving $\bar{s}_i, \bar{t}_i$ 
to simulate a directed path from $t_i$ to $s_i$ (the original sink and source), 
and let the edges incident to $\bar{s}_i, \bar{t}_i$ have capacity $M$,
the sum of all the edge capacities in the \textsc{2cff} instance.
This gadget guarantees that assuming the edges incident to $\bar{s}_i$ and $\bar{t}_i$ are saturated, 
the amount of flow routed from $t_i$ to $s_i$ through this gadget can be any number at most $M$.
To force the above edge-saturation assumptions to hold,  we require the amount of flow $\ff_i$ routed from $\bar{s}_i$ to $\bar{t}_i$
to be no less than $2M$ for each $i \in \{1,2\}$.

Now, this instance is close to a \textsc{2cf} instance except that 
we require a lower bound for each flow value instead of a lower bound for the sum of two flow values.
To handle this, we introduce new sources $\bar{\bar{s}}_1, \bar{\bar{s}}_2$
and for each $i \in \{1,2\}$, we add an edge $(\bar{\bar{s}}_i, \bar{s}_i)$ with capacity $2M$,
the lower bound required for the value of $\ff_i$.

One can check that in each reduction step, 
the reduction time is nearly linear and 
the problem size increases nearly linearly.
In addition, given a solution to the \textsc{2cf} instance, 
one can construct a solution to the \textsc{lp} instance
in nearly linear time.  

\subsubsection{Computing an Approximate Linear Program Solution from an Approximate 2-commodity Flow Solution }

We also establish an error bound for mapping an approximate 
solution to \textsc{2cfa} to an approximate solution to \textsc{lpa}.
We will outline the intuition behind our error analysis for  flow
problems: Even though we only maintain one error parameter per problem, we keep track with multiple types of error separately, (e.g., error in congestion, demand, selective types, and homology), depending on the problem settings. And we set the error parameter to be the error notion with the largest value. 
Suppose we reduce problem A to problem B with a certain gadget, then
we map a solution to B back to a solution to A.
We observe that each error notion of A is an additive accumulation of multiple error notions of B. 
This is because we have to map the flows of B passing through a gadget
including multiple edges back to a flow of A passing through a single edge. 
Each time we remove an edge, various errors related to this edge and incident vertices get transferred to its neighbors. Thus, the total error accumulation by the solution mapping can be polynomially bounded by the number of edges.
So, the final error only increases polynomially.


\pagebreak
{\small
\tableofcontents
}


\section{Proof details}
\label{sect: proof}
The chain of reductions from \textsc{lp(a)} to \textsc{2cf(a)} consists of nine steps.
In each step, we analyze the reduction from some problem A to some problem B in both the exact case and the approximate case.
In the exact case, we start with describing a nearly-linear-time method of reducing A to B, and a nearly-linear-time method of mapping a solution of B back to a solution of A. Then, we prove the correctness of the reduction method in the forward direction, that is if A has a solution then B has a solution. 
In addition, we provide a fine-grained analysis of the size of B given the size of A.
 
In the approximate case, we formally define approximately solving
each problem in
the first place, specifying bounds on various different types of error.
We always use the same reduction method from problem A to problem B in the
exact and approximate cases.
Thus the conclusion that problem B has a feasible solution when
problem A has a feasible solution also applies in the approximate
case.

We also always use a solution map back for the approximate case that
agrees with the exact case when there is no error. 
Crucially, we conduct an error analysis. That is given an approximate solution to B and its error parameters (by abusing notation, we use $\epsilon^B$ to denote), we map it back to an approximate solution to A and measure its error $\tau^A$ with respect to $\epsilon^B$. In other words, we can reduce an approximate version of A with error parameters $\epsilon^A\geq \tau^A$ to an approximate version of B with error parameters $\epsilon^B$. 
Note that the correctness of the reduction method in the backward
direction for the exact case follows from the approximate case analysis by setting all error parameters to zero, which completes the proof of correctness.

\subsection{LP(A) to LEN(A)}
\label{sect: lpa-lena}
\subsubsection{LP to LEN}
We show the reduction from an \textsc{lp} instance $(\AA,\bb,\cc,K,R)$
to an \textsc{len} instance $(\AAtil,\bbtil,\Rtil)$.
The \textsc{lp} instance has the following form:
\begin{equation}
	\label{eq: lp}
	\begin{aligned}
		\cc^\top\xx\geq K\\
		\AA\xx\leq\bb\\
		\xx\geq\boldsymbol{0}\\
	\end{aligned}
\end{equation}
To reduce it to an \textsc{len} instance,
we introduce slack variables $\alpha$ and $\ss$:
\begin{equation}
	\label{eq: len}
	\begin{aligned}
		\begin{pmatrix}
			\cc^\top&\boldsymbol{0}&-1\\
			\AA&\II&\boldsymbol{0}
		\end{pmatrix}\begin{pmatrix}
			\xx\\
			\ss\\
			\alpha
		\end{pmatrix}&=\begin{pmatrix}
			K\\
			\bb
		\end{pmatrix}\\
		\begin{pmatrix}
			\xx\\
			\ss\\
			\alpha
		\end{pmatrix}&\geq\boldsymbol{0}
	\end{aligned}
\end{equation}
Setting
\[\tilde{\AA}=\begin{pmatrix}
	\cc^\top&\boldsymbol{0}&-1\\
	\AA&\II&\boldsymbol{0}
\end{pmatrix}, \qquad \tilde{\xx}=\begin{pmatrix}
	\xx\\
	\ss\\
	\alpha
\end{pmatrix},\qquad \tilde{\bb}=\begin{pmatrix}
	K\\
	\bb
\end{pmatrix},\]
we get an \textsc{len} instance $(\AAtil, \bbtil, \Rtil)$
where $\Rtil = \max\{1, \max\{ \norm{\xxtil}_1: \AAtil \xxtil = \bbtil, \xxtil \ge \bf{0} \}\}$.

If an \textsc{len} solver returns $\xxtil = (\xx^\top, \ss^\top, \alpha)^\top$ for the \textsc{len} instance $(\AAtil, \bbtil, \Rtil)$, 
then we return $\xx$ for the \textsc{lp} instance $(\AA, \bb, \cc, K,R)$;
if the \textsc{len} solver returns ``infeasible'' for the \textsc{len} instance, 
then we return ``infeasible'' for the \textsc{lp} instance.

\begin{lemma}[\textsc{lp} to \textsc{len}]
	\label{lm: lp-len}
	Given an \textsc{lp} instance $(\AA,\bb,\cc,K, R)$ where $\AA\in\Z^{m\times n}, \bb \in \mathbb{Z}^m, \cc \in \mathbb{Z}^n, K \in \mathbb{Z}$, we can construct, in $O(\nnz(\AA))$ time, an \textsc{len} instance $(\tilde{\AA},\tilde{\bb}, \tilde{R})$ 
	where $\tilde{\AA}\in \Z^{\tilde{m}\times\tilde{n}}, \bbtil \in \mathbb{Z}^{\tilde{m}}$ 
	such that 
	\[\tilde{n}= n+m+1,~~\tilde{m}=m+1,~~\nnz(\tilde{\AA})\leq 4\nnz(\AA),\]
	\[
	\tilde{R}= 5mRX(\AA,\bb,\cc,K),~~
        X(\tilde{\AA},\tilde{\bb})=X(\AA,\bb,\cc,K)\geq 1. \]
        If the \textsc{lp} instance has a solution, then the \textsc{len} instance
        has a solution.
\end{lemma}	
\begin{remark*}
	In the backward direction, if the \textsc{len} instance has a
        solution, then in time $O(n)$, we can map this solution back
        to a solution to the \textsc{lp} instance. This is a
        consquence of Lemma \ref{lm: lpa-lena} by setting $\eps^{lp}=\eps^{le}=0$.
	
	This solution mapping conclusion is also true for all the rest steps in the chain of reductions.
\end{remark*}	
\begin{proof}
	Based on the reduction method described above, if $\xx$ is a solution to the \textsc{lp} instance as shown in Eq. \eqref{eq: lp}, 
	we can derive a solution $\xxtil=(\xx^\top, \ss^\top, \alpha)^\top$ to the \textsc{len} instance as shown in Eq. \eqref{eq: len}, by setting
	\[\ss=\bb-\AA\xx, \qquad \alpha = \cc^\top\xx-K.\]
	Thus, if the \textsc{lp} instance has a solution, then the \textsc{len} instance
	has a solution.

	Given the size of the \textsc{lp} instance with $n$ variables, $m$ linear constraints, and $\nnz(\AA)$ nonzero entries, we observe the size of the reduced \textsc{len} instance as following:
	\begin{enumerate}
		\item $\tilde{n}$ variables, where
		$\tilde{n}=n+m+1$.
		\item $\tilde{m}$ linear constraints, where
		$\tilde{m}=m+1$.
		\item $\nnz(\tilde{\AA})$ nonzeros, where
		\begin{equation}
			\label{eq: len-R}
			\nnz(\tilde{\AA})=\nnz(\AA)+\nnz(\cc)+m+1\leq 4\nnz(\AA),
		\end{equation}
		where we use $\nnz(\AA)\geq m,n\geq 1$, and $\nnz(\cc)\leq n$.
		\item $\tilde{R} = \max\{1, \max\{\norm{\tilde{\xx}}_1: \AAtil \xxtil = \bbtil, \xxtil \ge \bf{0} \} \}$, the radius of polytope in $\ell_1$ norm.
		Our goal is to upper bound $\norm{\xxtil}_1$ for every feasible solution to the \textsc{len} instance.
		By definition and the triangle inequality, 
		$$\norm{\tilde{\xx}}_1 \le \norm{\xx}_1+\alpha+\norm{\ss}_1.$$
		Note $\norm{\xx}_1 \le R$, the polytope radius in the \textsc{lp} instance.
		In addition,
		\[\alpha=
		\cc^\top\xx-K\leq \norm{\cc}_{\max}\norm{\xx}_1+|K|\leq \norm{\cc}_{\max}R+|K|,\]
		\[\norm{\ss}_1= \norm{\AA\xx-\bb}_1\leq \norm{\AA\xx}_1+\norm{\bb}_1\leq m\norm{\AA}_{\max}\norm{\xx}_1+\norm{\bb}_{1}\leq m(\norm{\AA}_{\max}R+\norm{\bb}_{\max}).\]
		Therefore, we have
		\begin{align*}
			\norm{\tilde{\xx}}_1
			&\leq R+|K|+\norm{\cc}_{\max}R+m(\norm{\AA}_{\max}R+\norm{\bb}_{\max})\\
			&\leq mR(1+|K|+\norm{\cc}_{\max}+\norm{\AA}_{\max}+\norm{\bb}_{\max}) && \text{Because $R\geq 1$}\\
			&\leq 5mRX(\AA,\bb,\cc,K). 
		\end{align*}
		Hence, it suffices to set
		\[\tilde{R}=5mRX(\AA,\bb,\cc,K). \]
		\item $X(\tilde{\AA},\tilde{\bb})=X(\AA,\bb,\cc,K)$ because \[\norm{\tilde{\AA}}_{\max}=\max \left\{\norm{\AA}_{\max}, \norm{\cc}_{\max}, 1\right\}=\max \left\{\norm{\AA}_{\max}, \norm{\cc}_{\max}\right\},\]
		\[\norm{\tilde{\bb}}_{\max}=\max\left\{|K|, \norm{\bb}_{\max}\right\}.\]
	\end{enumerate}
	
	To estimate the reduction time, as it takes $O\left(\nnz(\tilde{\AA})\right)$ time to construct $\tilde{\AA}$, and $O\left(\nnz(\tilde{\bb})\right)$ time to construct $\tilde{\bb}$, thus the reduction takes time
	\begin{align*}
		O\left(\nnz(\AAtil)+\nnz(\bbtil)\right)&=O\left(\nnz(\AAtil)\right) && \text{Because $\nnz(\bbtil)\leq m\leq \nnz(\AAtil)$}\\
		&=O\left(\nnz(\AA)\right).  && \text{By Eq. \eqref{eq: len-R}}
	\end{align*}

\end{proof}

\subsubsection{LPA to LENA}
The above lemma shows the reduction between exactly solving an \textsc{lp} instance and exactly solving an \textsc{len} instance. 
Next, we generalize the case with exact solutions to the case that allows approximate solutions.
First of all, we give a definition of the approximate version of \textsc{lp} and \textsc{len}.

\begin{definition}[LP Approximate Problem (\textsc{lpa})]
	\label{def: lpa}
	An \textsc{lpa} instance is given by an \textsc{lp} instance
	$(\AA,\bb,\cc,K, R)$ and an error parameter $\epsilon \in [0,1]$,
	which we collect in a tuple $(\AA,\bb,\cc,K, R, \epsilon)$.
	We say an algorithm solves the \textsc{lpa} problem,
	if, given any \textsc{lpa} instance,
	it returns a vector $\xx \ge \bf{0}$ such that
	\begin{align*}
		\cc^\top\xx &\geq K-\epsilon\\
		\AA \xx &\le \bb + \epsilon \vecone                      
	\end{align*}
	where $\vecone$ is the all-1 vector, or it correctly declares that the
	associated \textsc{lp} instance is infeasible.                   
\end{definition}

\begin{remark*}
	Note that our definition of \textsc{lpa}  does not require the algorithm to provide a
	certificate of infeasibility -- but our notion of an \emph{algorithm}
	for \textsc{lpa} requires
	the algorithm never incorrectly asserts infeasibility.
	Also note that when the \textsc{lp} instance is infeasible, the
	algorithm is still allowed to return an approximately feasible
	solution, if it finds one.	
	We use the same approach to defining all our approximate decision problems.
\end{remark*}

\begin{definition}[LEN Approximate Problem (\textsc{lena})]
	\label{def: lena-new}
	An \textsc{lena} instance is given by an \textsc{len} instance $(\AA,\bb, R)$ as in Definition \ref{def: len} and an error parameter $\eps\in[0,1]$, which we collect in a tuple $(\AA,\bb,R,\epsilon)$. We say an algorithm solves the \textsc{lena} problem, if, given any \textsc{lena} instance, it returns a vector $\xx\geq \boldsymbol{0}$ such that
	\begin{align*}
		\abs{\AA\xx-\bb}\leq \eps \vecone,
	\end{align*}
	where $\abs{\cdot}$ is entrywise absolute value and $\vecone$ is the all-1 vector, or it correctly declares that the associated \textsc{len} instance is infeasible.
\end{definition}

We use the same reduction method in the exact case to reduce an \textsc{lpa} instance to an \textsc{lena} instance. Furthermore, if an \textsc{lena} solver returns $\xxtil = (\xx^\top, \ss^\top, \alpha)^\top$ for the \textsc{lena} instance $(\AAtil, \bbtil, \Rtil, \epsilon^{le})$, 
then we return $\xx$ for the \textsc{lpa} instance $(\AA, \bb, \cc, K,R, \epsilon^{lp})$;
if the \textsc{lena} solver returns ``infeasible'' for the \textsc{lena} instance, 
then we return ``infeasible'' for the \textsc{lpa} instance.

Hence, the conclusions in the exact case (Lemma \ref{lm: lp-len}) also apply here, including the reduction time, problem size, and that the \textsc{len} instance has a feasible solution when the \textsc{lp} instance has one. 
In the approximate case, it remains to show the solution mapping time, as well as how the problem error changes by mapping an approximate solution to the \textsc{lena} instance back to an approximate solution to the \textsc{lpa} instance.

\begin{lemma}[\textsc{lpa} to \textsc{lena}]
	\label{lm: lpa-lena}
	Given an \textsc{lpa} instance $(\AA,\bb,\cc,K, R, \epsilon^{lp})$ where $\AA\in\Z^{m\times n}, \bb \in \mathbb{Z}^m, \cc \in \mathbb{Z}^n, K \in \mathbb{Z}$, 
	we can reduce it to an \textsc{lena} instance $(\tilde{\AA},\tilde{\bb}, \tilde{R}, \epsilon^{le})$ by letting 
	\[\eps^{le} = \eps^{lp},\] 
	and using Lemma \ref{lm: lp-len} to construct an \textsc{len} instance $(\tilde{\AA},\tilde{\bb}, \tilde{R})$ from the \textsc{lp} instance $(\AA,\bb,\cc,K, R)$. If $\tilde{\xx}$ is a solution to the \textsc{lena} (\textsc{len}) instance, 
	then in time $O(n)$, 
	we can compute a solution $\xx$ to the \textsc{lpa} (\textsc{lp}, respectively) instance,
	where the exact case holds when $\eps^{le} = \eps^{lp}=0$.
\end{lemma}

\begin{proof}	
	Based on the solution mapping method described above, given a solution $\tilde{\xx}$, we discard those entries of $\ss$ and $\alpha$, and map back trivially for those entries of $\xx$. 
	As it takes constant time to set the value of each entry of $\xx$ by mapping back trivially, and the size of $\xx$ is $n$, thus the solution mapping takes $O(n)$ time.
	
	If $\tilde{\xx}= (\xx^\top, \ss^\top, \alpha)^\top$ is a solution to the \textsc{lena} instance, then by Definition \ref{def: lena-new}, $\xxtil$ satisfies 
	\begin{align*}
		\abs{\cc^\top\xx-\alpha-K}&\leq \eps^{le},\\
		\abs{\AA\xx+\ss-\bb}&\leq \eps^{le}\vecone.
	\end{align*}
	
	Taking one direction of the absolution value, we obtain
	\begin{align*}
		\cc^\top\xx&\geq \alpha +  K-\eps^{le} \ge  K-\eps^{le},\\
		\AA\xx&\leq -\ss + \bb+\eps^{le}\vecone \le  \bb+\eps^{le}\vecone,
	\end{align*}
	
	As we set in the reduction that $\eps^{le}=\eps^{lp}$, then we have
	\begin{equation*}
		\label{eq: lpa}
		\begin{aligned}
			\cc^\top\xx&\geq K-\eps^{lp},\\
			\AA\xx&\leq \bb+\eps^{lp}\vecone,
		\end{aligned}
	\end{equation*}
	which indicates that $\xx$ is a solution to the \textsc{lpa} instance $(\AA,\bb,\cc,K, R, \epsilon^{lp})$ by Definition \ref{def: lpa}.
\end{proof}


\subsection{LEN(A) to $2$-LEN(A)}
\label{sect: lena-2lena}
\subsubsection{LEN to 2-LEN}
We show the reduction from an \textsc{len} instance $(\AAtil, \bbtil, \Rtil)$ to a \textsc{2-len} instance $(\bar{\AA},\bar{\bb}, \bar{R})$. 
The \textsc{len} instance has the form of $	\AAtil\xxtil=\bbtil$,
where $\AAtil\in\Z^{\tilde{m}\times\tilde{n}}, \bbtil\in\Z^{\tilde{m}}$. 
To reduce it to a \textsc{2-len} instance $(\bar{\AA},\bar{\bb}, \bar{R})$
in which the coefficients of $\bar{\AA}$ are in $\{\pm 1, \pm 2 \}$, 
we do \textit{bitwise decomposition}. 
Algorithm~\ref{algo: bitwise-decomposition} describes how to obtain 
a binary representation of an integer. 
The algorithm takes $z \in \mathbb{Z}$ as an input and output a list $L$ 
consisting of all the powers such that $z = s(z) \sum_{l \in L} 2^l$ where $s(z)$ is the sign of $z$.
For example, $z = -5$, then $L = \{2,0\}$ and $s(z) = -1$.

\begin{algorithm}
	\caption{\textsc{BinaryRepresentation}}
	\label{algo: bitwise-decomposition}
	\KwIn{$z\in \Z$}
	\KwOut{$L$ is a list of powers of 2 such that $z=s(z)\sum_{l\in L} 2^{l}$, where $s(z)$ returns the sign of $z$.}
	$r\leftarrow |z|$\;
	$L\leftarrow []$\;
	\For{$r>0$}{
		$L$.append($\floor{\log_2 r}$)\;
		$r\leftarrow r-2^{\floor{\log_2 r}}$\;}
\end{algorithm}

We will reduce each linear equation in \textsc{len} to 
a linear equation in \textsc{2-len}.
For an arbitrary linear equation $q$ of \textsc{len}: $\tilde{\aa}_q^\top\xxtil=\bbtil(q), q\in[\tilde{m}]$, 
we describe the reduction by the following 4 steps. 
\begin{enumerate}
	\item We run Algorithm \ref{algo: bitwise-decomposition} for each nonzero entry of $\tilde{\aa}_q$ and $\bbtil(q)$ so that each nonzero entry has a binary representation. To simplify notation, we denote the sign of $\tilde{\aa}_q(i)$ as $s_q^i$ and the list returned by \textsc{BinaryRepresentation}($\tilde{\aa}_q(i)$) as $L^i_q$, and the sign of $\bbtil(q)$ as $s_q$ and the list returned by \textsc{BinaryRepresentation}($\tilde{\bb}(q)$) as $L_q$. Thus, the $q$th linear equation of \textsc{len} can be rewritten as
	
	\begin{equation}
		\label{eq: binary-representation}
		\sum_{i\in[\tilde{n}]}\underbrace{\left(s_q^i \sum_{l\in L^i_q}2^{l}\right)}_{\tilde{\aa}_q(i)}\xxtil(i)= \underbrace{s_q \sum_{l\in L_q}2^{l}}_{\tilde{\bb}(q)}.	
	\end{equation}
	
	\item Letting $N_q$ denote the maximum element of $L_q$ and $L^i_q, i\in[\tilde{n}]$, i.e.,
	\[N_q=\floor{\log_2\max\left\{\norm{\tilde{\aa}_q}_{\max}, 
	\abs{\bbtil(q)}\right\}},\] 
	then we can rearrange the left hand side of Eq. \eqref{eq: binary-representation} by gathering those terms located at the same bit (i.e., with the same weight of power of 2), and obtain
	\begin{equation}
		\label{eq: bit-decomp}
		\sum_{l=0}^{N_q} \left(\sum_{\begin{subarray}{c}
				i\in [\tilde{n}]
		\end{subarray}}\mathbbm{1}_{[l\in L^i_q]} \cdot s_q^i\xxtil(i)\right) 2^l= s_q\sum_{l\in L_q}2^{l},
	\end{equation}
	where $\mathbbm{1}$ is an indicator function such that
	\[\mathbbm{1}_{[l\in L^i_q]}=\left\{\begin{matrix}
		1& \text{if $l\in L^i_q$}, \\ 
		0& \text{otherwise}.
	\end{matrix}\right.\]

	\item We decompose Eq. \eqref{eq: bit-decomp} into  $N_q+1$ linear equations such that each one representing a bit, by matching its left hand side and right hand side of Eq. \eqref{eq: bit-decomp} that are located at the same bit. 
	We will also introduce a carry term for each bit,  
	which passes the carry from the equation corresponding to that bit to the equation corresponding to the next bit. 
	Without carry terms, the new system may be infeasible even if the old one is. 
	Moreover, since a carry can be any real number, 
	we represent each carry as a difference of two nonnegative variables $\cc_q(i)-\dd_q(i), \cc_q(i), \dd_q(i)\geq 0$. 
	Starting from the lowest bit, the followings are $N_q+1$ linear equations after decomposition.
	\begin{equation}
		\label{eq: 2-len}
		\begin{aligned}
			\sum_{\begin{subarray}{c}
					i\in [\tilde{n}]
			\end{subarray}}\mathbbm{1}_{[0\in L^i_q]} \cdot s_q^i\xxtil(i)-2[\cc_q(0)-\dd_q(0)]&=s_q\mathbbm{1}_{[0\in L_q]}\\
			\sum_{\begin{subarray}{c}
					i\in [\tilde{n}]
			\end{subarray}}\mathbbm{1}_{[1\in L^i_q]} \cdot s_q^i\xxtil(i)+[\cc_q(0)-\dd_q(0)]-2[\cc_q(1)-\dd_q(1)]&=s_q\mathbbm{1}_{[1\in L_q]}\\
			\vdots\\
			\sum_{\begin{subarray}{c}
					i\in [\tilde{n}]
			\end{subarray}}\mathbbm{1}_{[N_q\in L^i_q]} \cdot s_q^i\xxtil(i)+[\cc_q(N_q-1)-\dd_q(N_q-1)]&=s_q\mathbbm{1}_{[N_q\in L_q]}
		\end{aligned}
	\end{equation}
	
	\item We add an additional constraint for each carry variable $\cc_q(i), \dd_q(i)$\footnote{We remark that Itai's reduction
	does not have these upper bound on carry variables. We need these constraints in our reduction to guarantee that the polytope radius is always well bounded.}:
	\begin{align}
		\cc_q(i) \le 2X(\AAtil, \bbtil) \Rtil, ~ \dd_q(i) \le 2X(\AAtil, \bbtil) \Rtil.
		\label{eqn:cd}
	\end{align}

	These constraints guarantee that the polytope radius of the reduced \textsc{2-len} instance cannot be too large
	\footnote{Note that without these additional constraints
	 the radius of polytope can be unbounded.}.
	In our proofs, we will show that these additional constraints
	do not affect the problem feasibility.
	We then add slack variables $\ss_q^c(i), \ss_q^d(i)\geq 0$ for each carry term
	and turn Eq.~\eqref{eqn:cd} to 
	\begin{align}
		\cc_q(i)+\ss_q^c(i) = 2X(\AAtil, \bbtil) \Rtil, ~~\dd_q(i)+\ss_q^d(i) = 2X(\AAtil, \bbtil) \Rtil, \qquad 0\leq i\leq N_q-1.
		\label{eqn:cd2}
	\end{align}
\end{enumerate}

Repeating the above process for $\tilde{m}$ times, 
we get a \textsc{2-len} instance $(\bar{\AA}, \bar{\bb}, \bar{R})$,
where $\bar{\AA}$ is the coefficient matrix, $\bar{\bb}$ is the right hand side vector,
and $\bar{R}$ is the polytope radius.

If a \textsc{2-len} solver returns $\bar{\xx} = (\xxtil^\top, \cc^\top, \dd^\top, \ss^{c\top}, \ss^{d\top})^\top$ for the \textsc{2-len} instance $(\bar{\AA}, \bar{\bb}, \bar{R})$, 
then we return $\xxtil$ for the \textsc{len} instance $(\AAtil, \bbtil, \Rtil)$;
if the \textsc{2-len} solver returns ``infeasible'' for the 
\textsc{2-len} instance, then we return ``infeasible'' for the 
\textsc{len} instance.

\begin{lemma}[\textsc{len} to \textsc{2-len}]
	\label{lm: len-2len}
	Given an \textsc{len} instance $(\AAtil, \bbtil, \tilde{R})$ where $\tilde{\AA}\in \Z^{\tilde{m}\times\tilde{n}}, \tilde{\bb}\in \Z^{\tilde{m}}$,
	we can construct, in $O\left(\nnz(\tilde{\AA})\log X(\tilde{\AA}, \tilde{\bb})\right)$ time, a \textsc{2-len} instance $(\bar{\AA},\bar{\bb},\bar{R})$ where $\bar{\AA}\in \Z^{\bar{m}\times\bar{n}}, \bar{\bb}\in\Z^{\bar{m}}$ such that 
	\[\bar{n}\leq \tilde{n}+4\tilde{m}\left(1+\log X(\tilde{\AA},\tilde{\bb})\right),~~
	\bar{m}\leq 3\tilde{m}\left(1+\log X(\tilde{\AA},\tilde{\bb})\right),~~\nnz(\bar{\AA})\leq 17\nnz(\tilde{\AA})\left(1+\log X(\tilde{\AA},\tilde{\bb})\right),\]
	\[\bar{R}= 8\tilde{m}\tilde{R}X(\tilde{\AA}, \tilde{\bb})\left(1+\log X(\tilde{\AA}, \tilde{\bb})\right),~~X(\bar{\AA},\bar{\bb})=2X(\tilde{\AA},\tilde{\bb})\tilde{R}.\]
	If the \textsc{len} instance has a solution, then the \textsc{2-len} instance has a solution.
\end{lemma}	
\begin{proof}
	Based on the reduction method described above, from any solution $\tilde{\xx}$ to the \textsc{len} instance such that $\AAtil\xxtil=\bbtil$, we can derive a solution $\bar{\xx} = (\xxtil^\top, \cc^\top, \dd^\top, \ss^{c\top}, \ss^{d\top})^\top$ to the \textsc{2-len} instance. 
	Concretely, for any linear equation $q$ in \textsc{len}, $q\in[\tilde{m}]$, with its decomposed equations as shown in Eq. \eqref{eq: 2-len}, we can set the value of $\cc_q, \dd_q$ from the highest bit as 
	\[\cc_q(N_q-1)=\max\left\{0, ~~s_q\mathbbm{1}_{[N_q\in L_q]}-\sum_{\begin{subarray}{c}
			i\in [\tilde{n}]
	\end{subarray}}\mathbbm{1}_{[N_q\in L^i_q]} \cdot s_q^i\xxtil(i)\right\},
		\]
	
	\[\dd_q(N_q-1)=\max\left\{0, ~~\sum_{\begin{subarray}{c}
			i\in [\tilde{n}]
	\end{subarray}}\mathbbm{1}_{[N_q\in L^i_q]} \cdot s_q^i\xxtil(i)-s_q\mathbbm{1}_{[N_q\in L_q]}\right\}.\]
	And then using backward substitution, we can set the value for the rest entries of $\cc_q$ and $\dd_q$ similarly.
	\small
	\begin{align*}
		\cc_q(N_q-2)&=\max\left\{0, ~~s_q\mathbbm{1}_{[(N_q-1)\in L_q]}+2[\cc_q(N_q-1)-\dd_q(N_q-1)]-\sum_{\begin{subarray}{c}
				i\in [\tilde{n}]
		\end{subarray}}\mathbbm{1}_{[(N_q-1)\in L^i_q]} \cdot s_q^i\xxtil(i)\right\}\\
		&=\max\left\{0, ~~\sum_{l=\{N_q-1, N_q\}}\left(s_q\mathbbm{1}_{[l\in L_q]}-\sum_{\begin{subarray}{c}
				i\in [\tilde{n}]
		\end{subarray}}\mathbbm{1}_{[l\in L^i_q]} \cdot s_q^i\xxtil(i)\right)2^{l-(N_q-1)}\right\},
	\end{align*}
	\small
	\[\dd_q(N_q-2)=\max\left\{0, ~~\sum_{l=\{N_q-1, N_q\}}\left(\sum_{\begin{subarray}{c}
			i\in [\tilde{n}]
	\end{subarray}}\mathbbm{1}_{[l\in L^i_q]} \cdot s_q^i\xxtil(i)-s_q\mathbbm{1}_{[l\in L_q]}\right)2^{l-(N_q-1)}\right\};\]

	\[\vdots\] 	
	\small
	\begin{align*}
		\cc_q(0)
		&=\max\left\{0, ~~\sum_{l=1}^{N_q}\left(s_q\mathbbm{1}_{[l\in L_q]}-\sum_{\begin{subarray}{c}
				i\in [\tilde{n}]
		\end{subarray}}\mathbbm{1}_{[l\in L^i_q]} \cdot s_q^i\xxtil(i)\right)2^{l-1}\right\},
	\end{align*}
\small
	\[\dd_q(0)=\max\left\{0, ~~\sum_{l=1}^{N_q}\left(\sum_{\begin{subarray}{c}
			i\in [\tilde{n}]
	\end{subarray}}\mathbbm{1}_{[l\in L^i_q]} \cdot s_q^i\xxtil(i)-s_q\mathbbm{1}_{[l\in L_q]}\right)2^{l-1}\right\}.\]	
	Substituting $\cc_q(0), \dd_q(0)$ back to the equation of the lowest bit, we will get Eq. \eqref{eq: bit-decomp}, 
	the rearranged binary representation of equation $q$: $\tilde{\aa}_q^\top\tilde{\xx}=\tilde{\bb}(q)$. 
By our setting of the carry terms, we have 
\begin{align*}
\cc_q(i), \dd_q(i) & \le \abs{\sum_{l=i+1}^{N_q}\left(\sum_{\begin{subarray}{c}
	i\in [\tilde{n}]
\end{subarray}}\mathbbm{1}_{[l\in L^i_q]} \cdot s_q^i\xxtil(i)-s_q\mathbbm{1}_{[l\in L_q]}\right)2^{l-i-1}} \\
& \le 2^{N_q} (\norm{\xxtil}_1 + 1)
\le 2 \max \left\{  \norm{\AAtil}_{\max}, \norm{\tilde{\bb}}_{\max} \right\} \Rtil
= 2X(\AAtil, \bbtil) \Rtil.
\end{align*}

By setting slacks
\[
	\ss_q^c(i) = 2X(\AAtil, \bbtil) \Rtil - \cc_q(i), ~~ 
	\ss_q^d(i) = 2X(\AAtil, \bbtil) \Rtil - \dd_q(i),
\]
we satisfy the equations~\eqref{eqn:cd2} in the \textsc{2-len} instance.
Repeating the above process for all $q\in[\tilde{m}]$, we get a feasible solution $\bar{\xx} = (\xxtil^\top, \cc^\top, \dd^\top, \ss^{c\top}, \ss^{d\top})^\top$
to the \textsc{2-len} instance.

	Now, we track the change of problem size after reduction. Based on the reduction method, each linear equation in \textsc{len} can be decomposed into at most $N$ linear equations, where 
	\begin{equation}
		\label{eq: N}
		N=1+\max_{q\in[\tilde{m}]} N_q=1+\floor{\log_2
		\norm{\tilde{\AA}}_{\max}
		} \le 1+\floor{\log X(\tilde{\AA},\tilde{\bb})}.
	\end{equation}	
	Thus, given an \textsc{len} instance with $\tilde{n}$ variables, $\tilde{m}$ linear equations, and $\nnz(\tilde{\AA})$ nonzero entries, we compute the size of the reduced \textsc{2-len} instance as follows.
	\begin{enumerate}
		\item $\bar{n}$ variables. 
		First, all $\tilde{n}$ variables in \textsc{len} are maintained. Then, for each of the $\tilde{m}$ equations in \textsc{len}, it is decomposed into at most $N$ equations, where at most a pair of carry variables are introduced for each newly added equation. Finally, we introduce a slack variable for each carry variable. Thus, we have
		\[\bar{n}\leq\tilde{n}+2\tilde{m}N +2\tilde{m}N\leq \tilde{n}+4\tilde{m}\left(1+\log X(\tilde{\AA},\tilde{\bb})\right).\]
		
		\item $\bar{m}$ linear constraints. First, each equation in \textsc{len} is decomposed into at most $N$ equations in \textsc{2-len}. Next, we add a new constraint for each carry variable. Thus, we have
		\begin{equation}
			\label{eq: bar-m}
			\bar{m}\leq \tilde{m}N+2\tilde{m}N\leq 3\tilde{m}\left(1+\log X(\tilde{\AA},\tilde{\bb})\right).
		\end{equation}
		
		\item $\nnz(\bar{\AA})$ nonzeros.
		To bound it, first, for each nonzero entry in $\tilde{\AA}$, it will be decomposed into at most $N$ bits, thus becomes at most $N$ nonzero entries in $\bar{\AA}$. 
		Then, each equation in \textsc{2-len} involves at most 4 carry variables.
		Furthermore, there are $2\tilde{m}N$ new constraints for carry variables, and each constraint involves a carry variable and a slack variable.
		In total, we have
		\begin{equation}
			\label{eq: 2len-nnz}
			\begin{aligned}
				\nnz(\bar{\AA})&\leq \nnz(\tilde{\AA})N+4\bar{m}+4\tilde{m}N\\
				&\overset{(1)}{\leq} \nnz(\tilde{\AA})\left(1+\log X(\tilde{\AA},\tilde{\bb})\right)+12\tilde{m}\left(1+\log X(\tilde{\AA},\tilde{\bb})\right)+4\tilde{m}\left(1+\log X(\tilde{\AA},\tilde{\bb})\right)\\
				&\overset{(2)}{\leq}  17\nnz(\tilde{\AA})\left(1+\log X(\tilde{\AA},\tilde{\bb})\right), 
			\end{aligned}
		\end{equation}
		where in step (1), we utilize Eq. \eqref{eq: bar-m} for the bound of $\bar{m}$; 
		and in step (2), we use $\tilde{m}\leq \nnz(\tilde{\AA})$.
		
		\item $\bar{R}$, the radius of the polytope in $\ell_1$ norm. 
		We want to upper bound $\bar{\xx}_1$ for every feasible solution to the \textsc{2-len} instance.
		By definition and the triangle inequality,
		\begin{align*}
			\norm{\tilde{\xx}}_1   \le \norm{\xxtil}_1 + \norm{\cc}_1 + \norm{\dd}_1 + \norm{\ss^{c}}_1 + \norm{\ss^{d}}_1 
		\end{align*}
		Note $\norm{\xxtil}_1 \le \Rtil$ and the maximum magnitude of the entries of $\cc, \dd, \ss^c, \ss^d$
		is at most $2X(\AAtil, \bbtil) \Rtil $ by Eq.~\eqref{eqn:cd2}.
		Also note the dimensions of $\cc, \dd, \ss^c, \ss^d$ are $\tilde{m} N $. 
		Thus, 
		\begin{align*}
			\norm{\tilde{\xx}}_1
			\le \Rtil + 2X(\AAtil, \bbtil) \Rtil  \cdot 4 \tilde{m} N 
			\le 8\tilde{m}\tilde{R}X(\tilde{\AA}, \tilde{\bb})\left(1+\log X(\tilde{\AA}, \tilde{\bb})\right).
		\end{align*}

		Hence, it suffices to set
		\[\bar{R}=8\tilde{m}\tilde{R}X(\tilde{\AA}, \tilde{\bb})\left(1+\log X(\tilde{\AA}, \tilde{\bb})\right).\]
		
		\item $X(\bar{\AA},\bar{\bb})=2X(\tilde{\AA},\tilde{\bb})\tilde{R}$ because by construction,
		\[\norm{\bar{\AA}}_{\max}=2, \qquad \norm{\bar{\bb}}_{\max}=2X(\tilde{\AA},\tilde{\bb})\tilde{R}\geq 2.\]
	\end{enumerate}
	
	To estimate the reduction time, it is noticed that it takes $O(N)$ time to run Algorithm \ref{algo: bitwise-decomposition} for each nonzero entry of $\tilde{\AA}$ and $\tilde{\bb}$. And there are at most $\nnz(\tilde{\AA})+\nnz(\tilde{\bb})=O\left(\nnz(\tilde{\AA})\right)$ entries to be decomposed. In addition, it takes $O\left(\nnz(\bar{\AA})+\nnz(\bar{\bb})\right)=O\left(\nnz(\bar{\AA})\right)$ to construct $\bar{\AA}$ and $\bar{\bb}$. By Eq. \eqref{eq: 2len-nnz}, performing such a reduction takes time in total
	\begin{align*}
		O\left(N\nnz(\tilde{\AA})+\nnz(\bar{\AA})\right)= O\left(\nnz(\tilde{\AA})\log X(\tilde{\AA}, \tilde{\bb})\right).
	\end{align*}
\end{proof}

\subsubsection{LENA to 2-LENA}

\begin{definition}[$k$-LEN Approximate Problem ($k$-\textsc{lena})]
	\label{def: k-lena-new}
	A $k$-\textsc{lena} instance is given by a $k$-\textsc{len} instance $(\AA,\bb,R)$ as in Definition \ref{def: k-len} and an error parameter $\eps\in[0,1]$, which we collect in a tuple $(\AA,\bb,R,\epsilon)$. We say an algorithm solves the $k$-\textsc{lena} problem, if, given any $k$-\textsc{lena} instance, it returns a vector $\xx\geq \boldsymbol{0}$ such that
	\begin{align*}
		\abs{\AA\xx-\bb}\leq \eps\vecone, 
	\end{align*}
	where $\abs{\cdot}$ is entrywise absolute value and $\vecone$ is the all-1 vector, or it correctly declares that the associated $k$-\textsc{len} instance is infeasible.
\end{definition}

We can use the same reduction method in the exact case to reduce an \textsc{lena} instance to a \textsc{2-lena} instance. Furthermore, if a \textsc{2-lena} solver returns $\bar{\xx} = (\xxtil^\top, \cc^\top, \dd^\top, \ss^{c\top}, \ss^{d\top})^\top$ for the \textsc{2-lena} instance $(\bar{\AA}, \bar{\bb}, \bar{R}, \epsilon^{2le})$, 
then we return $\xxtil$ for the \textsc{lena} instance $(\AAtil, \bbtil, \Rtil, \epsilon^{le})$;
if the \textsc{2-lena} solver returns ``infeasible'' for the 
\textsc{2-lena} instance, then we return ``infeasible'' for the \textsc{lena} instance.

\begin{lemma}[\textsc{lena} to \textsc{2-lena}]
	\label{lm: lena-2lena}
	Given an \textsc{lena} instance $(\tilde{\AA},\tilde{\bb},\tilde{R}, \epsilon^{le})$, where $\tilde{\AA}\in \Z^{\tilde{m}\times\tilde{n}}, \tilde{\bb}\in \Z^{\tilde{m}}$, we can reduce it to an \textsc{2-lena} instance $(\bar{\AA},\bar{\bb},\bar{R},\epsilon^{2le})$ by letting
	\[\epsilon^{2le}=\frac{\eps^{le}}{2X(\tilde{\AA},\tilde{\bb})},\]
	and using Lemma \ref{lm: len-2len} to construct a \textsc{2-len} instance $(\bar{\AA},\bar{\bb},\bar{R})$ from the \textsc{len} instance $(\tilde{\AA},\tilde{\bb},\tilde{R})$.
	If $\bar{\xx}$ is a solution to the \textsc{2-lena} (\textsc{2-len}) instance, then in time $O(\tilde{n})$, we can compute a solution $\tilde{\xx}$ to the \textsc{lena} (\textsc{len}, respectively) instance,
	where the exact case holds when $\eps^{2le} = \eps^{le}=0$.
\end{lemma}	
\begin{proof}
	
	Based on the solution mapping method described above, given a solution $\bar{\xx}$, we discard those entries of carry variables and slack variables, and map back trivially for those entries of $\xxtil$. As it takes constant time to set the value of each entry of $\xxtil$ by mapping back trivially, and the size of $\xxtil$ is $\tilde{n}$, thus the solution mapping takes $O(\tilde{n})$ time.
	
	Now, we conduct an error analysis. 
	By Definition \ref{def: k-lena-new}, the error of each linear equation in \textsc{2-lena} can be bounded by $\epsilon^{2le}$. 
	In particular, the equation in the \textsc{2-lena} instance 
	that corresponds to the highest bit of the $q$th equation in the \textsc{lena} instance
	satisfies
	\[\abs{\sum_{\begin{subarray}{c}
			i\in [\tilde{n}]
	\end{subarray}}\mathbbm{1}_{[N_q\in L^i_q]} \cdot s_q^i\xxtil(i)+[\cc_q(N_q-1)-\dd_q(N_q-1)]-s_q\mathbbm{1}_{[N_q\in L_q]}}\leq \epsilon^{2le},\]
	which can be rearranged as
	\begin{small}
	\begin{equation}
		\label{eq: 2len-error-1}
		-\epsilon^{2le}-[\cc_q(N_q-1)-\dd_q(N_q-1)]\leq\sum_{\begin{subarray}{c}
				i\in [\tilde{n}]
		\end{subarray}}\mathbbm{1}_{[N_q\in L^i_q]} \cdot s_q^i\xxtil(i)-s_q\mathbbm{1}_{[N_q\in L_q]}\leq \epsilon^{2le}-[\cc_q(N_q-1)-\dd_q(N_q-1)].
	\end{equation}
	\end{small}
	For the second highest bit, we have
	\begin{small}
	\begin{multline}
		\label{eq: 2len-error-2}
		-\epsilon^{2le}+2[\cc_q(N_q-1)-\dd_q(N_q-1)]-[\cc_q(N_q-2)-\dd_q(N_q-2)]\\
		\leq\sum_{\begin{subarray}{c}
				i\in [\tilde{n}]
		\end{subarray}}\mathbbm{1}_{[(N_q-1)\in L^i_q]} \cdot s_q^i\xxtil(i)-s_q\mathbbm{1}_{[(N_q-1)\in L_q]}\\
	\leq \epsilon^{2le}+2[\cc_q(N_q-1)-\dd_q(N_q-1)]-[\cc_q(N_q-2)-\dd_q(N_q-2)].
	\end{multline}
	\end{small}
	We can eliminate the pair of carry $[\cc_q(N_q-1)-\dd_q(N_q-1)]$ by computing $2\times\text{Eq. \eqref{eq: 2len-error-1}} + \text{Eq. \eqref{eq: 2len-error-2}}$, and obtain
	\begin{small}
		\begin{multline}
			\label{eq: 2len-error-3}
			-(2^0+2^1)\epsilon^{2le}-[\cc_q(N_q-2)-\dd_q(N_q-2)]\\
			\leq\sum_{l=\{N_q-1, N_q\}}\left(\sum_{\begin{subarray}{c}
					i\in [\tilde{n}]
			\end{subarray}}\mathbbm{1}_{[l\in L^i_q]} \cdot s_q^i\xxtil(i)-s_q\mathbbm{1}_{[l\in L_q]}\right)2^{l-(N_q-1)}\\
			\leq (2^0+2^1)\epsilon^{2le}-[\cc_q(N_q-2)-\dd_q(N_q-2)].
		\end{multline}
	\end{small}
	By repeating the process until the equation of the lowest bit, we can eliminate all pairs of carry variables and obtain
	\begin{small}
		\begin{multline}
			\label{eq: 2len-error-4}
			-(2^0+\cdots+2^{N_q})\epsilon^{2le}
			\leq\underbrace{\sum_{l=0}^{N_q}\left(\sum_{\begin{subarray}{c}
					i\in [\tilde{n}]
			\end{subarray}}\mathbbm{1}_{[l\in L^i_q]} \cdot s_q^i\xxtil(i)-s_q\mathbbm{1}_{[l\in L_q]}\right)2^{l}}_{=\tilde{\aa}_q^\top\tilde{\xx}-\tilde{\bb}(q) ~~ \text{by Eq.  \eqref{eq: bit-decomp}}}
			\leq (2^0+\cdots+2^{N_q})\epsilon^{2le}.
		\end{multline}
	\end{small}
	Hence, we can bound the $q$th linear equation in \textsc{len} by
	\[\abs{\tilde{\aa}_q^\top\tilde{\xx}-\tilde{\bb}(q)}\leq (2^0+\cdots+2^{N_q})\epsilon^{2le}\leq2^{N_q+1}\epsilon^{2le},\]
	which implies that the error of the $q$th equation of \textsc{lena} is accumulated as a weighted sum of at most $N_q$ equations in \textsc{2-lena}, where the weight is in the form of power of 2. 
	
	We use $\tau^{le}$ to denote the error of $\xxtil$ that is obtained by mapping back $\bar{\xx}$ with at most $\eps^{2le}$ additive error. %
	Then we have
	\begin{align*}
		\tau^{le}=&~\max_{q\in[\tilde{m}]}\abs{\tilde{\aa}_q^\top\tilde{\xx}-\tilde{\bb}(q)}\\
		\leq &~ \max_{q\in[\tilde{m}]} 2^{N_q+1}\epsilon^{2le}\\
		\leq &~ 2^{N}\epsilon^{2le} && \text{Because $N=1+\max_{q\in\tilde{m}}N_q$ as in Eq.~\eqref{eq: N}}\\
		\leq &~ 2^{1+\log X(\tilde{\AA},\tilde{\bb})}\epsilon^{2le} &&\text{Because $N=1+\floor{\log X(\tilde{\AA},\tilde{\bb})}$  as in Eq.~\eqref{eq: N}}\\
		\leq &~ 2X(\tilde{\AA},\tilde{\bb})\epsilon^{2le}
	\end{align*}	 
	As we set in the reduction that $\epsilon^{2le}=\frac{\eps^{le}}{2X(\tilde{\AA},\tilde{\bb})}$, then we have
	\begin{align*}
		\tau^{le}\leq 2X(\tilde{\AA},\tilde{\bb})\cdot\frac{\eps^{le}}{2X(\tilde{\AA},\tilde{\bb})}
		=\epsilon^{le},
	\end{align*}
	which indicates that $\tilde{\xx}$ is a solution to the \textsc{lena} instance $(\tilde{\AA},\tilde{\bb},\tilde{R}, \epsilon^{le})$.	
\end{proof}


\subsection{$2$-LEN(A) to $1$-LEN(A)}
\label{sect: 2lena-1lena}
\subsubsection{$2$-LEN to $1$-LEN}
We show the reduction from a \textsc{2-len} instance $(\bar{\AA},\bar{\bb},\bar{R})$ to a \textsc{1-len} instance $(\hat{\AA},\hat{\bb},\hat{R})$. The \textsc{2-len} instance has the form of $\bar{\AA}\bar{\xx}=\bar{\bb}$, where entries of $\bar{\AA}$ are integers between $[-2,2]$. 
To reduce it to a \textsc{1-len} instance, for each variable $\bar{\xx}(j)$ that has a $\pm 2$ coefficient, we introduce a new variable $\bar{\xx}'(j)$, replace every $\pm 2 \bar{\xx}(j)$
with $\pm(\bar{\xx}(j) + \bar{\xx}'(j))$, and add an additional equation $\bar{\xx}(j)-\bar{\xx}'(j)=0$.

If a \textsc{1-len} solver returns $\hat{\xx}=(\bar{\xx}^\top, (\bar{\xx}')^\top)^\top$ for the \textsc{1-len} instance $(\hat{\AA},\hat{\bb},\hat{R})$, then we return $\bar{\xx}$ for the \textsc{2-len} instance $(\bar{\AA},\bar{\bb},\bar{R})$;
if the \textsc{1-len} returns ``infeasible'' for the \textsc{1-len} instance, then we return ``infeasible'' for the \textsc{2-len} instance.

\begin{lemma}[\textsc{2-len} to \textsc{1-len}]
	\label{lm: 2len-1len}
	Given a \textsc{2-len} instance $(\bar{\AA},\bar{\bb},\bar{R})$ where $\bar{\AA}\in\Z^{\bar{m}\times\bar{n}}, \bar{\bb}\in \Z^{\bar{m}}$, 
	we can construct, in $O(\nnz(\bar{\AA}))$ time, a \textsc{1-len} instance $(\hat{\AA},\hat{\bb},\hat{R})$ where $\hat{\AA}\in \Z^{\hat{m}\times\hat{n}}, \hat{\bb}\in \Z^{\hat{m}}$ such that
	\[\hat{n}\leq 2\bar{n},~~\hat{m}\leq \bar{m}+\bar{n},~~ \nnz(\hat{\AA})\leq 4\nnz(\bar{\AA}),~~ \hat{R}=2\bar{R}, ~~ X(\hat{\AA},\hat{\bb})=X(\bar{\AA},\bar{\bb}).\]
	If the \textsc{2-len} instance has a solution, then the \textsc{1-len} instance has a solution.
\end{lemma}	
\begin{proof}
	Based on the reduction described above, from any solution $\bar{\xx}$ to the \textsc{2-len} instance such that $\bar{\AA}\bar{\xx}=\bar{\bb}$, 
	we can derive a solution $\hat{\xx}=(\bar{\xx}^\top, (\bar{\xx}')^\top)^\top$ to the \textsc{1-len} instance. 
	Concretely, for each $\bar{\xx}(j)$ having a $\pm2$ coefficient in $\bar{\AA}$, 
	we set $\bar{\xx}'(j) = \bar{\xx}(j)$, where $\bar{\xx}'(j)$ is the entry 
	that we use to replace $\pm 2 \bar{\xx}(j)$ with $\pm (\bar{\xx}_j + \bar{\xx}'_j)$
	in the reduction.
	We can check that $\hat{\xx}$ is a solution to the \textsc{1-len} instance.
	
	Now, we track the change of problem size after reduction. Given a \textsc{2-len} instance with $\bar{n}$ variables, $\bar{m}$ linear equations, and $\nnz(\bar{\AA})$ nonzero entries, we can compute the size of the reduced \textsc{1-len} instance as follows.
	\begin{enumerate}
		\item $\hat{n}$ variables, where
		$\hat{n}\leq 2\bar{n}$.
		It is because each variable $\bar{\xx}(j)$ in \textsc{2-len} is replaced by at most 2 variables $\bar{\xx}(j)+\bar{\xx}'(j)$ in \textsc{1-len}.
		
		\item $\hat{m}$ linear constraints. In addition to the original $\bar{m}$ linear equations, each variable $\bar{\xx}(j)$ with $\pm2$ coefficient in \textsc{2-len} will introduce a new equation $\bar{\xx}(j)-\bar{\xx}'(j)=0$ in \textsc{1-len}.  Thus,
		$\hat{m}\leq \bar{m}+\bar{n}$.

		\item $\nnz(\hat{\AA})$ nonzeros. To bound it,
		first, each nonzero entry in $\bar{\AA}$ becomes at most two nonzero entries in $\hat{\AA}$ because of the replacement of $\pm2\bar{\xx}(j) $ by $\pm(\bar{\xx}(j)+\bar{\xx}'(j))$.
		Next, at most $2\bar{n}$ new nonzero entries are generated because of the newly added equation $\bar{\xx}(j)-\bar{\xx}'(j)=0$. Thus,
		\[\nnz(\hat{\AA})\leq 2(\nnz(\bar{\AA})+\bar{n})\leq4\nnz(\bar{\AA}).\]

		\item $\hat{R}$ radius of polytope in $\ell_1$ norm. We have
		\[\norm{\hat{\xx}}_1=\norm{\bar{\xx}}_1+\norm{\bar{\xx}'}_1\leq 2\norm{\bar{\xx}}_1\leq 2\bar{R}.\]
		Hence, it suffices to set
		$\hat{R}= 2\bar{R}$.
		
		\item $X(\hat{\AA},\hat{\bb})=X(\bar{\AA},\bar{\bb})$ because by construction,
		\[\norm{\hat{\AA}}_{\max}\leq \norm{\bar{\AA}}_{\max},\qquad \norm{\hat{\bb}}_{\max}=\norm{\bar{\bb}}_{\max}.\]

			\end{enumerate}
	
	To estimate the reduction time, it takes constant time to deal with each $\bar{\AA}(i,j)$ being $\pm2$, and there are at most $\nnz(\bar{\AA})$ occurrences of $\pm2$ to be dealt with. Hence, it takes $O(\nnz(\bar{\AA}))$ time to eliminate all the occurrence of $\pm2$. Moreover, copy the rest coefficients also takes $O(\nnz(\bar{\AA}))$ time. Thus, the reduction of this step takes $O(\nnz(\bar{\AA}))$ time. 

\end{proof}
	
	\subsubsection{$2$-LENA to $1$-LENA}

	We can use the same reduction method in the exact case to reduce a \textsc{2-lena} instance to a \textsc{1-lena} instance. 
	We can also use the same solution mapping method, but we make a slight adjustment in the approximate case for the simplicity of the following error analysis. 
	More specifically, if a \textsc{1-lena} solver returns $\hat{\xx}=(\xx^\top, (\xx')^\top)^\top$ for the \textsc{1-lena} instance $(\hat{\AA}, \hat{\bb}, \hat{R}, \epsilon^{1le})$, 
	instead of returning $\xx$ directly as a solution to the \textsc{2-lena} instance $(\bar{\AA}, \bar{\bb}, \bar{R}, \epsilon^{2le})$, 
	we set $\bar{\xx}(i) = \frac{1}{2}(\xx(i)+ \xx'(i))$ if $\xx(i)$ has a coefficient $\pm 2$ 
	in the  \textsc{2-len} instance, and set $\bar{\xx}(i) = \xx(i)$ otherwise.
	In addition, if the \textsc{1-lena} solver returns ``infeasible'' for the \textsc{1-lena} instance, 
then we return ``infeasible'' for the \textsc{2-lena} instance.

\begin{lemma}[\textsc{2-lena} to \textsc{1-lena}]
	\label{lm: 2lena-1lena}
	Given a \textsc{2-lena} instance $(\bar{\AA},\bar{\bb},\bar{R},\epsilon^{2le})$ where $\bar{\AA}\in\Z^{\bar{m}\times\bar{n}}, \bar{\bb}\in \Z^{\bar{m}}$, 
	we can reduce it to a \textsc{1-lena} instance $(\hat{\AA},\hat{\bb},\hat{R},\epsilon^{1le})$ by letting 
	\[\epsilon^{1le}=\frac{\eps^{2le}}{\bar{n}+1},\]
	and using Lemma \ref{lm: 2len-1len} to construct a \textsc{1-len} instance $(\hat{\AA},\hat{\bb},\hat{R})$ from the \textsc{2-len} instance $(\bar{\AA},\bar{\bb},\bar{R})$. 
	If $\hat{\xx}$ is a solution to the \textsc{1-lena} (\textsc{1-len}) instance, then in time $O(\bar{n})$, we can compute a solution $\bar{\xx}$ to the \textsc{2-lena} (\textsc{2-len}, respectively) instance,
	where the exact case holds when $\eps^{1le} = \eps^{2le}=0$.

\end{lemma}	

\begin{proof}	
	
	Based on the solution mapping method described above, it takes constant time to set the value of each entry of $\bar{\xx}$ by computing an averaging or mapping back trivially, and the size of $\bar{\xx}$ is $\bar{n}$, thus the solution mapping takes $O(\bar{n})$ time.
	
	Now, we conduct an error analysis. By Definition \ref{def: k-lena-new}, the error of each linear equation in \textsc{1-lena} can be bounded by $\epsilon^{1le}$. For a single occurrence of $\bar{\AA}(i,j)=\pm2$, we first bound the error of the equation $\bar{\xx}(j)-\bar{\xx}'(j)=0$, and obtain
	\[|\bar{\xx}(j)-\bar{\xx}'(j)|\leq \epsilon^{1le},\]
	hence, we have
	\begin{equation}
		\label{eq: 2LEN1}
		-\epsilon^{1le}+2\bar{\xx}(j)\leq \bar{\xx}(j)+\bar{\xx}'(j)\leq \epsilon^{1le}+2\bar{\xx}(j).
	\end{equation}		
	We first consider the case that, in the equation $\bar{\aa}_i^\top\bar{\xx}=\bar{\bb}(i)$, there is only one entry $j$ such that $\bar{\aa}_i(j)=\pm2$. By separating this term, we can write
	\[\bar{\aa}_i^\top\bar{\xx}=\sum_{k\neq j}\bar{\AA}(i,k)\bar{\xx}(k)\pm2\bar{\xx}(j), \qquad \hat{\aa}_i^\top\hat{\xx}=\sum_{k\neq j}\bar{\AA}(i,k)\bar{\xx}(k)\pm(\bar{\xx}(j)+\bar{\xx}'(j)).\] 
	Adding $\sum_{k\neq j}\bar{\AA}(i,k)\bar{\xx}(k)$ to Eq, \eqref{eq: 2LEN1}, we have
	\begin{small}
		\begin{equation*}
			\label{eq: 2LEN2}
			-\epsilon^{1le}\underbrace{\pm2\bar{\xx}(j)+\sum_{k\neq j}\bar{\AA}(i,k)\bar{\xx}(k)}_{\bar{\aa}_i^\top\bar{\xx}}\leq \underbrace{\sum_{k\neq j}\bar{\AA}(i,k)\bar{\xx}(k)\pm(\bar{\xx}(j)+\bar{\xx}'(j))}_{\hat{\aa}_i^\top\hat{\xx}} \leq \epsilon^{1le}\underbrace{\pm2\bar{\xx}(j)+\sum_{k\neq j}\bar{\AA}(i,k)\bar{\xx}(k)}_{\bar{\aa}_i^\top\bar{\xx}}.
		\end{equation*}
	\end{small}	
	And since $\hat{\bb}(i)=\bar{\bb}(i)$, we can subtract it from all parts and get
	\begin{equation*}
		\label{eq: 2LEN3}
		-\epsilon^{1le}+\bar{\aa}_i^\top\bar{\xx}-\bar{\bb}(i)\leq \hat{\aa}_i^\top\hat{\xx}-\hat{\bb}(i) \leq \epsilon^{1le}+\bar{\aa}_i^\top\bar{\xx}-\bar{\bb}(i).
	\end{equation*}
	which can be further transformed to
	\begin{equation}
		\label{eq: 2LEN4}
		-\epsilon^{1le}+\hat{\aa}_i^\top\hat{\xx}-\hat{\bb}(i)\leq \bar{\aa}_i^\top\bar{\xx}-\bar{\bb}(i)\leq \hat{\aa}_i^\top\hat{\xx}-\hat{\bb}(i)+ \epsilon^{1le}.
	\end{equation}
	If there are $k_i$ occurrence of $\bar{\AA}(i,j)=\pm2$, then we can generalize Eq. \eqref{eq: 2LEN4} to
	\begin{equation}
		\label{eq: 2LEN4.5}
		-k_i\epsilon^{1le}+\hat{\aa}_i^\top\hat{\xx}-\hat{\bb}(i)\leq \bar{\aa}_i^\top\bar{\xx}-\bar{\bb}(i)\leq \hat{\aa}_i^\top\hat{\xx}-\hat{\bb}(i)+ k_i\epsilon^{1le},
	\end{equation}
	hence, we can bound
	\begin{equation}
		\label{eq: 2LEN5}
		\begin{aligned}
			\abs{\bar{\aa}_i^\top\bar{\xx}-\bar{\bb}(i)}\leq k_i\epsilon^{1le}+\abs{\hat{\aa}_i^\top\hat{\xx}-\hat{\bb}(i)}\leq (k_i+1)\epsilon^{1le},
		\end{aligned}
	\end{equation}	
	where the last inequality is because $\abs{\hat{\aa}_i^\top\hat{\xx}-\hat{\bb}(i)}\leq \epsilon^{1le}$ by applying the error of the \textsc{1-lena} instance.
	
	We use $\tau^{2le}$ to denote the error of $\bar{\xx}$ that is obtained by mapping back $\hat{\xx}$ with at most $\eps^{1le}$ additive error. %
	Then we have
	\begin{equation*}
		\label{eq: 2LEN8}
		\begin{aligned}
			\tau^{2le}&=\max_{i\in[\bar{m}]}\abs{\bar{\aa}_i^\top\bar{\xx}-\bar{\bb}(i)}\\
			&\leq  \max_{i\in[\bar{m}]}(k_i+1)\epsilon^{1le} && \text{Because of Eq. \eqref{eq: 2LEN5}}\\
			&\leq (\bar{n}+1)\epsilon^{1le} && \text{Because $k_i\leq \bar{n}$}
		\end{aligned}			
	\end{equation*}
	As we set in the reduction that $\epsilon^{1le}=\frac{\eps^{2le}}{\bar{n}+1}$, then we have
	\begin{align*}
		\tau^{2le}\leq (\bar{n}+1)\cdot\frac{\eps^{2le}}{\bar{n}+1}=\eps^{2le},
	\end{align*}
	which indicates that $\bar{\xx}$ is a solution to the \textsc{2-lena} instance $(\bar{\AA},\bar{\bb},\bar{R},\epsilon^{2le})$.
	\end{proof}


\subsection{$1$-LEN(A) to FHF(A)}
\label{sect: 1lena-fhfa}
\subsubsection{$1$-LEN to FHF}
The following is the approach to reduce a \textsc{1-len} instance $(\hat{\AA},\hat{\bb},\hat{R})$ to an \textsc{fhf} instance $(G^h, F^h, \uu^h, \calH^h, s, t)$. 
The \textsc{1-len} has the form of $\hat{\AA}\hat{\xx}=\hat{\bb}$, where $\hat{\AA}\in \Z^{\hat{m}\times\hat{n}}, \hat{\bb}\in\Z^{\hat{m}}$. For an arbitrary equation $i$ in the \textsc{1-len} instance, $\hat{\aa}_i^\top\hat{\xx}=\hat{\bb}(i), i\in[\hat{m}]$, let $J_i^{+}=\{j|\hat{\aa}_i(j)=1\}$ and $J_i^{-}=\{j|\hat{\aa}_i(j)=-1\}$ denote the set of indices of variables with coefficients being 1 and -1 in equation $i$, respectively. Then, each equation can be rewritten as a difference of the sum of variables with coefficient 1 and -1:
\begin{equation}
	\label{eq: l-u-LEN}
	\sum_{j\in J_i^+}\hat{\xx}(j) - \sum_{j\in J_i^{-}}\hat{\xx}(j) =\hat{\bb}(i), ~~ i\in[\hat{m}].
\end{equation}

We claim that $\hat{\AA}\hat{\xx}=\hat{\bb}$ can be represented by a graph that is composed of a number of homologous edges and fixed flow edges, as shown in Figure \ref{fig: homo}. 
\begin{figure}[t]
	\centering
	\includegraphics[width=0.7\textwidth]{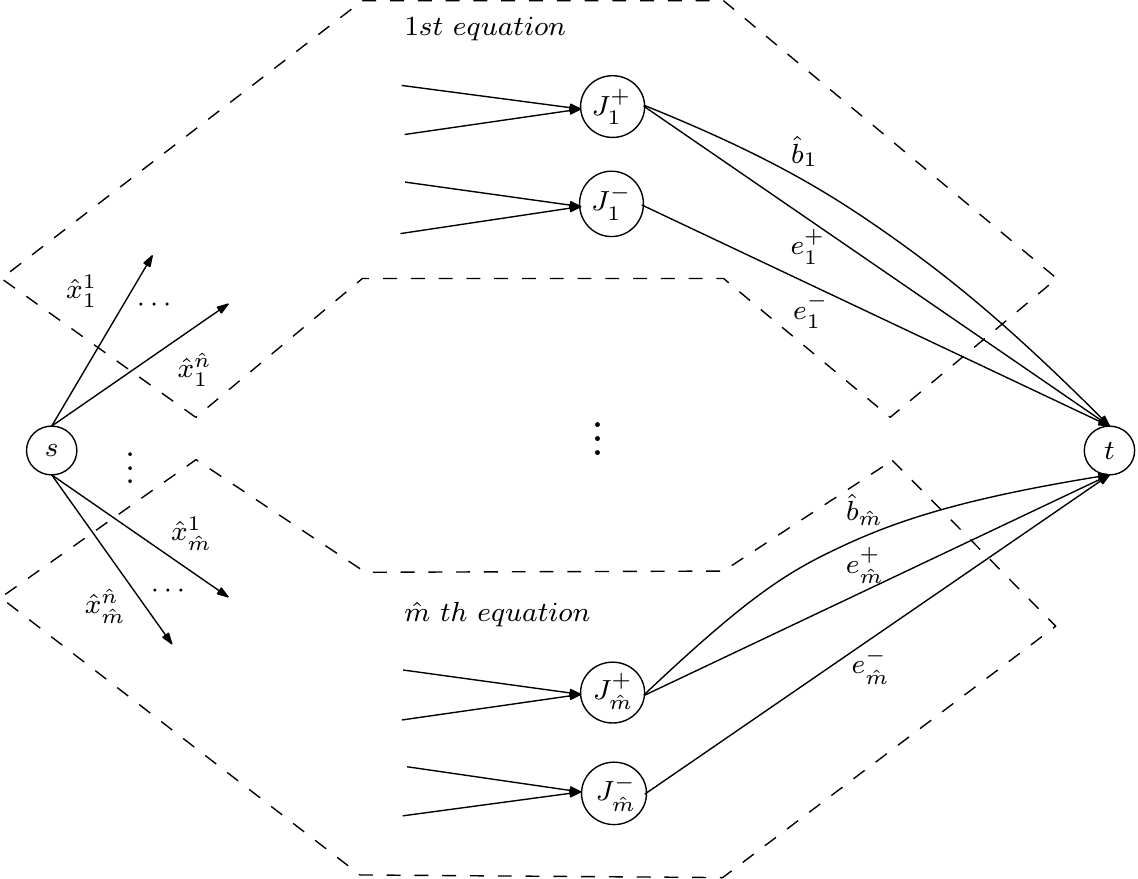}
	\caption{The reduction from \textsc{1-len} to \textsc{fhf}.} 
	\label{fig: homo}
\end{figure}

More specifically, the fixed homologous flow network consists of a source $s$, a sink $t$, and $\hat{m}$ sections such that each section $i$ represents the $i$th linear equation in \textsc{1-len}, as shown in Eq. \eqref{eq: l-u-LEN}.

Inside each section $i$, there are $2$ vertices $\{J_i^{-}, J_i^{+}\}$ and a number of edges:
\begin{itemize}		
	\item For the incoming edges of $\{J_i^{-}, J_i^{+}\}$,
	\begin{itemize}
		\item if $\hat{\aa}_i(j)=1$, then $s$ is connected to $J_i^+$ by edge $\hat{x}_i^j$ with capacity $\hat{R}$;
		\item if $\hat{\aa}_i(j)=-1$, then $s$ is connected to $J_i^-$ by edge $\hat{x}_i^j$ with capacity $\hat{R}$;
		\item if $\hat{\aa}_i(j)=0$, no edge is needed.
	\end{itemize}
	Note that the problem sparsity is preserved in the graph construction. 
	The amount of flow routed in these incoming edges equals the value of the corresponding variables. To ensure the consistency of the value of variables over $\hat{m}$ equations, those incoming edges that correspond to the same variable are forced to route the same amount of flow, i.e., $(\hat{x}_1^j, \cdots, \hat{x}_{\hat{m}}^j), j\in[\hat{n}]$ constitute a homologous edge set that corresponds to the variable $\hat{\xx}(j)$. Note that the size of such a homologous edge set is at most $\hat{m}$.		
	\item For the outgoing edges of $\{J_i^{-}, J_i^{+}\}$,
	\begin{itemize}
		\item $J_i^+$ is connected to $t$ by a fixed flow edge $\hat{b}_i$ that routes $\hat{\bb}(i)$ units of flow;
		\item $J_i^{+}$ and $J_i^-$ are connected to $t$ by a pair of homologous edges $e_i^{+}, e_i^{-}$ with capacity $\hat{R}$.
	\end{itemize}				
\end{itemize}

If an \textsc{fhf} solver returns $\ff^h$ for the \textsc{fhf} instance $(G^h, F^h, \uu^h, \calH^h, s, t)$, then we return $\hat{\xx}$ for the \textsc{1-len} instance $(\hat{\AA},\hat{\bb},\hat{R})$, by setting 
for every $j \in [\hat{n}]$,
\[\hat{\xx}(j)=\ff^h(\hat{x}^j_i), ~\text{for an arbitrary $i\in[\hat{m}]$ }\]
If the \textsc{fhf} solver returns ``infeasible'' for the \textsc{fhf} instance, 
then we return ``infeasible'' for the \textsc{1-len} instance.

\begin{lemma}[\textsc{1-len} to \textsc{fhf}]
	\label{lm: 1len-fhf}
	Given a \textsc{1-len} instance $(\hat{\AA},\hat{\bb},\hat{R})$ where $\hat{\AA}\in\Z^{\hat{m}\times\hat{n}}, \hat{\bb}\in \Z^{\hat{m}}$,
	we can construct, in time $O(\nnz(\hat{\AA}))$, an \textsc{fhf} instance $(G^h, F^h, \uu^h, \mathcal{H}^h=(H_1, \cdots, H_h), s, t)$ such that
	\[|V^h|=2\hat{m}+2,~~ |E^h|\leq4\nnz(\hat{\AA}),~~ |F^h|=\hat{m},~~ h=\hat{n}+\hat{m},~~ \norm{\uu^h}_{\max}=\max\left\{\hat{R}, X(\hat{\AA}, \hat{\bb})\right\}.\]
	If the \textsc{1-len} instance has a solution, then the \textsc{fhf} instance has a solution.
\end{lemma}	
\begin{proof}
	According to the reduction described above, from any solution $\hat{\xx}$ to the \textsc{1-len} instance such that $\hat{\AA}\hat{\xx}=\hat{\bb}$, we can derive a solution $\ff^h$ to the \textsc{fhf} instance. Concretely, we define a feasible flow $\ff^h$ as follows: 
	\begin{itemize}
		\item For incoming edges of $\{J_i^{-}, J_i^{+}\}$, we set
		\[\ff^h(\hat{x}^j_1)=\cdots=\ff^h(\hat{x}^j_{\hat{m}})=\hat{\xx}(j)\leq \hat{R}, \qquad \forall j\in[\hat{n}],\]
		which satisfies the homologous constraint for the homologous edge sets $(\hat{x}_1^j, \cdots, \hat{x}_{\hat{m}}^j), j\in[\hat{n}]$, as well as the capacity constraint. 
		\item For outgoing edges of $\{J_i^{-}, J_i^{+}\}$, we set
		\[\ff^h(\hat{b}_i)=\hat{\bb}(i),\qquad \forall i\in\hat{m},\]
		which satisfies the fixed flow constraint for edges $\hat{b}_1,\cdots,\hat{b}_{\hat{m}}$;
		and set
		\[\ff^h(e^+_i)=\ff^h(e^-_i)=\sum_{j\in J_i^{-}}\hat{\xx}(j)=\sum_{j\in J_i^{+}}\hat{\xx}(j)-\hat{\bb}(j),\]
		which satisfies the homologous constraint for edge $e_i^{+}, e_i^{-}$, and the conservation of flows for vertices $J_i^{+}, J_i^{-}$.
	\end{itemize}
	Therefore, we conclude that $\ff^h$ is a feasible flow to the \textsc{fhf} instance.

	 \vspace{8pt}
	 Now, we track the change of problem size after reduction. Based on the reduction method, given a \textsc{1-lena} instance with $\hat{n}$ variables, $\hat{m}$ linear equations, and $\nnz(\hat{\AA})$ nonzero entries, it is straightforward to get the size of the reduced \textsc{fhfa} instance as follows.
	\begin{enumerate}
		\item $|V^h|$ vertices, where
		$|V^h|=2\hat{m}+2$.
		\item $|E^h|$ edges, where
		$|E^h|=\nnz(\hat{\AA})+3\hat{m}\leq 4\nnz(\hat{\AA})$,
		since $\hat{m}\leq \nnz(\hat{\AA})$.
		\item $|F^h|$ fixed flow edges, where
		$|F^h|=\hat{m}$.
		\item $h$ homologous edge sets, where
		$h=\hat{m}+\hat{n}$.
		\item The maximum edge capacity is bounded by 
		\[\norm{\uu^h}_{\max}=\max\left\{\hat{R}, \norm{\hat{\bb}}_{\max}\right\}\leq \max\left\{\hat{R}, X(\hat{\AA}, \hat{\bb})\right\}.\]
	\end{enumerate}

	To estimate the reduction time, as there are $|V^h|=O(\hat{m})$ vertices and $|E^h|=O(\nnz(\hat{\AA}))$ edges in $G^h$, thus, performing such a reduction takes $O(\nnz(\hat{\AA}))$ time to construct $G^h$.
\end{proof}

\subsubsection{$1$-LENA to FHFA}
The above lemma shows the reduction between exactly solving a \textsc{1-len} instance and exactly solving a \textsc{fhf} instance. 
Next, we generalize the case with exact solutions to the case that allows approximate solutions.
First of all, we give a definition of the approximate version of \textsc{fhf}.

\begin{definition}[FHF Approximate Problem (\textsc{fhfa})]
	\label{def: fhfa-new}
	An \textsc{fhfa} instance is given by an \textsc{fhf} instance $\left(G, F, \uu, \mathcal{H}, s, t\right)$ as in Definition \ref{def: fhf}, and an error parameter $\epsilon\in[0,1]$, which we collect in a tuple $\left(G, F, \mathcal{H}, \uu, s,t, \epsilon\right)$.
	We say an algorithm solves the \textsc{fhfa} problem, if, given any \textsc{fhfa} instance, it returns a flow $\ff\geq\bf{0}$ that satisfies
	\begin{align}
		&\uu(e)-\eps\leq \ff(e)\leq\uu(e)+\eps, ~ \forall e \in F \label{eqn:fhfa_error_congestion_1} \\
		&0\leq \ff(e)\leq\uu(e)+\eps, ~ \forall e \in E\backslash F \label{eqn:fhfa_error_congestion_2} \\
		& \abs{\sum_{\begin{subarray}{c}
					u:
					(u,v)\in E
			\end{subarray}}\ff(u,v)-\sum_{\begin{subarray}{c}
					w:
					(v,w)\in E
			\end{subarray}}\ff(v,w)}\leq \epsilon, ~ \forall v\in V\backslash\{s,t\}  \label{eqn:fhfa_error_demand}\\
		& \abs{\ff(v,w)-\ff(v',w')}\leq \epsilon, ~ \forall (v,w),(v',w')\in H_i, H_i\in\mathcal{H} \label{eqn:fhfa_error_homology}
	\end{align}
	or it correctly declares that the associated \textsc{fhf} instance is infeasible. We refer to the error in \eqref{eqn:fhfa_error_congestion_1} and \eqref{eqn:fhfa_error_congestion_2} as error in congestion, error in \eqref{eqn:fhfa_error_demand} as error in demand, and error in \eqref{eqn:fhfa_error_homology} as error in homology.
\end{definition}

We can use the same reduction method and solution mapping method in the exact case to reduce a \textsc{1-lena} instance to an \textsc{fhfa} instance. Note that, though we still obtain $\hat{\xx}$ by setting 
for each $j \in [\hat{n}]$,
\[\hat{\xx}(j)=\ff^h(\hat{x}^i_j), \qquad \text{for an arbitrary $i\in[\hat{m}]$ }\]

\begin{lemma}[\textsc{1-lena} to \textsc{fhfa}]
	\label{lm: 1lena-fhfa}
	Given a \textsc{1-lena} instance $(\hat{\AA},\hat{\bb},\hat{R},\epsilon^{1le})$ where $\hat{\AA}\in\Z^{\hat{m}\times\hat{n}}, \hat{\bb}\in\Z^{\hat{m}}$,
	we can reduce it to an \textsc{fhfa} instance $(G^h, F^h, \mathcal{H}^h=(H_1, \cdots, H_h), \uu^h, s,t, \epsilon^h)$ by letting
	\[\eps^h=\frac{\eps^{1le}}{5\hat{n}X(\hat{\AA},\hat{\bb})},\]
	and using Lemma \ref{lm: 1len-fhf} to construct an \textsc{fhf} instance $(G^h, F^h,  \mathcal{H}^h, \uu^h, s,t)$ from the \textsc{1-len} instance $(\hat{\AA},\hat{\bb},\hat{R})$.
	If $\ff^h$ is a solution to the \textsc{fhfa} (\textsc{fhf}) instance, then in time $O(\hat{n})$, we can compute a solution $\hat{\xx}$ to the \textsc{1-lena} (\textsc{1-len}, respectively) instance,
	where the exact case holds when $\eps^{h} = \eps^{1le}=0$.

\end{lemma}	
\begin{proof}
	
	Based on the solution mapping method described above, it takes constant time to set the value of each entry of $\hat{\xx}$ as $\ff^h$ on certain edges. As $\hat{\xx}$ has $\hat{n}$ entries, such a solution mapping takes $O(\hat{n})$ time.		
	
	In order to analyze errors, we firstly investigate the error in the $i$th linear equation in the \textsc{1-lena} instance $\hat{\aa}_i^\top\hat{\xx}=\hat{\bb}(i)$.
	Let $\hat{\xx}_i(j)=\ff^h(\hat{x}^j_i)$ denote the amount of flow routed through in the $i$th section of $G^h$. 
	For each edge $\hat{x}_i^j$ in the $i$th section of $G^h$, by error in homology defined in Eq. \eqref{eqn:fhfa_error_homology}, we have
	\[\abs{\hat{\xx}(j)-\hat{\xx}_i(j)}\leq \epsilon^h,\qquad \forall j\in[\hat{n}]\]
	Thus, we can bound $|\hat{\aa}_i^\top\hat{\xx}-\hat{\aa}_i^\top\hat{\xx}_i|$ by
	\begin{equation}
		\label{eq: LEN1}
		\begin{aligned}
			|\hat{\aa}_i^\top\hat{\xx}-\hat{\aa}_i^\top\hat{\xx}_i|=|\hat{\aa}_i^\top(\hat{\xx}-\hat{\xx}_i)|
			\le \norm{\hat{\aa}_i}_1\epsilon^h.
		\end{aligned}
	\end{equation}	
	
	Next, we bound $\hat{\aa}_i^\top\hat{\xx}_i$.
	We use $\ff^h(s, J_i^{\pm})$ to denote the total incoming flow of vertex $J_i^{\pm}$. Then, we have 
	\begin{itemize}
		\item $\abs{\ff^h(\hat{b}_i)-\hat{\bb}(i)}\in\left[-\eps^h,\eps^h\right]$ because of error in congestion of fixed flow edge $\hat{b}_i$ (defined in Eq. \eqref{eqn:fhfa_error_congestion_1}); 
		\item $\abs{\ff^h(s,J_i^{+})-(\ff^h(e_i^{+})+\ff^h(\hat{b}_i))}\leq \epsilon^{h}$ because of error in demand of vertex $J_i^+$ (defined in Eq. \eqref{eqn:fhfa_error_demand});
		\item $\abs{\ff^h(s,J_i^{-})-\ff^h(e_i^-)}\leq \epsilon^{h}$ because of error in demand of vertex $J_i^-$;
		\item $|\ff^h(e_i^+)-\ff^h(e_i^-)|\leq \epsilon^h$ because of error in homology between edge $e_i^+, e_i^-$.
	\end{itemize}
	Since $\hat{\aa}_i^\top\hat{\xx}_i=\ff^h(s,J_i^{+})-\ff^h(s,J_i^{-})$, the error of $\abs{\hat{\aa}_i^\top\hat{\xx}_i-\hat{\bb}(i)}$ is an accumulation of the above four errors, which gives
	\[\abs{\hat{\aa}_i^\top\hat{\xx}_i-\hat{\bb}(i)}\leq 4\eps^h.\]
Combining with Eq. \eqref{eq: LEN1}, we have
	\begin{align*}
		|\hat{\aa}_i^\top\hat{\xx}-\hat{\bb}(i)|&=|\hat{\aa}_i^\top\hat{\xx}-\hat{\aa}_i^\top\hat{\xx}_i+\hat{\aa}_i^\top\hat{\xx}_i-\hat{\bb}(i)|\\
		&\leq |\hat{\aa}_i^\top\hat{\xx}-\hat{\aa}_i^\top\hat{\xx}_i|+|\hat{\aa}_i^\top\hat{\xx}_i-\hat{\bb}(i)|\\
		&\leq (\norm{\hat{\aa}_i}_1+4)\epsilon^h
	\end{align*}

	We use $\tau^{1le}$ to denote the error of $\hat{\xx}$ that is obtained by mapping back $\ff^h$ with at most $\eps^{h}$ additive error. 
	Then we have
	\begin{align*}
		\tau^{1le} = \max_{i\in[\hat{m}]} |\hat{\aa}_i^\top\hat{\xx}-\hat{\bb}(i)|\leq \max_{i\in[\hat{m}]}\left\{(\norm{\hat{\aa}_i}_1+4)\epsilon^h\right\}\leq (\hat{n}X(\hat{\AA})+4)\epsilon^h\leq 5\hat{n}X(\hat{\AA},\hat{\bb})\epsilon^h.
	\end{align*}	
	
	As we set in the reduction that $\eps^h=\frac{\eps^{1le}}{5\hat{n}X(\hat{\AA},\hat{\bb})}$, then we have
	\[\tau^{1le}\leq 5\hat{n}X(\hat{\AA},\hat{\bb})\epsilon^h=\eps^{1le},\]
	indicating that $\hat{\xx}$ is a solution to the \textsc{1-lena} instance $(\hat{\AA},\hat{\bb},\hat{R},\epsilon^{1le})$.
\end{proof}


\subsection{FHF(A) to FPHF(A)}
\label{sect: fhfa-fphfa}
\subsubsection{FHF to FPHF}
We show the reduction from an \textsc{fhf} instance $(G^h, H^h, \mathcal{H}^h, \uu^h,  s,t)$ to an \textsc{fphf} instance $(G^p, F^p, \mathcal{H}^p, \uu^p, s,t)$. 
Suppose that $(v_1, w_1),\cdots, (v_k,w_k) \in E^h$ belong to a homologous edge set of size $k$ in $G^h$. 
As shown in Figure \ref{fig: h-ph}\footnote{We use $(0,u)$ for an non-fixed flow edge of capacity $u$.}, we replace $(v_i,w_i), i\in\{2, \cdots, k-1\}$ by two edges $(v_i, z_i)$ and $(z_i, w_i)$ such that $z_i$ is a new vertex incident only to these two edges, and edge capacities of the two new edges are the same as that of the original edge $(v_i,w_i)$. Then, we can construct $k-1$ pairs of homologous edges: $(v_1,w_1)$ and $(v_2,z_2)$; $(z_2, w_2)$ and $(v_3,z_3)$; $\cdots$; $(z_i, w_i)$ and $(v_{i+1},z_{i+1})$; $\cdots$; $(z_{k-1}, w_{k-1})$ and $(v_k,w_k)$.
In addition, no reduction is performed on non-homologous edges in $G^h$, and we trivially copy these edges to $G^p$.

\begin{figure}[ht]
	\centering
	\includegraphics[width=0.9\textwidth]{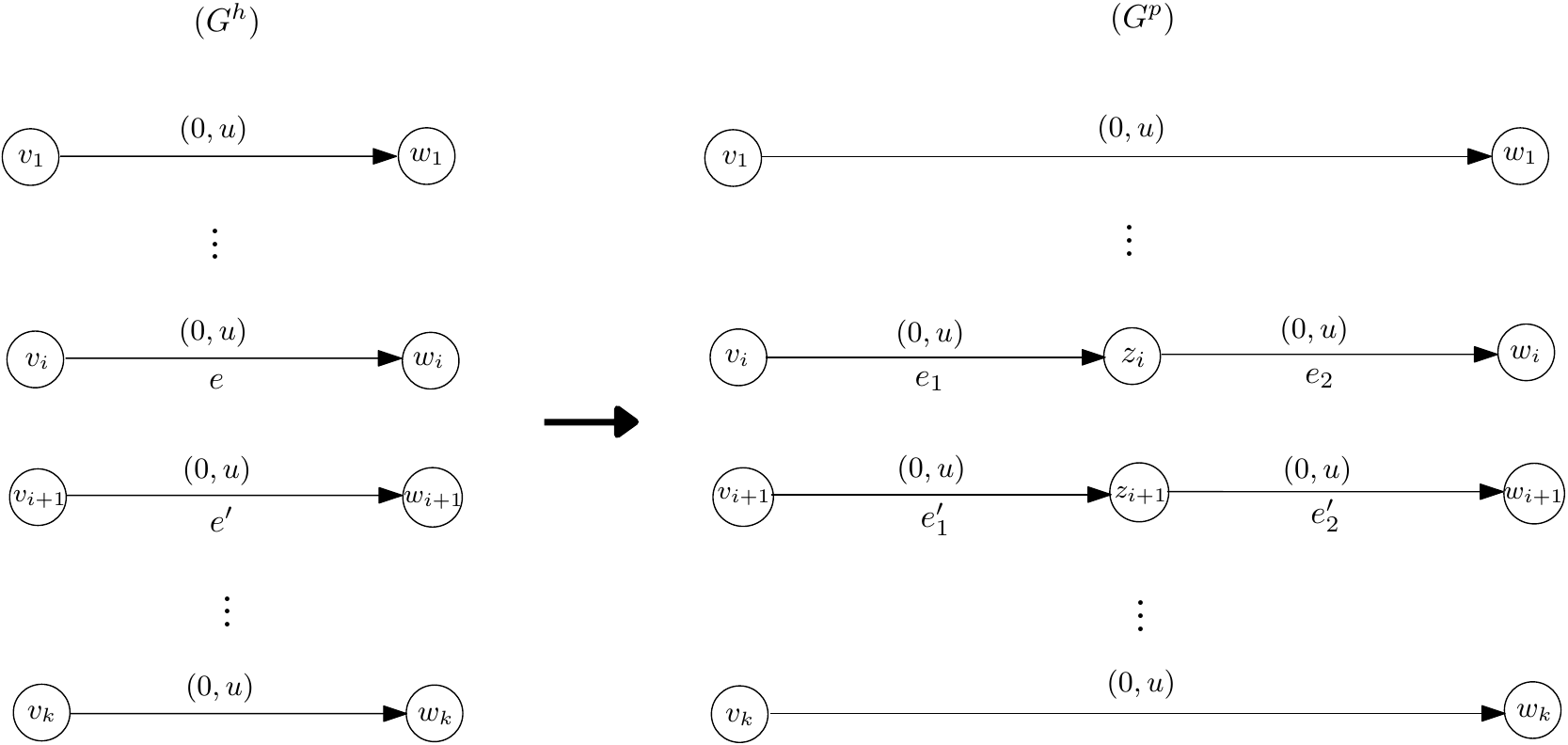}
	\caption{The reduction from \textsc{fhf} to \textsc{fphf}. } 
	\label{fig: h-ph}
\end{figure}

If an \textsc{fphf} solver returns $\ff^p$ for the \textsc{fphf} instance $(G^p, F^p,  \mathcal{H}^p, \uu^p,s,t)$, then we return $\ff^h$ for the \textsc{fhf} instance $(G^h, H^h, \mathcal{H}^h,\uu^h,  s,t)$, 
by setting $\ff^h(e)=\ff^p(e_1)$ if $e$ is a homologous edge 
and splits into two edges $e_1, e_2$ in the reduction such that $e_1$ precedes $e_2$,
and setting $\ff^h(e) = \ff^p(e)$ otherwise.
If the \textsc{fphf} solver returns ``infeasible'' for the \textsc{fphf} instance, 
then we return ``infeasible'' for the \textsc{fhf} instance.

\begin{lemma}[\textsc{fhf} to \textsc{fphf}]
	\label{lm: fhf-fphf}
	Given an \textsc{fhf} instance $(G^h, F^h, \mathcal{H}^h=(H_1, \cdots, H_h), \uu^h, s, t)$,
	we can construct, in time $O(|E^h|)$, an \textsc{fphf} instance $(G^p, F^p,  \mathcal{H}^p=(H_1, \cdots, H_p), \uu^p,s, t)$ such that
	\[|V^{p}|\leq|V^h|+|E^h|, \qquad |E^{p}|\leq2|E^h|, \qquad |F^p|=|F^h|, \qquad p \leq |E^h|, \qquad \norm{\uu^p}_{\max}=\norm{\uu^h}_{\max}.\]
	If the \textsc{fhf} instance has a solution, then the \textsc{fphf} instance has a solution.
\end{lemma}	
\begin{proof}
	According to the reduction described above, from any solution $\ff^h$ to the \textsc{fhf} instance, it is easy to derive a solution $\ff^p$ to the \textsc{fphf} instance. Concretely, for any homologous edge $e\in E^h$ that is split into two edges $e_1, e_2\in E^p$, we set
	$\ff^p(e_1)=\ff^p(e_2)=\ff^h(e)$.
	Since the vertex between $e_1$ and $e_2$ is only incident to these two edges, then the conservation of flows is satisfied on the inserted vertices. Moreover, since the edge capacity of $e_1$ and $e_2$ are the same as that of $e$, and they route the same amount of flows, thus the capacity constraint is also satisfied on the split edges. In addition, conservation of flows and capacity constraints are trivially satisfied for the rest vertices and edges. Therefore, $\ff^p$ is a feasible flow to the \textsc{fphf} instance.
	
	\vspace{8pt}
	Now, we track the change of problem size after reduction. Based on the reduction method, given an \textsc{fhf} instance with $|V^h|$ vertices, $|E^h|$ edges (including $|H^h|$ fixed flow edges, and $h$ homologous edge sets $\mathcal{H}^h=\{H_1, \cdots, H_h\}$), we can compute the size of the reduced \textsc{fphfa} instance as follows.
	\begin{enumerate}
		\item $|V^{p}|$ vertices. Since all vertices in $V^h$ are maintained, and for each homologous edge set $H_i\in \mathcal{H}^h$, $|H_i|-2$ new vertices are inserted, thus we have
		\begin{align*}
			|V^{p}|&= |V^h|+\sum_{i\in[h]}(|H_i|-2)\\
			&\leq |V^h|+|E^h|. &&\text{Because $\sum_{i\in[h]}|H_i|\leq |E^h|$}
		\end{align*}
		
		\item $|E^{p}|$ edges. Since all edges in $E^h$ are maintained and a new edge is generated by inserting a new vertex, thus we have
		\[|E^{p}|= |E^h|+\sum_{i\in[h]}(|H_i|-2)\leq 2|E^h|.\]
		
		\item $|F^p|$ fixed flow edges, where
		$|F^p|=|F^h|$.
		It is because all fixed flow edges in $F^h$ are maintained, and no new fixed flow edges are generated by reduction of this step.
		
		\item $p$ pairs of homologous edges. Since each homologous edge set $H_i\in\mathcal{H}^h$ can be transformed into $|H_i|-1$ pairs, then we have
		$p=\sum_{i\in[h]}(|H_i|-1)\leq |E^h|$.		
		
		\item The maximum edge capacity is bounded by
		$\norm{\uu^p}_{\max}=\norm{\uu^h}_{\max}$,
		because the reduction of this step only separate edges by inserting new vertices without modifying edge capacities.		
	\end{enumerate}
 	To estimate the reduction time, inserting a new vertex takes a constant time and there are $O(|E^h|)$ new vertices to be inserted. Also, it takes constant time to copy each of the rest $O(|E^h|)$ non-homologous edges. Hence, the reduction of this step takes $O(|E^h|)$ time. 
\end{proof}	

\subsubsection{FHFA to FPHFA}

\begin{definition}[FPHF Approximate Problem (\textsc{fphfa})]
	\label{def: fphfa-new}
	An \textsc{fphfa} instance is given by an \textsc{fphf} instance $\left(G, F,  \mathcal{H}=\{H_1, \cdots, H_p\}, \uu,s, t\right)$ as in Definition \ref{def: fphf}, and an error parameter $\epsilon\in[0,1]$, which we collect in a tuple $\left(G, F, \mathcal{H}, \uu, s,t, \epsilon\right)$.
	We say an algorithm solves the \textsc{fphfa} problem, if, given any \textsc{fphfa} instance, it returns a flow $\ff\geq\bf{0}$ that satisfies
	\begin{align}
		&\uu(e)-\eps\leq \ff(e)\leq\uu(e)+\eps, ~ \forall e \in F \label{eqn:fphfa_error_congestion_1} \\
		&0\leq \ff(e)\leq\uu(e)+\eps, ~ \forall e \in E\backslash F \label{eqn:fphfa_error_congestion_2} \\
		& \abs{\sum_{\begin{subarray}{c}
					u:
					(u,v)\in E
			\end{subarray}}\ff(u,v)-\sum_{\begin{subarray}{c}
					w:
					(v,w)\in E
			\end{subarray}}\ff(v,w)}\leq \epsilon, ~ \forall v\in V\backslash\{s,t\}  \label{eqn:fphfa_error_demand}\\
		& \abs{\ff(v,w)-\ff(y,z)}\leq \epsilon, ~ \forall H_i \ni (v,w), (y,z), i\in [p] \label{eqn:fphfa_error_homology}
	\end{align}
	or it correctly declares that the associated \textsc{fphf} instance is infeasible. We refer to the error in \eqref{eqn:fphfa_error_congestion_1} and \eqref{eqn:fphfa_error_congestion_2} as error in congestion, error in \eqref{eqn:fphfa_error_demand} as error in demand, and error in \eqref{eqn:fphfa_error_homology} as error in pair homology.
\end{definition}	

We can use the same reduction and solution mapping method in the exact case to the approximate case.

\begin{lemma}[\textsc{fhfa} to \textsc{fphfa}]
	\label{lm: fhfa-fphfa}
	Given an \textsc{fhfa} instance $(G^h, F^h, \mathcal{H}^h=(H_1,\cdots, H_h),\uu^h, s,t, \epsilon^h)$,
	we can reduce it to an \textsc{fphfa} instance $(G^{p}, F^{p},\mathcal{H}^p=(H_1, \cdots, H_p), \uu^{p}, s,t, \epsilon^{p})$ by letting
	\[\epsilon^p=\frac{\epsilon^h}{|E^h|},\]
	and using Lemma \ref{lm: fhf-fphf} to construct an \textsc{fphf} instance $(G^{p}, F^{p},\mathcal{H}^p, \uu^{p}, s,t)$ from the \textsc{fhf} instance $(G^h, H^h, \mathcal{H}^h,\uu^h, s,t)$.
	If $\ff^p$ is a solution to the \textsc{fphfa} (\textsc{fphf}) instance, then in time $O(\abs{E^h})$, we can compute a solution $\ff^h$ to the \textsc{fhfa} (\textsc{fhf}, respectively) instance,
	where the exact case holds when $\eps^{p} = \eps^{h}=0$.
	
\end{lemma}	
\begin{proof}

	Based on the solution mapping method described above, it takes constant time to set the value of each homologous or non-homologous entry of $\ff^h$, and $\ff^h$ has $|E^h|$ entries, such a solution mapping takes $O(|E^h|)$ time.

	Now, we conduct an error analysis. 
	We use $\tau^h$ to denote the error of $\ff^h$ that is obtained by mapping back $\ff^p$ with at most $\eps^p$ additive error. 
	By the error notions of \textsc{fhfa} (Definition \ref{def: fhfa-new}), there are three types of error to track after solution mapping: (1) error in congestion $\tau^h_u, \tau^h_l$; (2) error in demand $\tau^h_d$; (3) error in homology $\tau^h_h$. 
	Then, we set the additive error of \textsc{fhfa} as
	\[\tau^h=\max\{\tau^h_u, \tau^h_l, \tau^h_d, \tau^h_h\}.\]
	\begin{enumerate}
		\item Error in congestion.\\
		We first track the error of the upper bound of capacity $\tau^h_u$. 
		For any homologous edge $e$, by error in congestion defined in Eq. \eqref{eqn:fhfa_error_congestion_2} and Eq. \eqref{eqn:fphfa_error_congestion_2}, we have
		\[\ff^p(e_1)=\ff^h(e)\leq \uu^h(e)+\tau^h_u,\qquad \ff^p(e_1)\leq \uu^p(e_1)+\eps^p=\uu^h(e)+\eps^p.\]
		Thus, it suffices to set
		$\tau^h_u= \epsilon^p$.
		This $\tau^h_u$ also applies to the other edges since no reduction is made on them. 
		
		Next, we track the error of the lower bound of capacity $\tau^h_l$. %
		Since only fixed flow edges have lower bound of capacity, and no reduction is made on fixed flow edges, then we obtain trivially that $\tau^h_l= \epsilon^p$. 

		\item Error in demand.\\
		For simplicity, we denote $H^h=\bigcup_{i\in[h]}H_i$ as the set of all homologous edges. As defined in Eq. \eqref{eqn:fhfa_error_demand}, error in demand $\tau^h_d$ is computed as 
		\begin{scriptsize}
			\begin{equation}
				\label{eq: fhfa-fphfa-d}
				\begin{aligned}
					\tau^{h}_d&=\max_{w\in V^h\backslash\{s,t\}}\abs{\sum_{(v,w)\in E^h}\ff^h(v,w)-\sum_{(w,u)\in E^h}\ff^h(w,u)}\\
					&\overset{(1)}{=} \max_{w\in V^h\backslash\{s,t\}}\abs{\left(\sum_{(v,w)\in H^h}\ff^h(v,w)+\sum_{(v,w)\in E^h\backslash H^h}\ff^h(v,w)\right)-\left(\sum_{(w,u)\in H^h}\ff^h(w,u)+\sum_{(w,u)\in E^h\backslash H^h}\ff^h(w,u)\right)}\\
					&\overset{(2)}{=} \max_{w\in V^h\backslash\{s,t\}}\abs{\left(\sum_{e=(v,w)\in H^h}\ff^p(e_1)+\sum_{(v,w)\in E^h\backslash H^h}\ff^p(v,w)\right)-\left(\sum_{e=(w,u)\in H^h}\ff^p(e_1)+\sum_{(w,u)\in E^h\backslash H^h}\ff^p(w,u)\right)}\\
					&\overset{(3)}{\leq} \epsilon^p+\max_{w\in V^h\backslash\{s,t\}}\abs{\left(\sum_{e=(v,w)\in H^h}\ff^p(e_1)+\sum_{(v,w)\in E^h\backslash H^h}\ff^p(v,w)\right)-\left(\sum_{e=(v,w)\in H^h}\ff^p(e_2)+\sum_{(v,w)\in E^h\backslash H^h}\ff^p(v,w)\right)}\\
					&=\epsilon^p+ \max_{w\in V^h\backslash\{s,t\}}\abs{\sum_{e=(v,w)\in H^h}\ff^p(e_1)-\sum_{e=(v,w)\in H^h}\ff^p(e_2)}\\
					&=\epsilon^p+ \max_{w\in V^h\backslash\{s,t\}}\abs{\sum_{e=(v,w)\in  H^h}(\ff^p(e_1)-\ff^p(e_2))}\\
					&\leq \epsilon^p+\max_{w\in V^h\backslash \{s,t\}}\sum_{e=(v,w)\in  H^h}\abs{\ff^{p}(e_1)-\ff^{p}(e_2)}\\
					& \overset{(4)}{\leq}\epsilon^p+|H^h|\epsilon^p\\
					&\overset{(5)}{\leq} |E^h|\epsilon^p. 
				\end{aligned} 
			\end{equation}
		\end{scriptsize}

		For step (1), we separate homologous edges from other edges that are incident to vertex $w$.
		For step (2), we apply the rule of mapping $\ff^p$ back to $\ff^h$.
		For step (3), an example of a vertex $w\in V^h\backslash\{s,t\}$ is illustrated in Figure \ref{fig: demand-error}, where $w$ has an incoming homologous edge $(v,w)$ and an outgoing homologous edge $(w,u)$. We apply Eq. \eqref{eqn:fphfa_error_demand} in Definition \ref{def: fphfa-new} with respect to vertex $w$, so that we can replace the sum of outgoing flows of $w$ by the sum of its incoming flows with an error in demand of $G^p$ being introduced, i.e.,			
		\begin{figure}[ht]
				\centering
				\includegraphics[width=0.5\textwidth]{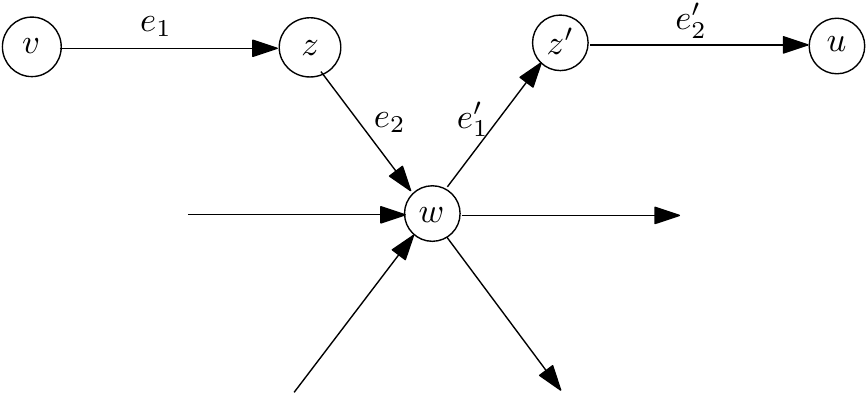}
				\caption{An example of the incoming and outgoing flows of vertex $w\in V^h\backslash\{s,t\}$ when get reduced to $G^p$. Suppose in $G^p$, $e=(v,w),e'=(w,u)$ are homologous edges incident to $w$, and the rest edges incident to $w$ are non-homologous.}
				\label{fig: demand-error}
			\end{figure}		
		\begin{small}
			\[\abs{\left(\sum_{e=(w,u)\in H^h}\ff^p(e_1)+\sum_{(w,u)\in E^h\backslash H^h}\ff^p(w,u)\right)-\left(\sum_{e=(v,w)\in H^h}\ff^p(e_2)+\sum_{(v,w)\in E^h\backslash H^h}\ff^p(v,w)\right)}\leq \epsilon^p.\]
		\end{small}
				
		For step (4), we apply Eq. \eqref{eqn:fphfa_error_demand} again with respect to the new inserted vertex. Since the new inserted vertex is only incident to $e_1$ and $e_2$, we have
		$\abs{\ff^p(e_1)-\ff^p(e_2)}\leq \epsilon^p$.
		For step (5), we utilize $|E^h|-|H^h|\geq 1$, i.e., there are more than one non-homologous edges in $G^h$.			

		\item Error in homology.\\
		As defined in Eq. \eqref{eqn:fhfa_error_homology}, error in homology $\tau^h_h$ is computed as
		\[\tau^h_h=\max_{\begin{subarray}{c}
				\forall e, e'\in H_i\\
				\forall H_i\in \mathcal{H}^h
			\end{subarray}
		} \abs{\ff^h(e)-\ff^h(e')}=\max_{\begin{subarray}{c}
			\forall e, e'\in H_i\\
			\forall H_i\in \mathcal{H}^h
		\end{subarray}
	} \abs{\ff^p(e_1)-\ff^h(e'_1)}.\]
		It is observed that error in homology in $G^h$ for a homologous edge set of size $k$ get accumulated by $(k-1)$ times of error in pair homology in $G^p$, and we have $k\leq |E_h|$. Thus, $\tau^h_h\leq |E^h|\eps^p$.
	\end{enumerate}
	
	Putting all together, we have
	\[\tau^h=\max\{\tau^h_u, \tau^h_l, \tau^h_d, \tau^h_h\}\leq |E^h|\eps^p.\] 	
	
	As we set in the reduction that $\eps^p=\frac{\eps^h}{|E^h|}$, then we have
	\[\tau^h\leq |E^h|\cdot\frac{\eps^h}{|E^h|}=\eps^h,\]
	indicating that $\ff^h$ is a solution to the \textsc{fhfa} instance $(G^h, F^h, \mathcal{H}^h,\uu^h, s,t, \epsilon^h)$.
\end{proof}


\subsection{FPHF(A) to SFF(A)}
\label{sect: fphfa-sffa}
\subsubsection{FPHF to SFF}
We show the reduction from an textsc{fphf} instance $(G^p, F^p, \mathcal{H}^p, \uu^p, s,t)$ to an \textsc{sff} instance $(G^s,F^s, S_1, S_2, \uu^s, s_1,t_1,s_2,t_2)$. 
Assume that $\{e,\hat{e}\}\in  \mathcal{H}^p$ is an arbitrary homologous edge set in $G^p$. 
As shown in Figure \ref{fig: phomo-slu2cf}, we map $\{e,\hat{e}\}$ in $G^p$ to a gadget consisting edges $\{e_1,e_2,e_3,e_4,e_5 = \hat{e}_3,
\hat{e}_1, \hat{e}_2,\hat{e}_4,\hat{e}_5\}$ in $G^s$. The key idea to remove the homologous requirement is to introduce a second commodity between a source-sink pair $(s_2,t_2)$. Concretely, we impose the fixed flow constraints on $(e_4, \hat{e}_4)$, the selective constraint of commodity 1 on $(e_1, e_2, \hat{e}_1,\hat{e}_2)$, and the selective constraint of commodity 2 on $(e_3, e_5 /\hat{e}_3,\hat{e}_5)$. Then, there is a flow of commodity 2 that routes through the directed path $e_3\rightarrow e_4\rightarrow e_5 /\hat{e}_3 \rightarrow \hat{e}_4\rightarrow \hat{e}_5$, and a flow of commodity 1 through paths $e_1 \rightarrow e_4\rightarrow e_2$ and $\hat{e}_1 \rightarrow \hat{e}_4 \rightarrow \hat{e}_2 $. The fixed flow constraint on $(e_4, \hat{e}_4)$ forces the flow of commodity 1 through the two paths to be equal, since the flows of commodity 2 on $(e_4, \hat{e}_4)$ are equal. Thus, the homologous requirement for edge $e$ and $\hat{e}$ is simulated.
In addition, as no reduction is performed on non-homologous edges in $G^p$, we trivially copy these edges to $G^s$, and restrict these edges to be selective for commodity 1.
\begin{figure}[ht]
	\centering
	\includegraphics[width=\textwidth]{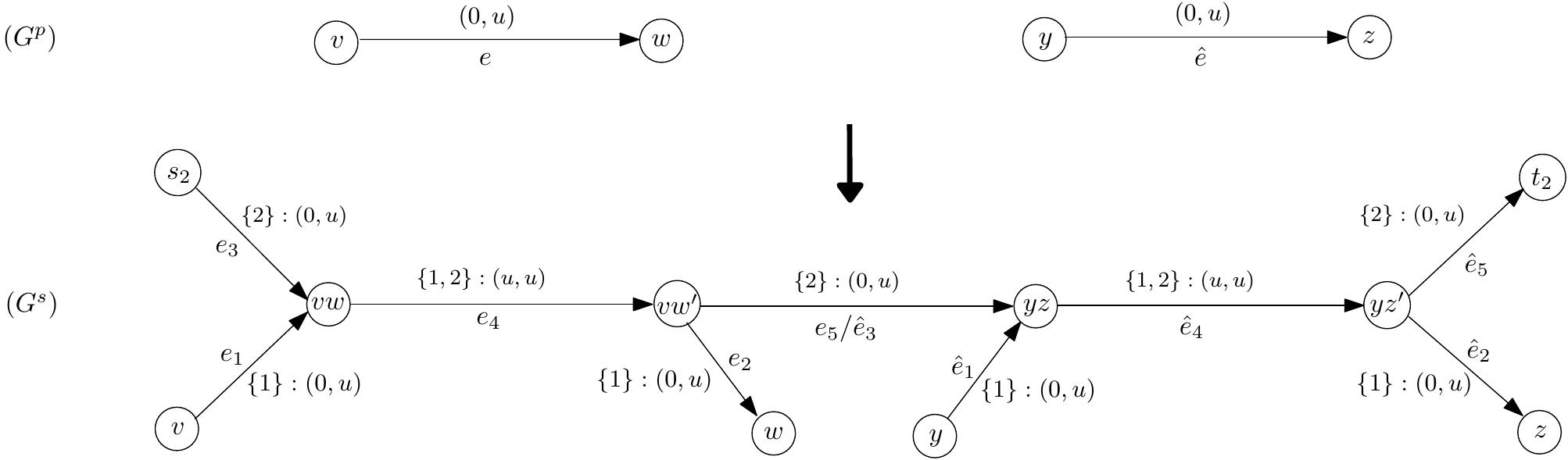}
	\caption{The reduction from \textsc{fphf} to \textsc{sff}. } 
	\label{fig: phomo-slu2cf}
\end{figure}

If an \textsc{sff} solver returns $\ff^s$ for the \textsc{sff} instance $(G^s,F^s, S_1, S_2, \uu^s, s_1,t_1,s_2,t_2)$, then we return $\ff^p$ for the \textsc{fphf} instance $(G^p, F^p, \mathcal{H}^p, \uu^p, s,t)$ by the following method:
for each pair of homologous edges $e$ and $\hat{e}$, we set $\ff^p(e)=\ff^s_1(e_1), \ff^p(\hat{e})=\ff^s_1(\hat{e}_1)$; for each non-homologous edge $e'$, we set $\ff^p(e')=\ff^s_1(e')$. Note that $\ff^s$ is a two-commodity flow\footnote{We denote $\ff=\ff_1+\ff_2$ for two-commodity flows.} while $\ff^p$ is a single-commodity flow.
If the \textsc{sff} solver returns ``infeasible'' for the \textsc{sff} instance, 
then we return ``infeasible'' for the \textsc{fphf} instance.

\begin{lemma}[\textsc{fphf} to \textsc{sff}]
	\label{lm: fphf-sff}
	Given an \textsc{fphf} instance $(G^p, F^p, \mathcal{H}^p=(H_1, \cdots, H_p),\uu^p,  s, t)$,
	we can construct, in time $O(|E^p|)$, an \textsc{sff} instance $(G^s,F^s, S_1, S_2, \uu^s, s_1,t_1,s_2,t_2)$ such that
	\[|V^s|=|V^{p}|+4p+2,\qquad |E^s|=|E^{p}|+7p,\qquad |F^s|=|F^p|+2p, \]
	\[ |S_1|=|E^{p}|+2p,\qquad |S_2|=3p, \qquad \norm{\uu^s}_{\max}=\norm{\uu^p}_{\max}.\]
	If the \textsc{fphf} instance has a solution, then the \textsc{sff} instance has a solution.
\end{lemma}	
\begin{proof}
	According to the reduction described above, from any solution $\ff^p$ to the \textsc{fphf} instance, it is easy to derive a solution $\ff^s$ to the \textsc{sff} instance. Concretely, we define a feasible flow $\ff^s$ as follows:
	\begin{itemize}
		\item For any pair of homologous edge $\{e,\hat{e}\}\in  \mathcal{H}^p$, in its corresponding gadget in $G^s$, we set
		\[\ff^s_1(e_1)=\ff^s_1(e_4)=\ff^s_1(e_2)=\ff^p(e)=\ff^p(\hat{e})=\ff^s_1(\hat{e}_1)=\ff^s_1(\hat{e}_4)=\ff^s_1(\hat{e}_2) \leq u,\]
		where $u=\uu^p(e)=\uu^p(\hat{e})$, and set
		\[\ff^s_2(e_1)=\ff^s_2(e_2)=\ff^s_2(\hat{e}_1)=\ff^s_2(\hat{e}_2)=0.\]
		It is obvious that $\ff^s$ satisfies the selective constraint and the capacity constraint on edges $e_1, e_2, \hat{e}_1,\hat{e}_2$, and satisfies conservation of flows for commodity 1 on vertices $vw, vw', yz, yz'$.
		
		Moreover, we set
		\[\ff^s_2(e_3)=\ff^s_2(e_4)=\ff^s_2(e_5)=\ff^s_2(\hat{e}_4)=\ff^s_2(\hat{e}_5)=u-\ff^p(e)\leq u,\]
		\[\ff^s_1(e_3)=\ff^s_1(e_4)=\ff^s_1(e_5)=\ff^s_1(\hat{e}_4)=\ff^s_1(\hat{e}_5)=0.\]
		Then it is obvious that $\ff^s$ satisfies the selective constraint and the capacity constraint on edges $e_3, e_5, \hat{e}_5$, and satisfies the flow conservation constraint of commodity 2 on vertices $vw, vw', yz, yz'$.
		
		It remains to verify if $\ff^s$ satisfies the fixed flow constraint on edges $e_4, \hat{e}_4$. According to the above constructions, we have
		(we abuse the notation to also let $\ff^s = \ff^s_1 + \ff^s_1$)
		\[\ff^s(e_4)=\ff^s_1(e_4)+\ff^s_2(e_4)=\ff^p(e)+(u-\ff^p(e))=u,\]
		\[\ff^s(\hat{e}_4)=\ff^s_1(\hat{e}_4)+\ff^s_2(\hat{e}_4)=\ff^p(\hat{e})+(u-\ff^p(\hat{e}))=u.\]
		
		\item For any non-homologous edge $e'\in E^p$, we also have $e'\in E^s$ since no reduction is made on this edge, and we copy it trivially to $G^s$. We set
		\[\ff^s_1(e')=\ff^p(e') ,\qquad \ff^s_2(e')=0.\]
		Since $\ff^p$ is a feasible flow in $G^p$, it is easy to check that $\ff^s$ also satisfies the selective constraint for commodity 1 and the capacity constraint on non-homologous edges, as well as conservation of flows on vertices incident to non-homologous edges. 
	\end{itemize}
	To conclude, $\ff^s$ is a feasible flow to the \textsc{sff} instance.
	\\
	\\
	Now, we track the change of problem size after reduction. Based on the reduction method, given an \textsc{fphf} instance with $|V^p|$ vertices, $|E^p|$ edges (including $|F^p|$ fixed flow edges, and $p$ pairwise homologous edge sets $\mathcal{H}^p=\{H_1,\cdots, H_p\}$), we can compute the size of the reduced \textsc{sff} instance as follows.
	\begin{enumerate}
		\item $|V^s|$ vertices. To bound it, first, all vertices in $V^p$ are maintained. 
		Then, for each pair of homologous edges $(v,w)$ and $(y,z)$, 4 auxiliary vertices $vw, vw',yz,yz'$ are added. 
		And finally, a source-sink pair of commodity 2 $(s_2, t_2)$ is added. Therefore, we have
		$|V^s|=|V^{p}|+4p+2$.
		
		\item $|E^s|$ edges. Since non-homologous edges are maintained, 
		and each pair of homologous edges is replaced by a gadget with 9 edges, thus we have
		$|E^s|=(|E^p|-2p)+9p=|E^p|+7p$.
		
		\item $|F^s|$ fixed flow edges. First, all fixed flow edges in $F^p$ are maintained since by Def. \ref{def: fphf}, fixed flow edges and homologous edges are disjoint, thus no reduction is made.
		Then, for each pair of homologous edges as shown in Figure \ref{fig: phomo-slu2cf}, the reduction generate 2 fixed flow edges (i.e., $e_4, \hat{e}_4$). Hence, we have
		$|F^s|=|F^p|+2p$.
				
		\item $|S_i|$ edges that select the $i$-th commodity. First, all non-homologous edges in $E^p$ are selective for the commodity 1.
		Then, for each pair of homologous edges as shown in Figure \ref{fig: phomo-slu2cf}, the reduction generates 4 edges selecting commodity 1 (i.e., $e_1, e_2, \hat{e}_1,\hat{e}_2$) and 3 edges selecting commodity 2 (i.e., $e_3, e_5/\hat{e}_3, \hat{e}_5$). Thus, we have
		$|S_1|=(|E^{p}|-2p)+4p=|E^p|+2p$ and $|S_2|=3p$.

		\item The maximum edge capacity is bounded by
		$\norm{\uu^s}_{\max}=\norm{\uu^p}_{\max}$.
		First, capacity of non-homologous edges is unchanged in $G^s$ since no reduction is made.
		Then, for each pair of homologous edges $\{e, \hat{e}\}$ with capacity $\uu^p(e) (=\uu^p(\hat{e}))$, the capacity of the corresponding 9 edges are either selective edges with capacity $\uu^p(e)$ or fixed flow edges with fixed flow $\uu^p(e)$.
		
	\end{enumerate}
	To estimate the reduction time, first, it takes constant time to reduce each pair of homologous edges since only a constant number of vertices and edges are added. Also, it takes constant time to copy each of the rest non-homologous edges, then the reduction of this step takes $O(|E^p|)$ time in total.
\end{proof}	

\subsubsection{FPHFA to SFFA}

\begin{definition}[SFF Approximate Problem (\textsc{sffa})]
	\label{def: sffa-new}
	An \textsc{sffa} instance is given by an \textsc{sff} instance $(G,F, S_1, S_2, \uu, s_1,t_1,s_2,t_2)$ as in Definition \ref{def: sff}, and an error parameter $\epsilon \in[0,1]$, which we collect in a tuple $\left(G,F, S_1, S_2, \uu, s_1,t_1,s_2,t_2, \epsilon\right)$.
	We say an algorithm solves the \textsc{sffa} problem, if, given any \textsc{sffa} instance, it returns a pair of flows $\ff_1, \ff_2\geq\bf{0}$ that satisfies
	\begin{align}
		&\uu(e)-\eps\leq \ff_1(e)+\ff_2(e)\leq\uu(e)+\eps, ~~ \forall e \in F \label{eqn:sffa_error_congestion_1} \\
		&0\leq \ff_1(e)+\ff_2(e)\leq\uu(e)+\eps, ~~ \forall e \in E\backslash F \label{eqn:sffa_error_congestion_2} \\
		& \abs{\sum_{\begin{subarray}{c}
					u:
					(u,v)\in E
			\end{subarray}}\ff_i(u,v)-\sum_{\begin{subarray}{c}
					w:
					(v,w)\in E
			\end{subarray}}\ff_i(v,w)}\leq \eps, ~~ \forall v\in V\backslash \{s_i,t_i\},~ i\in\{1,2\} \label{eqn:sffa_error_demand}\\
		& \ff_{\bar{i}}(e) \leq \eps, ~~ \forall e\in S_{i},~ \bar{i}=\{1,2\}\backslash i \label{eqn:sffa_error_type}
	\end{align}
	or it correctly declares that the associated \textsc{sff} instance is infeasible. We refer to the error in \eqref{eqn:sffa_error_congestion_1} and \eqref{eqn:sffa_error_congestion_2} as error in congestion, error in \eqref{eqn:sffa_error_demand} as error in demand, and error in \eqref{eqn:sffa_error_type} as error in type.
\end{definition}	

We can use the same reduction and solution mapping method in the exact case to the approximate case.

\begin{lemma}[\textsc{fphfa} to \textsc{sffa}]
	\label{lm: fphfa-sffa}
	Given an \textsc{fphfa} instance $(G^{p}, F^{p},\mathcal{H}^p=(H_1, \cdots, H_p), \uu^{p},  s,t, \epsilon^{p})$, 
	we can reduce it to an \textsc{sffa} instance $\left(G^s,F^s, S_1, S_2, \uu^s, s_1,t_1,s_2,t_2, \epsilon^s\right)$ by letting 
	\[\eps^s=\frac{\eps^p}{11|E^p|},\]
	and using Lemma \ref{lm: fphf-sff} to construct an \textsc{sff} instance $\left(G^s,F^s, S_1, S_2, \uu^s, s_1,t_1,s_2,t_2\right)$ from the \textsc{fphf} instance $(G^{p}, F^{p},\mathcal{H}^p, \uu^{p}, s,t)$.
	If $\ff^s$ is a solution to the \textsc{sffa} (\textsc{sff}) instance, then in time $O(\abs{E^p})$, we can compute a solution $\ff^p$ to the \textsc{fphfa} (\textsc{fphf}, respectively) instance,
	where the exact case holds when $\eps^{s} = \eps^{p}=0$.
\end{lemma}	
\begin{proof}
	
	Based on the solution mapping method described above, it takes constant time to set the value of each homologous or non-homologous entry of $\ff^p$. Since $\ff^p$ has $|E^p|$ entries, such a solution mapping takes $O(|E^p|)$ time.		

	Now, we conduct an error analysis. 
	We use $\tau^p$ to denote the error of $\ff^p$ that is obtained by mapping back $\ff^s$ with at most $\eps^s$ additive error. 
	By the error notions of \textsc{fphpa} (Definition \ref{def: fphfa-new}), there are three types of error to track after solution mapping: (1) error in congestion $\tau^p_u, \tau^p_l$; (2) error in demand $\tau^p_d$; (3) error in pair homology $\tau^p_h$. Then, we set the additive error of \textsc{fphfa} as
	\[\tau^p=\max\{\tau^p_u, \tau^p_l,\tau^p_d, \tau^p_h\}.\]
	\begin{enumerate}
		\item Error in congestion.\\
		We first track the error of the upper bound of capacity $\tau^p_u$. For any pair of homologous edges $e$ and $\hat{e}$ with capacity $\uu^p(e)=\uu^p(\hat{e})$, 
		by error in congestion defined in Eq. \eqref{eqn:fphfa_error_congestion_2} and Eq. \eqref{eqn:sffa_error_congestion_2}, we have
		\[\ff_1^s(e_1)=\ff^p(e)\leq \uu^p(e)+\tau^p_u, \qquad  \ff_1^s(e_1)\leq\ff^s(e_1)\leq \uu^s(e_1)+\eps^s=\uu^p(e)+\eps^s;\]
		\[\ff_1^s(\hat{e}_1)=\ff^p(\hat{e})\leq \uu^p(\hat{e})+\tau^p_u, \qquad  \ff_1^s(\hat{e}_1)\leq\ff^s(\hat{e}_1)\leq \uu^s(\hat{e}_1)+\eps^s=\uu^p(\hat{e})+\eps^s.\]
		Thus, it suffices to set $\tau^p_u= \epsilon^s$. This bound also applies for non-homologous edges since no reduction is conducted.

		Next, we track the error of the lower bound of capacity $\tau^p_l$. And we only need to take fixed flow edges (thus non-homologous edges) into consideration. For any fixed flow edge $e'\in F^p$, it is copied to an edge selecting commodity 1 in $G^s$.
		By error in congestion defined in Eq. \eqref{eqn:fphfa_error_congestion_1}, we have
		\[\ff^s_1(e')\geq \uu^p(e')-\tau^p_l,\]
		and
		\begin{align*}
			\ff^s_1(e')&\geq \ff^s(e')-\epsilon^s && \text{By error in type defined in Eq. \eqref{eqn:sffa_error_type}}\\
			&\geq \uu^s(e')-2\epsilon^s && \text{By error in congestion defined in Eq. \eqref{eqn:sffa_error_congestion_1}}\\
			&=\uu^p(e')-2\epsilon^s 
		\end{align*}			
		Thus, it suffices to set $\tau^{p}_l= 2\epsilon^s$.

		\item Error in demand.\\	
		For simplicity, we denote $H^p=\bigcup_{i\in[p]}H_i$ as the set of all homologous edges. We can use the similar strategy in Eq. \eqref{eq: fhfa-fphfa-d}. As defined in Eq. \eqref{eqn:fphfa_error_demand}, the error in demand $\tau^p_d$ is computed as 
		\begin{scriptsize}
			\begin{equation}
				\label{eq: ph-slu}
				\begin{aligned}
					\tau^{p}_d&=\max_{w\in V^p\backslash\{s,t\}}\abs{\sum_{(v,w)\in E^p}\ff^p(v,w)-\sum_{(w,u)\in E^p}\ff^p(w,u)}\\
					&= \max_{w\in V^p\backslash\{s,t\}}\abs{\left(\sum_{(v,w)\in H^p}\ff^p(v,w)+\sum_{(v,w)\in E^p\backslash H^p}\ff^p(v,w) \right)-\left(\sum_{(w,u)\in H^p}\ff^p(w,u)+\sum_{(w,u)\in E^p\backslash H^p}\ff^p(w,u)\right)}\\
					&= \max_{w\in V^p\backslash\{s,t\}}\abs{\left(\sum_{e=(v,w)\in H^p}\ff^s_1(e_1)+\sum_{(v,w)\in E^p\backslash H^p}\ff^s_1(v,w)\right)-\left(\sum_{e=(w,u)\in H^p}\ff^s_1(e_1)+\sum_{(w,u)\in E^p\backslash H^p}\ff^s_1(w,u)\right)}\\
					&\leq \epsilon^s+\max_{w\in V^p\backslash\{s,t\}}\abs{\left(\sum_{e=(v,w)\in H^p}\ff^s_1(e_1)+\sum_{(v,w)\in E^p\backslash H^p}\ff^s_1(v,w)\right)-\left(\sum_{e=(v,w)\in H^p}\ff^s_1(e_2)+\sum_{(v,w)\in E^p\backslash H^p}\ff^s_1(v,w)\right)}\\
					&=\epsilon^s+ \max_{w\in V^p\backslash\{s,t\}}\abs{\sum_{e=(v,w)\in H^p}\ff^{s}_1(e_1)-\sum_{e=(v,w)\in H^p}\ff^{s}_1(e_2)}\\
					&=\epsilon^s+ \max_{w\in V^p\backslash\{s,t\}}\abs{\sum_{e=(v,w)\in H^p}(\ff^{s}_1(e_1)-\ff^{s}_1(e_2))}\\
					&\leq \epsilon^s+\max_{w\in V^p\backslash \{s,t\}}\sum_{e=(v,w)\in H^p}\abs{\ff^{s}_1(e_1)-\ff^{s}_1(e_2)}.
				\end{aligned}
			\end{equation}
		\end{scriptsize}

		Now, we try to bound $\abs{\ff^{s}_1(e_1)-\ff^{s}_1(e_2)}$. By error in demand of vertex $vw$ (Eq. \eqref{eqn:sffa_error_demand}), since $vw$ is only incident to $e_1, e_3, e_4$, we have
		\begin{equation}
			\label{eq: ph-slu-vw}
			\abs{\ff^s_1(e_1)+\ff^s_1(e_3)-\ff^s_1(e_4)}\leq \epsilon^{s}.
		\end{equation}
		Similarly, by error in demand of vertex $vw'$, we have
		\begin{equation}
			\label{eq: ph-slu-vw'}
			\abs{\ff^s_1(e_2)+\ff^s_1(e_5)-\ff^s_1(e_4)}\leq \epsilon^{s}.
		\end{equation}
		
		Combining Eq. \eqref{eq: ph-slu-vw} and Eq. \eqref{eq: ph-slu-vw'} gives
		\begin{align*}
			\abs{(\ff^s_1(e_1)+\ff^s_1(e_3))-(\ff^s_1(e_2)+\ff^s_1(e_5))}\leq 2\epsilon^{s}.
		\end{align*}
		Moreover, we can also lower bound its left hand side by
		\begin{align*}
			\abs{(\ff^s_1(e_1)+\ff^s_1(e_3))-(\ff^s_1(e_2)+\ff^s_1(e_5))}&=\abs{(\ff^s_1(e_1)-\ff^s_1(e_2))+(\ff^s_1(e_3)-\ff^s_1(e_5))}\\
			&\geq\abs{\ff^s_1(e_1)-\ff^s_1(e_2)}-\abs{\ff^s_1(e_3)-\ff^s_1(e_5)}.
		\end{align*}
		By rearranging, $\abs{\ff^s_1(e_1)-\ff^s_1(e_2)}$ can be upper bounded by
		\begin{equation}
			\label{eq: ph-slu-2}
			\begin{aligned}
				&\abs{\ff^s_1(e_1)-\ff^s_1(e_2)}\\
				\leq & \abs{\ff^s_1(e_3)-\ff^s_1(e_5)}+\abs{(\ff^s_1(e_1)+\ff^s_1(e_3))-(\ff^s_1(e_2)+\ff^s_1(e_5))}\\
				\leq& \abs{\ff^s_1(e_3)-\ff^s_1(e_5)}+2\epsilon^s \\
				\leq& \max\{\ff^s_1(e_3), \ff^s_1(e_5)\}+2\epsilon^s && \text{Because $\ff^s_1(e_3), \ff^s_1(e_5)\geq 0$}\\
				\leq& 3\epsilon^s
			\end{aligned}
		\end{equation}
		where the last inequality comes from error in type for $e_3$ and $e_5$, since they are both selective edges for commodity 2.
		
		Finally, applying the upper bound of $\abs{\ff^s_1(e_1)-\ff^s_1(e_2)}$ Eq. \eqref{eq: ph-slu-2} back to Eq. \eqref{eq: ph-slu}, we obtain $\tau_d^p\leq \epsilon^s+3|H^p|\epsilon^s\leq 3|E^p|\epsilon^s.$

		\item Error in pair homology.\\
		As defined in Eq. \eqref{eqn:fphfa_error_homology}, error in pair homology $\tau^p_h$ is computed as
		\begin{equation}
			\label{eq: ph-ph}
			\begin{aligned}
				\tau^p_h=\max_{\mathcal{H}^p\ni H_i=\{(v,w),(y,z)\}} \abs{\ff^p(v,w)-\ff^p(y,z)}=\max_{\mathcal{H}^p\ni H_i=\{e=(v,w),\hat{e}=(y,z)\}} \abs{\ff^s_1(e_1)-\ff^s_1(\hat{e}_1)}.
			\end{aligned}
		\end{equation} 

		Now, we try to bound $\abs{\ff^s_1(e_1)-\ff^s_1(\hat{e}_1)}$ for an arbitrary pair of homologous edges $e,\hat{e}$ with capacity $\uu^p(e) (=\uu^p(\hat{e}))$.
		By error in congestion on $e_4$ and $\hat{e}_4$, we have 
			\[\abs{\ff^s(e_4)-\ff^s(\hat{e}_4)}\leq 2\epsilon^s.\]
		By error in demand on vertices $vw$ and $yz$ for two commodities, we have
		\[\abs{\ff^s(e_1)+\ff^s(e_3)-\ff^s(e_4)}\leq 2\epsilon^{s},\]
		\[\abs{\ff^s(\hat{e}_1)+\ff^s(\hat{e}_3)-\ff^s(\hat{e}_4)}\leq 2\epsilon^{s}.\]
		Combining the above three inequalities, we have
		\begin{align*}
			\abs{(\ff^s(e_1)+\ff^s(e_3))-(\ff^s(\hat{e}_1)+\ff^s(\hat{e}_3))}\leq 6\epsilon^{s}.
		\end{align*}
		Again, we can also lower bound its left hand side by 
		\begin{align*}
			&\abs{(\ff^s(e_1)+\ff^s(e_3))-(\ff^s(\hat{e}_1)+\ff^s(\hat{e}_3))}\\
			=&\abs{(\ff^s_1(e_1)+\ff^s_2(e_1)+\ff^s_1(e_3)+\ff^s_2(e_3))-(\ff^s_1(\hat{e}_1)+\ff^s_2(\hat{e}_1)+\ff^s_1(\hat{e}_3)+\ff^s_2(\hat{e}_3))}\\
			\geq& \abs{\ff^s_1(e_1)-\ff^s_1(\hat{e}_1)}-\abs{(\ff^s_2(e_1)-\ff^s_2(\hat{e}_1))+(\ff^s_1(e_3)-\ff^s_1(\hat{e}_3))+(\ff^s_2(e_3)-\ff^s_2(\hat{e}_3))}.
		\end{align*}
		By rearranging, we can upper bound $\abs{\ff^s_1(e_1)-\ff^s_1(\hat{e}_1)}$ by
		\begin{small}
			\begin{align*}
				&\abs{\ff^s_1(e_1)-\ff^s_1(\hat{e}_1)}\\
				\leq&\abs{(\ff^s(e_1)+\ff^s(e_3))-(\ff^s(\hat{e}_1)+\ff^s(\hat{e}_3))} +\abs{(\ff^s_2(e_1)-\ff^s_2(\hat{e}_1))+(\ff^s_1(e_3)-\ff^s_1(\hat{e}_3))+(\ff^s_2(e_3)-\ff^s_2(\hat{e}_3))}\\
				\leq& 6\epsilon^{s}+\abs{\ff^s_2(e_1)-\ff^s_2(\hat{e}_1)}+\abs{\ff^s_1(e_3)-\ff^s_1(\hat{e}_3)}+\abs{\ff^s_2(e_3)-\ff^s_2(\hat{e}_3)},
			\end{align*}
		\end{small}
		where we have 
		\begin{itemize}
			\item $\abs{\ff^s_2(e_1)-\ff^s_2(\hat{e}_1)}\leq \max\{\ff^s_2(e_1),\ff^s_2(\hat{e}_1)\}\leq \epsilon^{s}$ because of error in type in $G^s$ and $e_1,\hat{e}_1$ are selective edges for commodity 1;
			\item $\abs{\ff^s_1(e_3)-\ff^s_1(\hat{e}_3)}\leq \max\{\ff^s_1(e_3), \ff^s_1(\hat{e}_3)\}\leq \epsilon^{s}$ because of error in type in $G^s$ and $e_3,\hat{e}_3$ are selective edges for commodity 2;
			\item $\abs{\ff^s_2(e_3)-\ff^s_2(\hat{e}_3)} \leq 3\epsilon^s$, which can be obtained directly by symmetry from the bound of $\abs{\ff^s_1(e_1)-\ff^s_1(e_2)}$ in Eq. \eqref{eq: ph-slu-2}.
		\end{itemize}
		
		Putting all together, we have
		\begin{align*}
			\abs{\ff^s_1(e_1)-\ff^s_1(\hat{e}_1)}&\leq 11\epsilon^s.
		\end{align*}
		Thus, by Eq. \eqref{eq: ph-ph}, we have $\tau^p_h\leq 11\epsilon^s$.
	\end{enumerate}
	
	To summarize, we can set
	\[\tau^p=\max\{\tau^p_u, \tau^p_l,\tau^p_d, \tau^p_h\}\leq 11|E^p|\eps^s.\]
	
	As we set in the reduction that $\eps^s=\frac{\eps^p}{11|E^p|}$, then we have
	\[\tau^{p}\leq  11|E^p|\cdot\frac{\eps^p}{11|E^p|}=\eps^p,\]
	indicating that $\ff^p$ is a solution to the \textsc{fphfa} instance $(G^{p}, F^{p},\mathcal{H}^p, \uu^{p},  s,t, \epsilon^{p})$.
\end{proof}


\subsection{SFF(A) to 2CFF(A)}
\label{sect: sffa-2cffa}
\subsubsection{SFF to 2CFF}
We show the reduction from an \textsc{sff} instance $(G^s,F^s, S_1, S_2, \uu^s, s_1,t_1,s_2,t_2)$ to a \textsc{2cff} instance $(G^f, F^f, \uu^f, s_1,t_1, s_2,t_2)$. 
Assume that $e\in S_i$ is an arbitrary selective edge for commodity $i$ in $G^p$. As shown in Figure \ref{fig: slu2cf-lu2cf}, we map $e$ in $G^s$ to a gadget consisting of edges $\{e_1,e_2,e_3,e_4,e_5\}$ in $G^f$. 
Note that a selective edge $e$ can be either a fixed flow edge or a non-fixed flow edge.
In Figure \ref{fig: slu2cf-lu2cf},   $l=u$ if $e$ is a fixed flow edge 
and $l=0$ if $e$ is a non-fixed flow edge.
Moreover, no reduction is performed on non-selective edges in $G^s$, and we trivially copy these edges to $G^f$. 

The key idea to remove the selective requirement is utilizing edge directions and the source-sink pair $(s_i, t_i)$ to simulate a selective edge $e$ for commodity $i$. 
More specifically, in the gadget, the flow of commodity $i$ routes through three directed paths: (1) $e_1\rightarrow e_4$, (2) $e_5\rightarrow e_2\rightarrow e_4$, (3) $e_5\rightarrow e_3$. 
The selective requirement is realized because $e_4$ is the only outgoing edge of $xy$ and only flow of commodity $i$ is allowed in $e_4$ (since its tail is $t_i$), thus in $e_1$. Similarly, $e_5$ is the only incoming edge of $xy'$ and only flow of commodity $i$ is allowed in $e_5$, thus in $e_3$.
In addition, to ensure that $e_1$ and $e_3$ route the same amount of flow, flows in $e_4$ and $e_5$ must be equal by the conservation of flows. 
Therefore, we impose the fixed flow constraint on $e_4$ and $e_5$ by setting the fixed flow to be $u$. 
We remove $e_2$ if $e$ is a fixed flow edge (in which case $e_2$ has capacity 0), 
and set capacity of $e_2$ to be $u$ otherwise (in which case $u-l =u- 0 = u$). 

\begin{figure}[ht]
	\centering
	\includegraphics[width=0.85\textwidth]{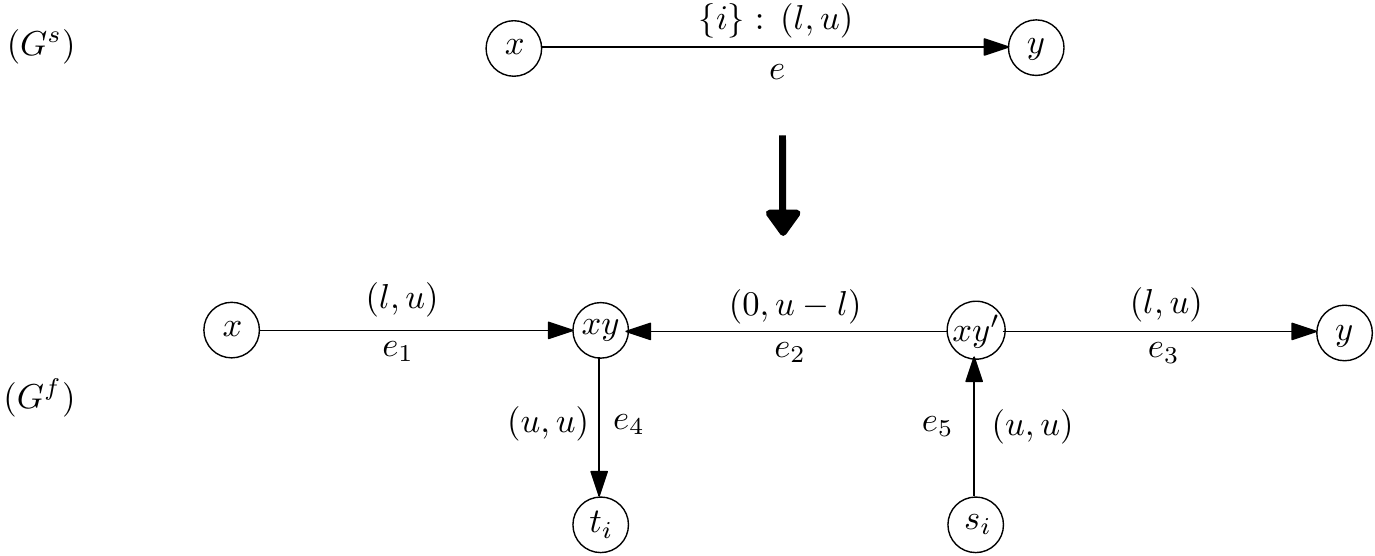}
	\caption{The reduction from \textsc{sff} to \textsc{2cff}. $l=u$ if $e$ is a fixed flow edge, and $l=0$ if $e$ is a non-fixed flow edge.} 
	\label{fig: slu2cf-lu2cf}
\end{figure}

If a \textsc{2cff} solver returns $\ff^f$ for the \textsc{2cff} instance $(G^f, F^f, \uu^f, s_1,t_1, s_2,t_2)$, then we return $\ff^s$ for the \textsc{sff} instance $(G^s,F^s, S_1, S_2, \uu^s, s_1,t_1,s_2,t_2)$ by the following method:
for any selective edge $e\in S_i$, we set $\ff^s_i(e)=\ff^f_i(e_1), \ff^s_{\bar{i}}(e)=\ff^f_{\bar{i}}(e_1)=0, \bar{i}\in\{1,2\}\backslash i$; and for any non-selective edge $e\in E^s\backslash (S_1\bigcup S_2)$, we map back trivially by setting $\ff^s_i(e)=\ff^f_i(e), i\in\{1,2\}$. 
If the \textsc{2cff} solver returns ``infeasible'' for the \textsc{2cff} instance, 
then we return ``infeasible'' for the \textsc{sff} instance.

\begin{lemma}[\textsc{sff} to \textsc{2cff}]
	\label{lm: sff-2cff}
	Given an \textsc{sff} instance $(G^s,F^s, S_1, S_2, \uu^s, s_1,t_1,s_2,t_2)$,
	we can construct, in time $O(|E^s|)$, a \textsc{2cff} instance $(G^f, F^f, \uu^f, s_1,t_1, s_2,t_2)$ such that
	\[|V^f|=|V^s|+2(|S_1|+|S_2|),~~ |E^f|=|E^s|+4(|S_1|+|S_2|), ~~ |F^f|\leq4(|F^s|+|S_1|+|S_2|),~~ \norm{\uu^f}_{\max}= \norm{\uu^s}_{\max}. \]
	If the \textsc{sff} instance has a solution, then the \textsc{2cff} instance has a solution.
\end{lemma}	
\begin{proof}
	According to the reduction described above, from any solution $\ff^s$ to the \textsc{sff} instance, it is easy to derive a solution $\ff^f$ to the \textsc{2cff} instance. Concretely, we define a feasible flow $\ff^f$ as follows:
	\begin{itemize}
		\item For any selective edge $e\in S_i$, in its corresponding gadget in $G^f$, we set
		\[l\leq \ff^f_i(e_1)=\ff^f_i(e_3)=\ff^s_i(e)\leq u,\]
		\[\ff^f_i(e_4)=\ff^f_i(e_5)=u,\]
		\[0\leq \ff^f_i(e_2)=u-\ff^s_i(e) \leq u-l,\]
		where $l=0$ if $e$ is a non-fixed flow edge, or $l=u=\uu^s(e)$ if $e$ is a fixed flow edge. And we set
		\[\ff^f_{\bar{i}}(e_1)=\ff^f_{\bar{i}}(e_2)=\ff^f_{\bar{i}}(e_3)=\ff^f_{\bar{i}}(e_4)=\ff^f_{\bar{i}}(e_5)=0, \qquad \bar{i}=\{1,2\}\backslash i.\]
		It is obvious that $\ff^f$ satisfies the capacity constraint on edges $(e_1, \ldots, e_5)$, and the flow conservation constraint on vertices $xy, xy'$.
		
		\item For any non-selective edge $e'\in E^s\backslash (S_1\bigcup S_2)$, we also have $e'\in E^f$ since no reduction is made on this edge, and we copy it trivially to $G^f$. We set
		\[\ff^f_i(e')=\ff^s_i(e'), \qquad i\in\{1,2\}.\]
		Since $\ff^s$ is a feasible flow in $G^s$, it is easy to check that $\ff^f$ also satisfies the capacity constraint on non-selective edges, as well as the flow conservation constraint on vertices incident to non-selective edges. Moreover, if $e'$ is a fixed flow edge, the fixed flow edge constraint is also satisfied.
	\end{itemize}
	To conclude, $\ff^f$ is a feasible flow to the \textsc{2cff} instance.
	\\
	\\
	Now, we track the change of problem size after reduction. Based on the reduction method, given an \textsc{sff} instance with $|V^s|$ vertices, $|E^s|$ edges (including $|F^s|$ fixed flow edges, and $|S_i|$ edges selecting commodity $i$), we can compute the size of the reduced \textsc{2cff} instance as follows.
	\begin{enumerate}
		\item $|V^f|$ vertices. As all vertices in $V^s$ are maintained, and for each selective edge $(x,y)\in (S_1\cup S_2)$, 2 auxiliary vertices $xy, xy'$ are added, we have
		$|V^f|=|V^s|+2(|S_1|+|S_2|)$.
		
		\item $|E^f|$ edges. Since all non-selective edges in $E^s\backslash (S_1\cup S_2)$ are maintained, and each selective edge $e\in (S_1\cup S_2)$ is replaced by a gadget with 5 edges $(e_1, \cdots, e_5)$, we have
		\[|E^f|=(|E^s|-|S_1|-|S_2|)+5(|S_1|+|S_2|)=|E^s|+4(|S_1|+|S_2|).\]
		
		\item $|F^f|$ fixed flow edges. First, all non-selective fixed flow edges $e\in F^s\backslash (S_1\cup S_2)$ are maintained since no reduction is made.
		Then, for all selective fixed flow edges $e\in F^s\bigcap (S_1\cup S_2)$, the reduction generates 4 fixed flow edges (i.e., $e_1, e_3, e_4, e_5$).
		Finally, for all non-fixed flow selective edges $e\in (S_1\cup S_2)\backslash F^s$, the reduction generates 2 fixed flow edges (i.e., $e_4, e_5$). Hence, we have
		\[|F^f|=|F^s\backslash (S_1\cup S_2)|+4|F^s\cap (S_1\cup S_2)|+2|(S_1\cup S_2)\backslash F^s|\leq 4|F^s\cup S_1\cup S_2| \leq 4(|F^s|+|S_1|+|S_2|).\]
		
		\item The maximum edge capacity is bounded by 
		$\norm{\uu^f}_{\max}=\norm{\uu^s}_{\max}$.
		First, capacity of all non-selective edges in $E^s\backslash (S_1\cup S_2)$ is unchanged in $G^f$ since no reduction is made. Then, for all selective edges $e\in (S_1\cup S_2)$ with capacity (or fixed flow) $\uu^s(e)$, capacity of the 5 edges $(e_1,\ldots, e_5)$ in the gadget are with capacity at most $\uu^s(e)$.
	\end{enumerate}
	To estimate the reduction time, it is observed that it takes constant time to reduce each selective edge in $S_1\cup S_2$ since only a constant number of vertices and edges are added. And it takes constant time to copy each of the other non-selective edges in $E^s\backslash (S_1\cup S_2)$. Thus, the reduction of this step takes $O(|E^s|)$ time.
\end{proof}

\subsubsection{SFFA to 2CFFA}

\begin{definition}[2CFF Approximate Problem (\textsc{2cffa})]
	\label{def: 2cffa-new}
	A \textsc{2cffa} instance is given by a \textsc{2cff} instance $(G, F, \uu, s_1,t_1, s_2,t_2)$ as in Definition \ref{def: 2cff}, and an error parameter $\epsilon \in[0,1]$, which we collect in a tuple $\left(G,F, \uu, s_1,t_1,s_2,t_2, \epsilon\right)$.
	We say an algorithm solves the \textsc{2cffa} problem, if, given any \textsc{2cffa} instance, it returns a pair of flows $\ff_1, \ff_2\geq\bf{0}$ that satisfies
	\begin{align}
		&\uu(e)-\eps\leq \ff_1(e)+\ff_2(e)\leq\uu(e)+\eps, ~~ \forall e \in F \label{eqn:2cffa_error_congestion_1} \\
		&0\leq \ff_1(e)+\ff_2(e)\leq\uu(e)+\eps, ~~ \forall e \in E\backslash F \label{eqn:2cffa_error_congestion_2} \\
		& \abs{\sum_{\begin{subarray}{c}
					u:
					(u,v)\in E
			\end{subarray}}\ff_i(u,v)-\sum_{\begin{subarray}{c}
					w:
					(v,w)\in E
			\end{subarray}}\ff_i(v,w)}\leq \epsilon, ~~ \forall v\in V\backslash \{s_i,t_i\},~ i\in\{1,2\} \label{eqn:2cffa_error_demand}
	\end{align}
	or it correctly declares that the associated \textsc{2cff} instance is infeasible. We refer to the error in \eqref{eqn:2cffa_error_congestion_1} and \eqref{eqn:2cffa_error_congestion_2} as error in congestion, error in \eqref{eqn:2cffa_error_demand} as error in demand.
\end{definition}	

We can use the same reduction method and solution mapping method in the exact case to the approximate case.
Note that, different from the exact case where $\ff^s_{\bar{i}}(e)=\ff^f_{\bar{i}}(e_1)=0$ if $e\in S_i$, it can be nonzero in the approximate case, which gives rise to error in type in $G^s$.

\begin{lemma}[\textsc{sffa} to \textsc{2cffa}]
	\label{lm: sffa-2cffa}
	Given an \textsc{sffa} instance $(G^s,F^s, S_1, S_2, \uu^s, s_1,t_1,s_2,t_2, \eps^s)$, 
	we can reduce it to a \textsc{2cffa} instance $\left(G^f,F^f, \uu^f, s_1,t_1,s_2,t_2, \epsilon^f\right)$ by letting 
	\[\epsilon^f=\frac{\eps^s}{6|E^s|},\]
	and using Lemma \ref{lm: sff-2cff} to construct a \textsc{2cff} instance $\left(G^f,F^f, \uu^f, s_1,t_1,s_2,t_2\right)$ from the \textsc{sff} instance $(G^s,F^s, S_1, S_2, \uu^s, s_1,t_1,s_2,t_2)$.
	If $\ff^f$ is a solution to the \textsc{2cffa} (\textsc{2cff}) instance, then in time $O(|E^s|)$, we can compute a solution $\ff^s$ to the \textsc{sffa} (\textsc{sff}, respectively) instance,
	where the exact case holds when $\eps^{f} = \eps^{s}=0$.

\end{lemma}	

\begin{proof}
	
	Based on the solution mapping method described above, it takes constant time to set the value of each selective or non-selective entry of $\ff^s$. As $\ff^s$ has $|E^s|$ entries, such a solution mapping takes $O(|E^s|)$ time.	

	Now, we conduct an error analysis. 
	We use $\tau^s$ to denote the error of $\ff^s$ that is obtained by mapping back $\ff^f$ with at most $\eps^f$ additive error. 
	By the error notions of \textsc{sffa} (Definition \ref{def: sffa-new}), there are three types of error to track after solution mapping: (1) error in congestion $\tau^s_u, \tau^s_l$; (2) error in demand $\tau^s_{d1}, \tau^s_{d2}$; (3) error in type $\tau^s_{t1}, \tau^s_{t2}$.
	Then, we set the additive error of \textsc{sffa} as
	\[\tau^s=\max\{\tau^s_u, \tau^s_l,\tau^s_{d1}, \tau^s_{d2}, \tau^s_{t1}, \tau^s_{t2}\}.\]
	\begin{enumerate}
		\item Error in congestion.\\
		Since we map solution back trivially for non-selective edges in $E^s\backslash (S_1\cup S_2)$, error in congestion of $G^f$ (i.e., $\epsilon^f_l, \epsilon^f_u$) still holds for $G^s$. Therefore, we focus on how error in congestion changes when mapping solution back to selective edges in $S_1 \cup S_2$.
		
		We first track the error of the upper bound of capacity $\tau^s_u$. For any selective edge $e$, by error in congestion defined in Eq. \eqref{eqn:sffa_error_congestion_2} and Eq. \eqref{eqn:2cffa_error_congestion_2}, we have 
		\[\ff^s(e)=\sum_{i=\{1,2\}}\ff^s_i(e)\leq \uu^s(e)+\tau^s_u, \qquad \ff^s(e)=\ff^f(e_1)\leq \uu^f(e)+\eps^f=\uu^s(e)+\eps^f.\]
		Thus, it suffices to set $\tau^s_u= \epsilon^f$.
	
		Next, we track the error of the lower bound of capacity $\tau^s_l$. We only need to take selective fixed flow edges into consideration since there is no error of the lower bound of capacity for non-fixed flow edges. If $e$ is a fixed flow edge, then $e_1$ is also a fixed flow edge with the same capacity as $e$. Hence, by error in congestion defind in Eq. \eqref{eqn:sffa_error_congestion_1} and Eq. \eqref{eqn:2cffa_error_congestion_1}, 
		\[\ff^s(e)\geq \uu^s(e)-\tau^s_l, \qquad \ff^s(e)=\ff^f(e_1)\geq\uu^f(e)-\eps^f= \uu^s(e)-\eps^f.\]
		Thus, it suffices to set $\tau^s_l= \epsilon^f$.
		
		\item Error in demand.\\
		For simplicity, we denote $S=S_1\bigcup S_2$. 
		We consider commodity $i, i\in\{1,2\}$, then we analyze error in demand for vertices other than $s_i, t_i$.
		\begin{itemize}
			\item \textbf{Case 1:} For $s_{\bar{i}}, t_{\bar{i}}$.\\
			 We use $\tilde{\tau}^s_{di}$ to denote error in demand of $\ff^s$ for vertices $s_{\bar{i}}, t_{\bar{i}}$ with respect to commodity $i$.
			 Since incident edges for $s_{\bar{i}}, t_{\bar{i}}$ decreases in $G^s$, error in demand on $s_{\bar{i}}, t_{\bar{i}}$ will not increase. Thus, $\tilde{\tau}^s_{di}\leq\eps^f$.
			 \item \textbf{Case 2:} For vertices other than $\{s_1, t_1,s_2, t_2\}$.\\
			 We use $\bar{\tau}^s_{di}$ to denote error in demand of $\ff^s$ for vertices other than $\{s_1, t_1,s_2, t_2\}$ with respect to commodity $i$.
			 As defined in Eq. \eqref{eqn:sffa_error_demand}, error in demand $\bar{\tau}^s_{di}$ is computed as 
			 \begin{scriptsize}
			 	\begin{equation}
			 		\label{eq: slu-lu}
			 		\begin{aligned}
			 			\bar{\tau}_{di}^s&=\max_{y\in V^s\backslash\{s_1, t_1, s_2, t_2\}}\abs{\sum_{(x,y)\in E^s}\ff^s_i(x,y)-\sum_{(y,z)\in E^s}\ff^s_i(y,z)}\\
			 			&= \max_{y\in V^s\backslash\{s_1, t_1, s_2, t_2\}}\abs{\left(\sum_{(x,y)\in S}\ff^s_i(x,y)+\sum_{(x,y)\in E^s\backslash S}\ff^s_i(x,y)\right)-\left(\sum_{(y,z)\in S}\ff^s_i(y,z)+\sum_{(y,z)\in E^s\backslash S}\ff^s_i(y,z)\right)}\\
			 			&= \max_{y\in V^s\backslash\{s_1, t_1, s_2, t_2\}}\abs{\left(\sum_{e=(x,y)\in S}\ff^f_i(e_1)+\sum_{(x,y)\in E^s\backslash S}\ff^f_i(x,y)\right)-\left(\sum_{e=(y,z)\in S}\ff^f_i(e_1)+\sum_{(y,z)\in E^s\backslash S}\ff^f_i(y,z)\right)} \\
			 			&\leq \epsilon^f+\max_{y\in V^s\backslash\{s_1, t_1, s_2, t_2\}}\abs{\left(\sum_{e=(x,y)\in S}\ff^f_i(e_1)+\sum_{(x,y)\in E^s\backslash S}\ff^f_i(x,y)\right)-\left(\sum_{e=(x,y)\in S}\ff^f_i(e_3)+\sum_{(x,y)\in E^s\backslash S}\ff^f_i(x,y)\right)} \\
			 			&= \epsilon^f+\max_{y\in V^s\backslash\{s_1, t_1, s_2, t_2\}}\abs{\sum_{e=(x,y)\in S}\ff^f_i(e_1)-\sum_{e=(x,y)\in S}\ff^f_i(e_3)}\\
			 			&\leq \epsilon^f+\max_{y\in V^s\backslash\{s_1, t_1, s_2, t_2\}}\sum_{e=(x,y)\in S}\abs{\ff^f_i(e_1)-\ff^f_i(e_3)}\\
			 			&= \epsilon^f+\max_{y\in V^s\backslash\{s_1, t_1, s_2, t_2\}}\left\{\sum_{e=(x,y)\in S_1}\abs{\ff^f_i(e_1)-\ff^f_i(e_3)}+\sum_{\hat{e}=(x,y)\in S_2}\abs{\ff^f_i(\hat{e}_1)-\ff^f_i(\hat{e}_3)}\right\}.
			 		\end{aligned}
			 	\end{equation}
			 \end{scriptsize}

			 Now, we try to bound $\abs{\ff^f_i(e_1)-\ff^f_i(e_3)}$. 
			 Applying error in demand as defined in Eq. \eqref{eqn:2cffa_error_demand} to vertices $xy$ and $xy'$, we have
			 \[\left|\ff^f_i(e_1)+\ff^f_i(e_2)-\ff^f_i(e_4)\right|\leq \epsilon^f,\]
			 \[\left|\ff^f_i(e_3)+\ff^f_i(e_2)-\ff^f_i(e_5)\right|\leq \epsilon^f.\]
			 Combining the above two inequalities gives
			 \begin{equation}
			 	\label{eq: slu-lu-s1}
			 	\begin{aligned}
			 		\left|\ff^f_i(e_1)-\ff^f_i(e_3)\right|&\leq2\epsilon^f+\left|\ff^f_i(e_4)-\ff^f_i(e_5)\right|.
			 	\end{aligned}
			 \end{equation}	
			 
			 Now, we assume $e\in S_1$. We can get the same bound for $\hat{e}\in S_2$ by symmetry.
			 Eq. \eqref{eq: slu-lu-s1} can be further bounded by
			 \begin{equation}
			 	\label{eq: slu-lu-s1(1)}
			 	\begin{aligned}
			 		\left|\ff^f_1(e_1)-\ff^f_1(e_3)\right|&\leq2\epsilon^f+\left|\ff^f_1(e_4)-\ff^f_1(e_5)\right|\\
			 		&=2\epsilon^f+\left|(\ff^f(e_4)-\ff^f_{2}(e_4))-(\ff^f(e_5)-\ff^f_{2}(e_5))\right|\\
			 		&\leq2\epsilon^f+\left|\ff^f(e_4)-\ff^f(e_5)\right|+\left|\ff^f_{2}(e_4)-\ff^f_{2}(e_5)\right|\\
			 		&\overset{(1)}{\leq}4\epsilon^f+\left|\ff^f(e_4)-\ff^f(e_5)\right|\\
			 		&\overset{(2)}{\leq} 6\epsilon^f.
			 	\end{aligned}
			 \end{equation}	
			 and
			 \begin{equation}
			 	\label{eq: slu-lu-s1(2)}
			 	\begin{aligned}
			 		\left|\ff^f_2(e_1)-\ff^f_2(e_3)\right|&\leq2\epsilon^f+\left|\ff^f_2(e_4)-\ff^f_2(e_5)\right|\overset{(3)}{\leq} 4\epsilon^f.
			 	\end{aligned}
			 \end{equation}
			 For step (1), we use $\ff^f_{2}(e_4), \ff^f_{2}(e_5)\leq \eps^f$ by error in demand on $t_1, s_1$ in $G^f$.
			 For step (2) and (3), we use the error in congestion for $e_4$ and $e_5$ in $G^f$.

			 By symmetry, for $\hat{e}\in S_2$, we have 
			 \begin{equation}
			 	\label{eq: slu-lu-s2(1)}
			 	\abs{\ff^f_1(\hat{e}_1)-\ff^f_1(\hat{e}_3)}\leq 4\epsilon^f,
			 \end{equation}			
			 \begin{equation}
			 	\label{eq: slu-lu-s2(2)}
			 	\abs{\ff^f_2(\hat{e}_1)-\ff^f_2(\hat{e}_3)}\leq 6\epsilon^f.
			 \end{equation}
			 
			 Finally, applying Eq. \eqref{eq: slu-lu-s1(1)} and Eq. \eqref{eq: slu-lu-s2(1)} to Eq. \eqref{eq: slu-lu}, we obtain
			 \begin{equation}
			 	\label{eq: slu-lu-d1}
			 	\begin{aligned}
			 		\bar{\tau}_{d1}^s&\leq \epsilon^f+\max_{y\in V^s\backslash\{s_1,t_1\}}\left\{\sum_{e=(x,y)\in S_1}\abs{\ff^f_1(e_1)-\ff^f_1(e_3)}+\sum_{\hat{e}=(x,y)\in S_2}\abs{\ff^f_1(\hat{e}_1)-\ff^f_1(\hat{e}_3)}\right\}\\
			 		&\leq \epsilon^f+\max_{y\in V^s\backslash\{s_1,t_1\}}\left\{\sum_{e=(x,y)\in S_1}6\eps^f+\sum_{\hat{e}=(x,y)\in S_2}4\epsilon^f\right\}\\
			 		&\leq \epsilon^f+(6|S_1|+4|S_2|)\epsilon^f\leq 6(|S_1|+|S_2|)\eps^f.
			 	\end{aligned}
			 \end{equation}			

			 Similarly, applying Eq. \eqref{eq: slu-lu-s1(2)} and Eq. \eqref{eq: slu-lu-s2(2)} to Eq. \eqref{eq: slu-lu}, we obtain
			 \begin{equation}
			 	\label{eq: slu-lu-d2}
			 	\begin{aligned}
					\bar{\tau}_{d2}^s\leq \epsilon^f+(4|S_1|+6|S_2|)\epsilon^f\leq 6(|S_1|+|S_2|)\eps^f.
			 	\end{aligned}
			 \end{equation} 
		\end{itemize}
		Combining the above cases, we can bound the error in demand uniformly by
		\[\tau^s_{di}:=\max\{\tilde{\tau}^s_{di}, \bar{\tau}^s_{di}\}\leq 6|E^s|\epsilon^f, \qquad i\in\{1,2\}.\]

		\item Error in type.\\
		As defined in Eq. \eqref{eqn:sffa_error_type}, error in type $\tau^s_{ti}$ for edges selecting commodity $i$ is computed as 
		\begin{equation*}
			\label{eq: slu-t}
			\begin{aligned}
				\tau_{ti}^s=\max_{e\in S_i}\ff^s_{\bar{i}}(e)=\max_{e\in S_i}\ff^f_{\bar{i}}(e_1)\leq \max_{e\in S_i}(\ff^f_{\bar{i}}(e_1)+\ff^f_{\bar{i}}(e_2)),
			\end{aligned}
		\end{equation*}
		where the last inequality follows from $\ff^f_{\bar{i}}(e_2)\geq 0$.
		Again, we first consider $e\in S_1$. 
		For an arbitrary edge $e=(x,y)\in S_1$, it is known that $e_4$ and $e_5$ are incident to $t_1$ and $s_1$.
		By error in demand on vertex $xy$, we have
		\[\abs{\ff^f_2(e_1)+\ff^f_2(e_2)-\ff^f_2(e_4)}\leq \epsilon^f.\]
		As $\ff^f_2(e_4)\leq \epsilon^f$ because of error in demand on $t_1$ in $G^f$, then we can bound
		\[\ff^f_2(e_1)+\ff^f_2(e_2)\leq 2\epsilon^f.\]
		Since $e$ is an arbitrary edge in $S_1$, thus $\tau_{t1}^s\leq 2\epsilon^f$.		
		By symmetry, we also have $\tau_{t2}^s\leq 2\epsilon^f$.
		
	\end{enumerate}
	To summarize, we can set
	\[\tau^s=\max\{\tau^s_u, \tau^s_l,\tau^s_{d1}, \tau^s_{d2}, \tau^s_{t1}, \tau^s_{t2}\}\leq 6|E^s|\eps^f.\]
	
	As we set in the reduction that $\eps^f=\frac{\eps^s}{6|E^s|}$, then we have
	\[\tau^s\leq 6|E^s|\cdot\frac{\eps^s}{6|E^s|}=\eps^s,\]
	indicating that $\ff^s$ is a solution to the \textsc{sffa} instance $(G^s,F^s, S_1, S_2, \uu^s, s_1,t_1,s_2,t_2, \eps^s)$.
\end{proof}


\subsection{2CFF(A) to 2CFR(A)}
\label{sect: 2cffa-2cfra}
\subsubsection{2CFF to 2CFR}
We show the reduction from a \textsc{2cff} instance $(G^f, F^f, \uu^f, s_1,t_1, s_2,t_2)$ to a \textsc{2cfr} instance \\$(G^r, \uu^r, \bar{s}_1,\bar{t}_1, \bar{s}_2, \bar{t}_2, R_1, R_2)$. 
First of all, we add two new sources $\bar{s}_1, \bar{s}_2$ and two new sinks $\bar{t}_1, \bar{t}_2$. Then, for each edge $e\in E^f$, we map it to a gadget consisting edges $\{e_1,e_2,e_3,e_4,e_5, e_6,e_7\}$ in $G^{r}$, as shown in the upper part of Figure \ref{fig: lu2cf-2cfr}. Additionally, there is another gadget with 5 edges in $G^r$ that connects the original sink $t_i$ and source $s_i$, as shown in the lower part of Figure \ref{fig: lu2cf-2cfr}. Capacity of these edges is the sum capacities of all edges in $G^f$, i.e., $M^f=\sum_{e\in E^f} \uu^f(e)$. 
Additionally, we set $R_1=R_2=2M^f$, indicating that at least $2M^f$ unit of flow should be routed from $\bar{s}_i$ to $\bar{t}_i$, $i\in\{1,2\}$. 

The key idea to remove the fixed flow constraint is utilizing edge directions and the requirements that $2M^f$ units of the flow of commodity $i$ to be routed from the new source $\bar{s}_i$ to the new sink $\bar{t}_i$, $i\in\{1,2\}$.
It is noticed that all edges that are incident to the new sources and sinks should be saturated to fulfill the requirements. Therefore, for a fixed flow edge $e$ in the first gadget, the incoming flow of vertex $xy'$ and the outgoing flow of vertex $xy$ must be $2u$. Since the capacity of $e_2$ is $2u-u=u$, then the flow of $e_1, e_2, e_3$ are forced to be $u$. As such, the fixed flow constraint can be simulated. 
Note that instead of simply copying non-fixed flow edges to $G^r$, 
we also need to map non-fixed flow edges to the designed gadget in this step. 
This guarantees that the requirement on flow values can be satisfied if the \textsc{2cff} instance is feasible.

\begin{figure}[ht]
	\centering
	\includegraphics[width=\textwidth]{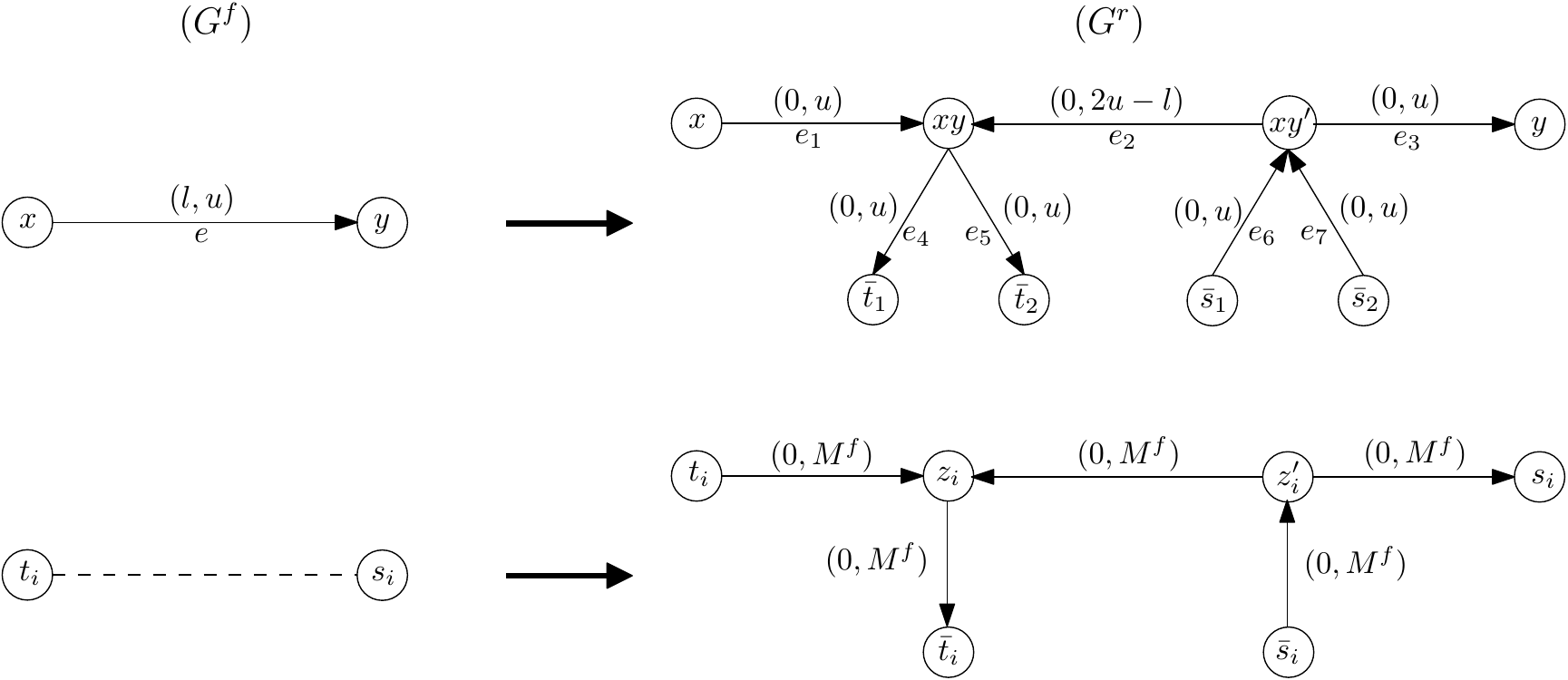}
	\caption{The reduction from \textsc{2cff} to \textsc{2cfr}. $l=u$ if $e$ is a fixed flow edge, and $l=0$ if $e$ is a non-fixed flow edge.} 
	\label{fig: lu2cf-2cfr}
\end{figure}

If a \textsc{2cfr} solver returns $\ff^r$ for the \textsc{2cfr} instance $(G^r, \uu^r, \bar{s}_1,\bar{t}_1, \bar{s}_2, \bar{t}_2, R_1, R_2)$, then we return $\ff^f$ for the \textsc{2cff} instance $(G^f, F^f, \uu^f, s_1,t_1, s_2,t_2)$ by setting $\ff_i^f(e)=\ff^r_i(e_1), \forall e\in E^f, i\in\{1,2\}$.
If the \textsc{2cfr} solver returns ``infeasible'' for the \textsc{2cfr} instance, 
then we return ``infeasible'' for the \textsc{2cff} instance.

\begin{lemma}[\textsc{2cff} to \textsc{2cfr}]
	\label{lm: 2cff-2cfr}
	Given a \textsc{2cff} instance $(G^f, F^f, \uu^f, s_1,t_1, s_2,t_2)$,
	we can construct, in time $O(|E^f|)$, a \textsc{2cfr} instance $(G^r, \uu^r, \bar{s}_1,\bar{t}_1, \bar{s}_2, \bar{t}_2, R_1, R_2)$ such that
	\[|V^r|=|V^f|+2|E^f|+8,~~|E^r|=7|E^f|+10,~~ \norm{\uu^r}_{\max}= \max\left\{2\norm{\uu^f}_{\max}, M^f\right\},~~ R_1=R_2=2M^f,\] 
	where $M^f=\sum_{e\in E^f} \uu^f(e)$. 
	If the \textsc{2cff} instance has a solution, then the \textsc{2cfr} instance has a solution.
\end{lemma}	
\begin{proof}
	According to the reduction described above, from any solution $\ff^f$ to the \textsc{2cff} instance, it is easy to derive a solution $\ff^r$ to the \textsc{2cfr} instance. Concretely, we define a feasible flow $\ff^r$ as follows. For any edge $e\in E^f$, we set:

		\[\ff^r_i(e_1)=\ff^r_i(e_3)=\ff^f_i(e), \qquad i\in\{1,2\},\]
		\[\ff^r_i(e_2)=u-\ff^f_i(e), \qquad i\in\{1,2\},\]
		\[\ff^r_1(e_4)=\ff^r_1(e_6)=u, \qquad \ff^r_2(e_5)=\ff^r_2(e_7)=u,\]
		\[\ff^r_2(e_4)=\ff^r_2(e_6)= \ff^r_1(e_5)=\ff^r_1(e_7)=0,\]
		where $u=\uu^f(e)$.
		And we set
		\[\ff^r_i(t_i, z_i)=\ff^r_i(z'_i,s_i)=\sum_{\begin{subarray}{c}
				w:
				(s_i, w)\in E^f
		\end{subarray}}\ff_i^f(s_i, w)\leq M^f,\qquad i\in\{1,2\},\]
		\[\ff^r_i(z'_i, z_i)=M^f-\sum_{\begin{subarray}{c}
				w:
				(s_i, w)\in E^f
		\end{subarray}}\ff_i^f(s_i, w),\qquad i\in\{1,2\},\]
		\[\ff^r_i(\bar{s}_i,z'_i)=\ff^r_i(z_i,\bar{t}_i)=M^f, \qquad i\in\{1,2\},\]
		\[\ff^r_{\bar{i}}(t_i, z_i)=\ff^r_{\bar{i}}(z'_i,s_i)=\ff^r_{\bar{i}}(z'_i, z_i)=\ff^r_{\bar{i}}(\bar{s}_i,z'_i)=\ff^r_{\bar{i}}(z_i,\bar{t}_i)=0,\qquad \bar{i}\in\{1,2\}\backslash i.\]
		
		Now, we prove that $\ff^r$ is a solution to the \textsc{2cfr} instance. 
		\begin{enumerate}
			\item Capacity constraint.\\
			Only the capacity constraint on edge $e_2$ is nontrivial. 
			\begin{itemize}
				\item If $e$ is a fixed flow edge, then we have $l=u$, thus the capacity of $e_2$ is $u$. By construction, we have 
				\[\sum_{i=\{1,2\}}\ff^r_i(e_2)=2u-\sum_{i=\{1,2\}}\ff^f_i(e)=2u-u=u,\]
				where we use $\sum_{i=\{1,2\}}\ff^f_i(e)=u$.
				
				\item If $e$ is a non-fixed flow edge, then we have $l=0$, thus the capacity of $e_2$ is $2u$.
				By construction, we have
				\[\sum_{i=\{1,2\}}\ff^r_i(e_2)=2u-\sum_{i=\{1,2\}}\ff^f_i(e)\leq 2u,\]
				where we use $\sum_{i=\{1,2\}}\ff^f_i(e)\geq 0$.
			\end{itemize}
			Therefore, we conclude that $\ff^r$ fulfills the capacity constraint.	
			
			\item Flow conservation constraint.\\
			Only the flow conservation constraint on vertices $xy, xy', z_i, z'_i$ are nontrivial.
			\begin{itemize}
				\item For vertex $xy$, we have
				\[\ff^r_1(e_1)+\ff^r_1(e_2)=\ff^f_1(e)+(u-\ff^f_1(e))=u=\ff^r_1(e_4),\]
				\[\ff^r_2(e_1)+\ff^r_2(e_2)=\ff^f_2(e)+(u-\ff^f_2(e))=u=\ff^r_2(e_5).\]
				We can prove for vertex $xy'$ similarly.
				\item For vertex $z_i, i\in\{1,2\}$, we have
				\begin{align*}
					\small
					\ff^r_i(t_i, z_i)+\ff^r_i(z'_i, z_i)=\sum_{\begin{subarray}{c}
							w:
							(s_i, w)\in E^f
					\end{subarray}}\ff_i^f(s_i, w)+\left(M^f-\sum_{\begin{subarray}{c}
							w:
							(s_i, w)\in E^f
					\end{subarray}}\ff_i^f(s_i, w)\right)=M^f=\ff^r_i(z_i,\bar{t}_i).
				\end{align*}
				We can prove for vertex $z'_i$ similarly.
			\end{itemize}
			Therefore, we conclude that $\ff^r$ fulfills the flow conservation constraint.	
			
			\item Requirement $R_1=R_2=2M^f$.\\
			For commodity 1, we have
			\[R_1=\sum_{\begin{subarray}{c}
					w:
					(\bar{s}_1, w)\in E^r
			\end{subarray}}\ff_1^r(\bar{s}_1, w)=\ff^r(\bar{s}_1, z'_1)+\sum_{e\in E^f}\ff^r_1(e_6)=M^f+\sum_{e\in E^f}\uu^f(e)=2M^f,\]
			\[R_1=\sum_{\begin{subarray}{c}
					u:
					(u,\bar{t}_1)\in E^r
			\end{subarray}}\ff_1^r(u,\bar{t}_1)=\ff^r(z_1, \bar{t}_1)+\sum_{e\in E^f}\ff^r_1(e_4)=M^f+\sum_{e\in E^f}\uu^f(e)=2M^f.\]
			We can prove that the requirement $R_2=2M^f$ is satisfied similarly.
		\end{enumerate}
	To conclude, $\ff^r$ is a feasible flow to the \textsc{2cfr} instance.
	\\
	\\
	Now, we track the change of problem size after reduction. Based on the reduction method, given a \textsc{2cff} instance with $|V^f|$ vertices and $|E^f|$ edges (including $|F^f|$ fixed flow edges), we can compute the size of the reduced \textsc{2cfr} instance as follows.
	\begin{enumerate}
		\item $|V^r|$ vertices. First, all vertices in $V^f$ are maintained.
		Then, for each edge $e=(x,y)$, two new auxiliary vertices $xy, xy'$ are added.
		And finally, in the second gadget, $\{\bar{s}_i, \bar{t}_i, z_i, z_i'\}, i\in\{1,2\}$ are added. Hence, we have
		$|V^r|=|V^f|+2|E^f|+8$.
		
		\item $|E^r|$ edges. First, each edge $e\in E^f$ is replaced by the first gadget in Figure \ref{fig: lu2cf-2cfr} with 7 edges $(e_1, \ldots, e_7)$.
		Then, 5 edges are added between $t_i$ and $s_i$ according to the second gadget in Figure \ref{fig: lu2cf-2cfr}, $i\in\{1,2\}$. Thus, we have
		$|E^r|=7|E^f|+10$.
		
		\item The maximum edge capacity is bounded by
		$\norm{\uu^r}_{\max}=\max\left\{2\norm{\uu^f}_{\max}, M^f\right\}$.		
		For all fixed flow edges $e\in F^f$ with fixed flow $\uu^f(e)$, the 7 edges $(e_1, \ldots, e_7)$ in the first gadget are all with capacity $\uu^f(e)$.
		For all non-fixed flow edges in $e\in E^f\backslash F^f$ with capacity $\uu^f(e)$, the capacity of $e_2$ becomes $2\uu^f(e)$ and the capacity of the rest 6 edges in the first gadget remains $\uu^f(e)$. 
		Moreover, the 5 edges in the second gadget are all with capacity $M^f$, which is the sum of the capacity of all edges.
	\end{enumerate}
	To estimate the reduction time, it is observed that it takes constant time to reduce each edge in $E^f$ since only a constant number of vertices and edges are added. Thus, the reduction of this step takes $O(|E^f|)$ time.
\end{proof}	

\subsubsection{2CFFA to 2CFRA}

\begin{definition}[2CFR Approximate Problem (\textsc{2cfra})]
	\label{def: 2cfra-new}
	A \textsc{2cfra} instance is given by a \textsc{2cfr} instance $(G, \uu, s_1,t_1, s_2,t_2,R_1,R_2)$ as in Definition \ref{def: 2cfr}, and an error parameter $\epsilon \in[0,1]$, which we collect in a tuple $\left(G, \uu, s_1,t_1, s_2,t_2,R_1,R_2,\epsilon \right)$.
	We say an algorithm solves the \textsc{2cfra} problem, if, given any \textsc{2cfra} instance, it returns a pair of flows $\ff_1, \ff_2\geq\bf{0}$ that satisfies
	\begin{align}
		& \ff_1(e) + \ff_2 (e) \le  \uu(e) + \epsilon, ~ \forall e \in E \label{eqn:error_congestion_2cfra} \\
		& \abs{\sum_{u: (u,v)\in E}\ff_i(u,v)-\sum_{w: (v,w)\in E}\ff_i(v,w)}\leq \eps, ~ \forall v \in V\backslash \{s_i, t_i\}, i\in\{1,2\} \label{eqn:error_demand_1_2cfra} \\
		& \abs{\sum_{w: (s_i,w)\in E} \ff_i(s_i,w)-R_i}\leq \eps, \qquad \abs{\sum_{u: (u,t_i)\in E} \ff_i(u,t_i)-R_i}\leq \eps, ~ i\in\{1,2\} \label{eqn:error_demand_2_2cfra} 
	\end{align}
	or it correctly declares that the associated \textsc{2cfr} instance is infeasible. We refer to the error in \eqref{eqn:error_congestion_2cfra} as error in congestion, error in~\eqref{eqn:error_demand_1_2cfra} and \eqref{eqn:error_demand_2_2cfra} as error in demand.
\end{definition}	

We can use the same reduction method and solution mapping method in the exact case to the approximate case.

\begin{lemma}[\textsc{2cffa} to \textsc{2cfra}]
	\label{lm: 2cffa-2cfra}
	Given a \textsc{2cffa} instance $\left(G^f,F^f, \uu^f, s_1,t_1,s_2,t_2, \epsilon^f\right)$, 
	we can reduce it to a \textsc{2cfra} instance $\left(G^r, \uu^r, \bar{s}_1,\bar{t}_1, \bar{s}_2,\bar{t}_2,R_1,R_2,\epsilon^r \right)$ by letting 
	\[\epsilon^r= \frac{\eps^f}{12|E^f|},\]
	and using Lemma \ref{lm: 2cff-2cfr} to construct a \textsc{2cfr} instance $\left(G^r, \uu^r, \bar{s}_1,\bar{t}_1, \bar{s}_2,\bar{t}_2,R_1,R_2\right)$ from the \textsc{sff} instance $\left(G^f,F^f, \uu^f, s_1,t_1,s_2,t_2\right)$.
	If $\ff^r$ is a solution to the \textsc{2cfra} (\textsc{2cfr}) instance, then in time $O(|E^f|)$, we can compute a solution $\ff^f$ to the \textsc{2cffa} (\textsc{2cff}, respectively) instance,
	where the exact case holds when $\eps^{r} = \eps^{f}=0$.
	
\end{lemma}	
\begin{proof}
	Based on the solution mapping method described above, it takes constant time to set the value of each entry of $\ff^f$, and $\ff^f$ has $|E^f|$ entries, such a solution mapping takes $O(|E^f|)$ time.	

	Now, we conduct an error analysis. 
	We use $\tau^f$ to denote the error of $\ff^f$ that is obtained by mapping back $\ff^r$ with at most $\eps^r$ additive error. 
	By the error notions of \textsc{2cffa} (Definition \ref{def: 2cffa-new}), there are two types of error to track after solution mapping: (1) error in congestion $\tau^f_u, \tau^f_l$; (2) error in demand $\tau^f_{d1}, \tau^f_{d2}$.
	Then, we set the additive error of \textsc{2cffa} as
	\[\tau^f=\max\{\tau^f_u, \tau^f_l,\tau^f_{d1}, \tau^f_{d2}\}.\]
	
	\begin{enumerate}
		\item Error in congestion.\\ 
		We first track the error of the upper bound of capacity $\tau^f_u$. For any edge $e\in E^f$, by error in congestion defined in Eq. \eqref{eqn:2cffa_error_congestion_2} and Eq. \eqref{eqn:error_congestion_2cfra}, we have
		\[\ff^f(e)=\sum_{i=\{1,2\}}\ff_i^f(e)\leq \uu^f(e)+\tau_u^f,\qquad \ff^f(e)=\ff^r(e_1)\leq \uu^r(e_1)+\epsilon^r=\uu^f(e)+\epsilon^r.\]
		Thus, it suffices to set $\tau_u^f=\epsilon^r$.

		Next, we track the error of the lower bound of capacity $\tau^f_l$ for fixed flow edges. 
		As defined in Eq. \eqref{eqn:error_demand_2_2cfra}, error in demand of vertices $\bar{t}_1, \bar{t}_2$ gives
		\[\abs{\sum_{e\in E^f}\ff^r_1(e_4)+\ff^r(z_1, \bar{t}_1)-2M^f}\leq \epsilon^r,\]
		\[\abs{\sum_{e\in E^f}\ff^r_2(e_5)+\ff^r(z_2, \bar{t}_2)-2M^f}\leq\epsilon^r.\]

		We now try to compute the minimum flow that can be routed through any fixed flow edge in $F^f$. For an arbitrary $\hat{e}\in F^f$, we can rearrange the above equations and have
		\begin{align*}
			\ff^r_1(\hat{e}_4)&\geq 2M^f-\epsilon^r-\sum_{e\in E^f\backslash\hat{e}}\ff^r_1(e_4)-\ff^r(z_1, \bar{t}_1)\\
			&\geq \max\left\{0, 2M^f-\epsilon^r-\left(M^f-\uu^f(\hat{e})+(|E^f|-1)\eps^r\right)-(M^f+\eps^r)\right\} \\
			&=\max\left\{0,\uu^f(\hat{e})-(|E^f|+1)\eps^r\right\} 
		\end{align*}

		Similarly, we have $\ff^r_2(\hat{e}_5) \geq\max\left\{0,\uu^f(\hat{e})-(|E^f|+1)\eps^r\right\}$.
		Moreover, by error in demand on vertex $xy$, we have
		\[\abs{\ff^r_i(\hat{e}_1)+\ff^r_i(\hat{e}_2)-\ff^r_i(\hat{e}_4)-\ff^r_i(\hat{e}_5)}\leq \eps^r, ~ i\in\{1,2\}.\]
		Therefore, for any fixed flow edge $\hat{e}=(x,y) \in F^f$, by error in congestion defined in Eq. \eqref{eqn:2cffa_error_congestion_1}, we have
		\[\ff^f(\hat{e})\geq\uu^f(\hat{e})-\tau_l^f,\]
		and
		\begin{equation}
			\label{eq: lu2cf_1}
			\begin{aligned}
				\ff^f(\hat{e})&=\ff^r(\hat{e}_1)=\ff^r_1(\hat{e}_1)+\ff^r_2(\hat{e}_1)\\
				&\geq \ff^r_1(\hat{e}_4)+\ff^r_1(\hat{e}_5)-\ff^r_1(\hat{e}_2)-\eps^r + \ff^r_2(\hat{e}_4)+\ff^r_2(\hat{e}_5)-\ff^r_2(\hat{e}_2)-\eps^r\\
				&\geq \ff^r_1(\hat{e}_4)+\ff^r_2(\hat{e}_5)-\ff^r(\hat{e}_2)-2\eps^r \\
				&\overset{(1)}{\geq} \max\{0,2(\uu^f(\hat{e})-(|E^f|+1)\eps^r)-(\uu^f(\hat{e})+\eps^r)-2\eps^r\}\\
				&=\max\{0,\uu^f(\hat{e})-(2|E^f|+5)\eps^r\}.
			\end{aligned}
		\end{equation}
		For step (1), we plug in the lower bound of $\ff^r_1(\hat{e}_4), \ff^r_2(\hat{e}_5)$ that are proved previously, and the upper bound of $\ff^r(\hat{e}_2)$.
		Thus, we have
		\[\tau_l^f\leq (2|E^f|+5)\eps^r\leq 7|E^f|\eps^r.\]

		\item Error in demand.\\
		We focus on the error in demand for commodity 1 $\tau^f_{d1}$ first, and then the error in demand for commodity 2 $\tau^f_{d2}$ can be achieved directly by symmetry. Then we analyze error in demand for vertices other than $s_1, t_1$.
		\begin{itemize}
			\item \textbf{Case 1:} For vertex $s_2$.\\
			We use $\tilde{\tau}^f_{d1}$ to denote error in demand of $\ff^f$ for vertex $s_2$.
			By error in demand as defined in Eq. \eqref{eqn:error_demand_1_2cfra} of vertex $s_2$ and $\bar{s}_2$, we have $\abs{\ff^r_1(z'_2,s_2)-\sum_{e=(s_2,y)\in E^f}\ff^r_1(e_1)}\leq \eps^r$ and $\ff^r_1(z'_2,s_2)\leq \ff^r_1(\bar{s}_2,z'_2)\leq\eps^r$, respectively. Thus,
			\[\tilde{\tau}^f_{d1}=\sum_{e=(s_2,y)\in E^f}\ff^f_1(e)=\sum_{e=(s_2,y)\in E^f}\ff^r_1(e_1)\leq 2\eps^r.\]		

			\item \textbf{Case 2:} For vertex $t_2$.\\
			We use $\bar{\tau}^f_{d1}$ to denote error in demand of $\ff^f$ for vertex $s_2$.
			By error in demand as defined in Eq. \eqref{eqn:error_demand_1_2cfra} of vertex $t_2$ and $\bar{t}_2$, we have $\abs{\ff^r_1(t_2, z_2)-\sum_{e=(x,t_2)\in E^f}\ff^r_1(e_3)}\leq \eps^r$ and $\ff^r_1(t_2,z_2)\leq \ff^r_1(z_2, \bar{t}_2)\leq\eps^r$, respectively. Thus,
			\[\sum_{e=(x,t_2)\in E^f}\ff^r_1(e_3)\leq 2\eps^r,\]
			and
			\begin{multline*}
				\bar{\tau}^f_{d1}=\sum_{e=(x,t_2)\in E^f}\ff^f_1(e)= \sum_{e=(x,t_2)\in E^f}\ff^r_1(e_1)\leq \\\sum_{e=(x,t_2)\in E^f}\ff^r_1(e_3)+\sum_{e=(x,t_2)\in E^f}\abs{\ff^r_1(e_1)-\ff^r_1(e_3)}\leq 2\eps^r+10|E^f|\eps^r\leq 12|E^f|\eps^r,
			\end{multline*}
			where the last inequality comes from Eq. \eqref{eq: lu-2cfa}, which will be proved soon. %
			
			\item \textbf{Case 3:} For vertices other than $\{s_1, t_1, s_2, t_2\}$.\\
			We use $\hat{\tau}^f_{d1}$ to denote error in demand of $\ff^f$ for vertices other than $\{s_1, t_1, s_2, t_2\}$.
			As error in demand defined in Eq. \eqref{eqn:2cffa_error_demand}, $\hat{\tau}^f_{d1}$ is computed as	
			\begin{small}
				\begin{equation}
					\label{eq: lu-2cfa}
					\begin{aligned}
						\hat{\tau}_{d1}^f=&\max_{y\in V^f\backslash\{s_1,t_1, s_2,t_2\}}\abs{\sum_{(x,y) \in E^f}\ff^f_1(x,y)-\sum_{(y,z) \in E^f}\ff^f_1(y,z)}\\
						=&\max_{y\in V^f\backslash\{s_1,t_1, s_2,t_2\}}\abs{\sum_{e=(x,y) \in E^f}\ff^r_1(e_1)-\sum_{e=(y,z) \in E^f}\ff^r_1(e_1)}\\
						\leq&\eps^r+\max_{y\in V^f\backslash\{s_1,t_1, s_2,t_2\}}\abs{\sum_{e=(x,y) \in E^f}\ff^r_1(e_1)-\sum_{e=(x,y) \in E^f}\ff^r_1(e_3)}\\
						\leq& \eps^r+\max_{y\in V^f\backslash\{s_1,t_1, s_2,t_2\}}\sum_{e=(x,y) \in E^f}\left|\ff^r_1(e_1)-\ff^r_1(e_3)\right|. 
					\end{aligned}
				\end{equation}
			\end{small}

			Now, we try to bound $\left|\ff^r_1(e_1)-\ff^r_1(e_3)\right|$.
			As defined in Eq. \eqref{eqn:error_demand_1_2cfra}, error in demand on vertex $xy$ and $xy'$ with respect to commodity 1, we have
			\[\abs{\ff^r_1(e_1)+\ff^r_1(e_2)-\ff^r_1(e_4)-\ff^r_1(e_5)}\leq \eps^r,\]
			\[\abs{\ff^r_1(e_2)+\ff^r_1(e_3)-\ff^r_1(e_6)-\ff^r_1(e_7)}\leq \eps^r.\]
			We also have $\ff^r_1(e_5), \ff^r_1(e_7)\leq \eps^r$. Combining these inequalities gives
			\begin{align*}
				\left|\ff^r_1(e_1)-\ff^r_1(e_3)\right|
				\leq \left|\ff^r_1(e_4)-\ff^r_1(e_6)\right|+4\eps^r,
			\end{align*}
			which can be plugged back into Eq. \eqref{eq: lu-2cfa} to give
			\[\hat{\tau}_{d1}^f\leq 5|E^f|\eps^r+\max_{y\in V^f\backslash\{s_1,t_1,s_2,t_2\}}\sum_{e=(x,y) \in E^f}\left|\ff^r_1(e_4)-\ff^r_1(e_6)\right|.\]
			
			The following is a claim that shows the sum of $\abs{\ff^r_1(e_4)-\ff^r_1(e_6)}$ for all edges $e\in E^f$ can be upper bounded.
			
			\begin{claim}
				\label{cl: 2cfra}
				\[\sum_{e\in E^f} \abs{\ff^r_1(e_4)-\ff^r_1(e_6)}\leq 6|E^f|\eps^r.\]
			\end{claim}	
			
			By Claim \ref{cl: 2cfra}, we have 
			\begin{align*}
				\hat{\tau}_{d1}^f&\leq 5|E^f|\eps^r+\max_{y\in V^f\backslash\{s_1,t_1,s_2,t_2\}}\sum_{e=(x,y) \in E^f}\left|\ff^r_1(e_4)-\ff^r_1(e_6)\right|\\
				&\leq 5|E^f|\eps^r+\sum_{e\in E^f} \abs{\ff^r_1(e_4)-\ff^r_1(e_6)} \leq11|E^f|\eps^r.
			\end{align*}
		
		\end{itemize}
		Combining all the three cases, we can bound the error in demand for commodity 1 uniformly by
		\[\tau_{d1}^f:=\max\{\tilde{\tau}^f_{d1}, \bar{\tau}^f_{d1}, \hat{\tau}^f_{d1}\}\leq 12|E^f|\eps^r.\]			
		And by symmetry of the two commodities, we can also bound the error in demand for commodity 2 by
		\[\tau_{d2}^f\leq 12|E^f|\eps^r.\]				
	\end{enumerate}

To summarize, we can set
\[\tau^f=\max\{\tau^f_u, \tau^f_l,\tau^f_{d1}, \tau^f_{d2}\}\leq 12|E^f|\eps^r.\]

As we set in the reduction that $\eps^r=\frac{\eps^f}{12|E^f|}$, then we have
\[\tau^f\leq 12|E^f|\cdot\frac{\eps^f}{12|E^f|}=\eps^f, \]
indicating that $\ff^f$ is a solution to the \textsc{2cffa} instance $\left(G^f,F^f, \uu^f, s_1,t_1,s_2,t_2, \epsilon^f\right)$.
\end{proof}

Now, we provide a proof to Claim \ref{cl: 2cfra}.
\begin{proof}[Proof of Claim~\ref{cl: 2cfra}]
	We divide all edges in $E^f$ into two groups: 
	\[E_I=\{e\in E^f ~s.t. ~\ff^r_1(e_4)\leq \ff^r_1(e_6)\},\]
	\[E_{II}=\{e\in E^f~ s.t.~ \ff^r_1(e_4)> \ff^r_1(e_6)\}.\]
	We denote
	\[T_I=\sum_{e\in E_I}\ff^r_1(e_4), \qquad S_I=\sum_{e\in E_I}\ff^r_1(e_6);\]
	\[T_{II}=\sum_{e\in E_{II}}\ff^r_1(e_4), \qquad S_{II}=\sum_{e\in E_{II}}\ff^r_1(e_6).\]
	Then,
	\[
	\sum_{e\in E^f} \abs{\ff^r_1(e_4)-\ff^r_1(e_6)}=(S_I-T_I)+(T_{II}-S_{II})	
	\]
	On  one hand, by error in congestion, we have
	\begin{equation}
		\label{eq: lu2cf}
		S_I+T_{II}\leq \sum_{e\in E^f} \uu^f(e)+|E^f|\eps^r=M^f+|E^f|\eps^r.
	\end{equation}
	On the other hand, applying error in demand on vertex $t_i$ and $s_i$, we have
	\[T_I+T_{II}\geq (2M^f-\eps^r)-\ff^r_1(z_1,\bar{t}_1)\geq 2M^f-\eps^r-M^f-\eps^r=M^f-2\epsilon^r,\]
	\[S_I+S_{II}\geq(2M^f-\eps^r)-\ff^r_1(\bar{s}_1, z_1')\geq2M^f-\eps^r-M^f-\eps^r=M^f-2\epsilon^r.\]
	Thus,
	\[S_I+T_{II}=S_I+(T_I+T_{II})-T_{I}\geq M^f-2\epsilon^r+(S_I-T_{I}),\]
	\[S_I+T_{II}=(S_I+S_{II})+T_{II}-S_{II}\geq M^f-2\epsilon^r+(T_{II}-S_{II}).\]
	Together with Eq. \eqref{eq: lu2cf}, we obtain
	\[S_I-T_I\leq S_I+T_{II}-M^f+2\epsilon^r\leq M^f+|E^f|\eps^r-M^f+2\epsilon^r=(|E^f|+2)\epsilon^r,\]
	\[T_{II}-S_{II}\leq S_I+T_{II}-M^f+2\epsilon^r\leq M^f+|E^f|\eps^r-M^f+2\epsilon^r=(|E^f|+2)\epsilon^r.\]
	Hence, we have
	\[\sum_{e\in E^f} \abs{\ff^r_1(e_4)-\ff^r_1(e_6)}=(S_I-T_I)+(T_{II}-S_{II})\leq 2(|E^f|+2)\epsilon^r\leq 6|E^f|\eps^r,\]
	which finishes the proof.
\end{proof}


\subsection{2CFR(A) to 2CF(A)}
\label{sect: 2cfra-2cfa}
\subsubsection{2CFR to 2CF}
We show the reduction from a \textsc{2cfr} instance $(G^r, \uu^r, \bar{s}_1,\bar{t}_1, \bar{s}_2, \bar{t}_2, R_1, R_2)$ to a \textsc{2cf} instance\\ $(G^{2cf}, \uu^{2cf}, \bar{\bar{s}}_1,\bar{t}_1, \bar{\bar{s}}_2, \bar{t}_2, R^{2cf})$. 
We copy the entire graph structure of $G^r$ to $G^{2cf}$. In addition, we add to $G^{2cf}$ with two new sources $\bar{\bar{s}}_1, \bar{\bar{s}}_2$ and two new edges $(\bar{\bar{s}}_1,\bar{s}_1), (\bar{\bar{s}}_2, \bar{s}_2)$ with capacity $R_1, R_2$, respectively. We set $R^{2cf}=R_1+R_2$. 

If a \textsc{2cf} solver returns $\ff^{2cf}$ for the \textsc{2cf} instance $(G^{2cf}, \uu^{2cf}, \bar{\bar{s}}_1,\bar{t}_1, \bar{\bar{s}}_2, \bar{t}_2, R^{2cf})$, then we return $\ff^r$ for the \textsc{2cfr} instance $(G^r, \uu^r, \bar{s}_1,\bar{t}_1, \bar{s}_2, \bar{t}_2, R_1, R_2)$ by
setting $\ff_i^{r}(e)=\ff^{2cf}_i(e), \forall e\in E^r, i\in\{1,2\}$.
If the \textsc{2cf} solver returns ``infeasible'' for the \textsc{2cf} instance, 
then we return ``infeasible'' for the \textsc{2cfr} instance.

\begin{lemma}[\textsc{2cfr} to \textsc{2cf}]
	\label{lm: 2cfr-2cf}
	Given a \textsc{2cfr} instance $(G^r, \uu^r, \bar{s}_1,\bar{t}_1, \bar{s}_2, \bar{t}_2, R_1, R_2)$,
	we can construct, in time $O(|E^r|)$, a \textsc{2cf} instance $(G^{2cf}, \uu^{2cf}, \bar{\bar{s}}_1,\bar{t}_1, \bar{\bar{s}}_2, \bar{t}_2, R^{2cf})$ such that
	\[|V^{2cf}|= |V^r|+2, \qquad |E^{2cf}|=|E^r|+2, \qquad \norm{\uu^{2cf}}_{\max}=\max\{R_1, R_2\}, \qquad R^{2cf}=R_1+R_2.\]
	If the \textsc{2cfr} instance has a solution, then the \textsc{2cf} instance has a solution.
\end{lemma}	
\begin{proof}
	According to the reduction described above, if there exists a solution $\ff^{r}$ to the \textsc{2cfr} instance such that $F_1^r\geq R_1, F_2^r\geq R_2$, then there also exists a solution $\ff^{2cf}$ to the \textsc{2cf} instance because $F_1^{2cf}=R_1, F_2^{2cf}=R_2$, and thus $F_1^{2cf}+F_2^{2cf}=R^{2cf}$.
	
	For the problem size after reduction, it is obvious that
	\[|V^{2cf}|=|V^r|+2, \qquad |E^{2cf}|=|E^r|+2, \qquad \norm{\uu^{2cf}}_{\max}=\max\{R_1, R_2\}.\] 
	For the reduction time, it takes constant time to add two new vertices and edges, and it take $O(|E^r|)$ time to copy the rest edges. Thus, the reduction of this step can be performed in $O(|E^r|)$ time.	
\end{proof}

\subsubsection{2CFRA to 2CFA}

\begin{definition}[2CF Approximate Problem (\textsc{2cfa})]
	\label{def: 2cfa}
	A \textsc{2cfa} instance is given by a \textsc{2cf} instance
	$(G,\uu, s_1, t_1, s_2, t_2,R)$ and an error parameter $\eps \in [0,1]$, 
	which we collect in a tuple $(G,\uu, s_1, t_1, s_2, t_2,R, \eps)$.
	We say an algorithm solves the \textsc{2cfa} problem, if, given any 
	\textsc{2cfa} instance, it returns 
	a pair of flows $\ff_1, \ff_2 \ge 0$ that satisfies
		\begin{align}
		& \ff_1(e) + \ff_2 (e) \le  \uu(e) + \epsilon, ~ \forall e \in E \label{eqn:error_congestion_2cfa} \\
		& \abs{\sum_{u: (u,v)\in E}\ff_i(u,v)-\sum_{w: (v,w)\in E}\ff_i(v,w)}\leq \eps, ~ \forall v \in V\backslash \{s_i, t_i\}, i\in\{1,2\} \label{eqn:error_demand_1_2cfa} \\
		& \abs{\sum_{w: (s_i,w)\in E} \ff_i(s_i,w)-F_i}\leq \eps, \qquad \abs{\sum_{u: (u,t_i)\in E} \ff_i(u,t_i)-F_i}\leq \eps, ~ i\in\{1,2\} \label{eqn:error_demand_2_2cfa} 
	\end{align}
	where $F_1+F_2 = R$\footnote{
		If we encode \textsc{2cf} as an \textsc{lp} instance, and approximately solve the \textsc{lp} with at most $\eps$ additive error. Then, the approximately solution also agrees with the error notions of \textsc{2cf}, except that we get $F_1+F_2\geq R-\eps$ instead of $F_1+F_2\geq R$.
		This inconsistency can be eliminated by setting $\eps'=2\eps$, and slightly adjusting $F_1, F_2$ to $F'_1, F'_2$ such that $F'_1+F'_2\geq R$. This way,
		we obtain an approximate solution to \textsc{2cf} with at most $\eps'$ additive error.}; or it correctly declares that the associated \textsc{2cf}
	instance is infeasible.
	We refer to the error in~\eqref{eqn:error_congestion_2cfa} as error in congestion, error in~\eqref{eqn:error_demand_1_2cfa} and \eqref{eqn:error_demand_2_2cfa} as error in demand.
\end{definition}

We use the same reduction method and solution mapping method in the exact case to the approximate case.

\begin{lemma}[\textsc{2cfra} to \textsc{2cfa}]	
	\label{lm: 2cfra-2cfa}
	Given a \textsc{2cfra} instance $(G^r, \uu^r, \bar{s}_1,\bar{t}_1, \bar{s}_2, \bar{t}_2, R_1, R_2, \epsilon^r)$, 
	we can reduce it to a \textsc{2cfa} instance $(G^{2cf}, \uu^{2cf}, \bar{\bar{s}}_1,\bar{t}_1, \bar{\bar{s}}_2, \bar{t}_2, R^{2cf}, \epsilon^{2cf})$ by letting 
	\[\eps^{2cf}=\frac{\eps^r}{4},\]
	and using Lemma \ref{lm: 2cfr-2cf} to construct a \textsc{2cf} instance $(G^{2cf}, \uu^{2cf}, \bar{\bar{s}}_1,\bar{t}_1, \bar{\bar{s}}_2, \bar{t}_2, R^{2cf})$ from the \textsc{2cfr} instance $(G^r, \uu^r, \bar{s}_1,\bar{t}_1, \bar{s}_2, \bar{t}_2, R_1, R_2)$.
	If $\ff^{2cf}$ is a solution to the \textsc{2cfa} (\textsc{2cf}) instance, then in time $O(|E^{r}|)$, we can compute a solution $\ff^r$ to the \textsc{2cfra} (\textsc{2cfr}, respectively) instance,
	where the exact case holds when $\eps^{2cf} = \eps^{r}=0$.

\end{lemma}	
\begin{proof}
	Based on the solution mapping method described above, it takes constant time to set the value of each entry of $\ff^r$, and $\ff^r$ has $|E^r|$ entries. Thus, such a solution mapping takes $O(|E^r|)$ time.		
	
	Now, we conduct an error analysis. We use $\tau^r$ to denote the error of $\ff^r$ that is obtained by mapping back $\ff^{2cf}$ with at most $\eps^{2cf}$ additive error.
	By the error notions of \textsc{2cfra} (Definition \ref{def: 2cfra-new}), there are two types of error to track after solution mapping: (1) error in congestion $\tau^r_u$; (2) error in demand $\tau^r_{d1}, \tau^r_{d2}$. Then, we set the additive error of \textsc{2cfra} as
	\[\tau^r=\max\{\tau^r_u, \tau^r_{d1}, \tau^r_{d2}\}.\]
	\begin{enumerate}
		\item Error in congestion.\\
		Error in congestion does not increase since flows and edge capacities are unchanged for those edges other than $(\bar{\bar{s}}_1, \bar{s}_1), (\bar{\bar{s}}_2, \bar{s}_2)$. Therefore, it suffices to set $\tau^r_u= \eps^{2cf}$.
		\item Error in demand.\\
		It is noticed that only the incoming flow of vertex $\bar{s}_1, \bar{s}_2$ changes by solution mapping. Therefore, error in demand does not increase for vertices other than $\bar{s}_1, \bar{s}_2$, thus we only need to bound the error in demand of $\bar{s}_1, \bar{s}_2$.
		We consider commodity $i, i\in\{1,2\}$. 
		\begin{itemize}
			\item \textbf{Case 1:} For vertex $\bar{s}_{\bar{i}}$.\\
			We use $\tilde{\tau}^r_{di}$ to denote error in demand of $\ff^r$ for vertex $\bar{s}_{\bar{i}}$ with respect to commodity $i$.
			We have $\tilde{\tau}^r_{di}=\sum_{w: (\bar{s}_{\bar{i}},w)\in E}\ff^{r}_i(\bar{s}_{\bar{i}},w)\leq 2\eps^{2cf}$ because error in demand is accumulated twice over two vertices $\bar{\bar{s}}_{\bar{i}}, \bar{s}_{\bar{i}}$. 
			\item \textbf{Case 2:} For vertex $\bar{s}_i$.\\
			We use $\bar{\tau}^r_{di}$ to denote error in demand of $\ff^r$ for vertex $\bar{s}_{i}$ with respect to commodity $i$.
			We need to bound $\abs{\sum_{w: (\bar{s}_i,w)\in E}\ff^{2cf}_i(\bar{s}_i,w)-R_i}$.
			By construction, we have
			\begin{equation}
				\label{eqn:error_in_requirement}
				F_1^{2cf}+F_2^{2cf}= R^{2cf}=R_1+R_2.
			\end{equation}
			
			Applying error in congestion as defined in Eq. \eqref{eqn:error_congestion_2cfa} to edges $(\bar{\bar{s}}_i,\bar{s}_i), i\in\{1,2\}$, we have
			\begin{equation}
				\label{eqn:error_in_congestion}
				\ff^{2cf}_i(\bar{\bar{s}}_i,\bar{s}_i)-R_i\leq \ff^{2cf}(\bar{\bar{s}}_i,\bar{s}_i)-R_i\leq\eps^{2cf}.
			\end{equation}

			Applying error in demand as defined in Eq. \eqref{eqn:error_demand_1_2cfa} to vertices $\bar{s}_i, i\in\{1,2\}$, we have
			\begin{equation}
				\label{eqn:error_in_demand_s_2}
				\abs{\ff^{2cf}_i(\bar{\bar{s}}_i,\bar{s}_i)-\sum_{w: (\bar{s}_i,w)\in E}\ff^{2cf}_i(\bar{s}_i,w)}\leq \eps^{2cf}.
			\end{equation}

			Applying error in demand as defined in Eq. \eqref{eqn:error_demand_2_2cfa} to vertices $\bar{\bar{s}}_i, i\in\{1,2\}$, we have
			\begin{equation}
				\label{eqn:error_in_demand_s_1}
				\abs{\ff^{2cf}_i(\bar{\bar{s}}_i,\bar{s}_i)-F_i^{2cf}}\leq \eps^{2cf}.
			\end{equation}
			
			Combining Eqs. \eqref{eqn:error_in_requirement}, \eqref{eqn:error_in_congestion}, \eqref{eqn:error_in_demand_s_1} gives
			\begin{equation}
				\label{eqn:2}
				\abs{F_i^{2cf}-R_i}\leq 2\eps^{2cf}, ~i\in\{1,2\};
			\end{equation}
			and combining Eqs. \eqref{eqn:error_in_demand_s_2}, \eqref{eqn:error_in_demand_s_1} gives
			\begin{equation}
				\label{eqn:1}
				\abs{F_i^{2cf}-\sum_{w: (\bar{s}_i,w)\in E}\ff^{2cf}_i(\bar{s}_i,w)}\leq 2\eps^{2cf}, ~i\in\{1,2\}.
			\end{equation}
			
			Eqs. \eqref{eqn:2}, \eqref{eqn:1} gives
			\begin{equation*}
				\bar{\tau}^r_{di}=\abs{\sum_{w: (\bar{s}_i,w)\in E}\ff^{2cf}_i(\bar{s}_i,w)-R_i}\leq 4\eps^{2cf}, ~i\in\{1,2\}.
			\end{equation*}
		\end{itemize}
	Combining all cases, we can bound the error in demand for commodity $i$ uniformly by
	\[\tau^r_{di}:=\max\{\tilde{\tau}^r_{di}, \bar{\tau}^r_{di}\}\leq 4\eps^{2cf}.\]
	\end{enumerate}
	To summarize, we can set
	\[\tau^r=\max\{\tau^r_u, \tau^r_{d1}, \tau^r_{d2}\}\leq 4\eps^{2cf}.\]
	
	As we set in the reduction that $\eps^{2cf}=\frac{\eps^r}{4}$, then we have
	\[\tau^{r}\leq 4\cdot \frac{\eps^r}{4} =\eps^{r},\]
	indicating that $\ff^r$ is a solution to the \textsc{2cfra} instance $(G^r, \uu^r, \bar{s}_1,\bar{t}_1, \bar{s}_2, \bar{t}_2, R_1, R_2, \epsilon^r)$.
\end{proof}



\section{Main Theorem}
\label{sect: main}

Now, we are ready to prove the main theorem. 
\begin{theorem}[Restatement of Theorem~\ref{thm: main}.]
	\label{thm: main_formal}
	Given an \textsc{lpa} instance $(\AA,\bb,\cc,K,R,\epsilon^{lp})$ where $\AA\in\Z^{m\times n}, \bb\in Z^m, \cc\in \Z^n, K\in\Z$, we can construct, 
	in time $O(\nnz(\AA)\log X)$ where $X = X(\AA,\bb,\cc,K)$, a \textsc{2cfa} instance $(G^{2cf},\uu^{2cf}, s_1, t_1, s_2, t_2, R^{2cf}, \epsilon^{2cf})$ such that 
	\[|V^{2cf}|, |E^{2cf}|\leq 10^6\nnz(\AA)(3+\log X), \]
	\[\norm{\uu^{2cf}}_{\max}, R^{2cf}\leq 10^8 \nnz^3(\AA)RX^2 (2+\log X)^2,\]
	\[\epsilon^{2cf}\geq\frac{\eps^{lp}}{10^{24}\nnz^7(\AA)RX^3(3+\log X)^6}.\]
	If the \textsc{lp} instance $(\AA,\bb,\cc,K,R)$ has a solution, then the \textsc{2cf} instance $(G^{2cf},\uu^{2cf}, s_1, t_1, s_2, t_2, R^{2cf})$ has a solution.
	Furthermore, if $\ff^{2cf}$ is a solution to the \textsc{2cfa} (\textsc{2cf}) instance, then in time $O(\nnz(\AA)\log X)$, we can compute a solution $\xx$ to the \textsc{lpa} (\textsc{lp}, respectively) instance, where the exact case holds when $\eps^{2cf}=\eps^{lp}=0$.
\end{theorem}
\begin{proof}
	The theorem can be proved by putting all lemmas related to approximate problems in Section \ref{sect: proof} together. 
	Given an \textsc{lpa} instance $(\AA,\bb,\cc,K,R,\epsilon^{lp})$, for each problem along the reduction chain, we bound its problem size, problem error, reduction time, and solution mapping time, 
	in terms of the input parameters of the \textsc{lpa} instance. 
	\begin{enumerate}
		\item \textsc{lpa}
		 \begin{itemize}
		 	\item problem size: $n, m, \nnz(\AA), R, X = X(\AA,\bb,\cc,K)$
		 	\item problem error: $\epsilon^{lp}$ 
		 	\item reduction time: $\backslash$
		 	\item solution mapping time (Lemma \ref{lm: lpa-lena}) \footnote{The time needed to turn a solution to the next problem in the 
			 reduction chain to a solution to this problem.}: $O(n)$
		 \end{itemize}

	 	\item \textsc{lena} \\
	 	\textbf{Note:} for simplification, in the following calculations, 
		 we replace $n,m$ with $\nnz(\AA)$ since wlog we can assume $1\leq n,m\leq \nnz(\AA)$.
	 	\begin{itemize}
	 		\item problem size (Lemma \ref{lm: lp-len}): 
	 		\begin{align*}
	 			&\tilde{n} =n+m+1\leq 3\nnz(\AA)\\
	 			&\tilde{m}=m+1\leq2\nnz(\AA)\\
	 			&\nnz(\tilde{\AA})\leq 4\nnz(\AA)\\ 
	 			&\tilde{R}= 5mRX\leq 5\nnz(\AA)RX\\
	 			&X(\tilde{\AA},\tilde{\bb})=X
	 		\end{align*}	 		
	 		\item problem error (Lemma \ref{lm: lpa-lena}): $\epsilon^{le}=\epsilon^{lp}$ 
	 		\item reduction time (Lemma \ref{lm: lp-len}): $O(\nnz(\AA))$
	 		\item solution mapping time (Lemma \ref{lm: lena-2lena}): $O(\tilde{n})\leq O(\nnz(\AA))$
	 	\end{itemize}
 	
 		\item \textsc{2-lena} (by Lemma \ref{lm: lena-2lena} and Lemma \ref{lm: 2lena-1lena})
 		\begin{itemize}
 			\item problem size (Lemma \ref{lm: len-2len}): 
 			\begin{align*}
 				&\bar{n}\leq \tilde{n}+4\tilde{m}\left(1+\log X(\tilde{\AA},\tilde{\bb})\right)\leq 8\nnz(\AA)(2+\log X)\\
 				&\bar{m}\leq 3\tilde{m}\left(1+\log X(\tilde{\AA},\tilde{\bb})\right)\leq 6\nnz(\AA)(1+\log X)\\
 				&\nnz(\bar{\AA})\leq 17\nnz(\tilde{\AA})\left(1+\log X(\tilde{\AA},\tilde{\bb})\right)\leq 68\nnz(\AA)(1+\log X)\\
 				&\bar{R}= 8\tilde{m}\tilde{R}X(\tilde{\AA}, \tilde{\bb})\left(1+\log X(\tilde{\AA}, \tilde{\bb})\right)\leq 80\nnz^2(\AA)RX^2(1+\log X)\\
 				&X(\bar{\AA}, \bar{\bb})=2X(\tilde{\AA},\tilde{\bb})\tilde{R}=10\nnz(\AA)RX^2
 			\end{align*}
 			\item problem error (Lemma \ref{lm: lena-2lena}): 
 			\[\epsilon^{2le}=\frac{\eps^{le}}{2X(\tilde{\AA},\tilde{\bb})}=\frac{\epsilon^{lp}}{2X}\]
 			\item reduction time Lemma \ref{lm: len-2len}: $O\left(\nnz(\tilde{\AA})\log X(\tilde{\AA}, \tilde{\bb})\right)=O(\nnz(\AA)\log X)$
 			\item solution mapping time (Lemma \ref{lm: 2lena-1lena}): $O(\bar{n})=O(\nnz(\AA)\log X)$
 		\end{itemize}
 	
 		\item \textsc{1-lena} (by Lemma \ref{lm: 2lena-1lena} and Lemma \ref{lm: 1lena-fhfa})
 		\begin{itemize}
 			\item problem size (Lemma \ref{lm: 2len-1len}): 
 			\begin{align*}
 				&\hat{n}\leq 2\bar{n}\leq 16\nnz(\AA)(2+\log X)\\
 				&\hat{m}\leq \bar{m}+\bar{n}\leq 14\nnz(\AA)(2+\log X)\\ 
 				&\nnz(\hat{\AA})\leq 4\nnz(\bar{\AA})\leq 272\nnz(\AA)(1+\log X)\\
 				&\hat{R}=2\bar{R}\leq 160\nnz^2(\AA)RX^2(1+\log X)\\
 				&X(\hat{\AA},\hat{\bb})=X(\bar{\AA},\bar{\bb})=10\nnz(\AA)RX^2
 			\end{align*}

 			\item problem error (Lemma \ref{lm: 2lena-1lena}):
 			\begin{align*}
 				\epsilon^{1le}=\frac{\eps^{2le}}{\bar{n}+1}\geq\frac{\epsilon^{lp}}{16\nnz(\AA)X(3+\log X)}
 			\end{align*}

 			\item reduction time (Lemma \ref{lm: 2len-1len}): $O(\nnz(\bar{\AA}))=O(\nnz(\AA)\log X)$
 			\item solution mapping time (Lemma \ref{lm: 1lena-fhfa}): $O(\hat{n})=O(\nnz(\AA)\log X)$
 		\end{itemize}
 		
 		\item \textsc{fhfa} 
 		\begin{itemize}
 			\item problem size (Lemma \ref{lm: 1len-fhf}): 
 			\begin{align*}
 				&|V^h|=2\hat{m}+2\leq 28\nnz(\AA)(3+\log X)\\
 				&|E^h|\leq4\nnz(\hat{\AA})\leq 1088\nnz(\AA)(1+\log X)\\
 				&|F^h|=\hat{m}\leq 14\nnz(\AA)(2+\log X)\\
 				&h=\hat{n}+\hat{m}\leq 30\nnz(\AA)(2+\log X)\\
 				&\norm{\uu^h}_{\max}=\max\left\{\hat{R}, X(\hat{\AA}, \hat{\bb})\right\}\leq 160\nnz^2(\AA)RX^2(1+\log X)
 			\end{align*}

 			\item problem error (Lemma \ref{lm: 1lena-fhfa}):
 			\begin{align*}
 				\epsilon^h=\frac{\eps^{1le}}{5\hat{n}X(\hat{\AA},\hat{\bb})}\geq\frac{\eps^{lp}}{12800\nnz^3(\AA)RX^3(3+\log X)^2}
 			\end{align*}

 			\item reduction time (Lemma \ref{lm: 1len-fhf}): $O(\nnz(\hat{\AA}))=O(\nnz(\AA)\log X)$
 			\item solution mapping time (Lemma \ref{lm: fhfa-fphfa}): $O(\abs{E^h})=O(\nnz(\AA)\log X)$
 		\end{itemize}
 		
 		\item \textsc{fphfa}
 		\begin{itemize}
 			\item problem size (Lemma \ref{lm: fhf-fphf}): 
 			\begin{align*}
 				&|V^p|\leq|V^h|+|E^h|\leq 1116\nnz(\AA)(3+\log X)\\
 				&|E^p|\leq2|E^h|\leq 2176\nnz(\AA)(1+\log X)\\
 				&|F^p|=|F^h|\leq 14\nnz(\AA)(2+\log X)\\
 				&p\leq |E^h|\leq 1088\nnz(\AA)(1+\log X)\\
 				&\norm{\uu^p}_{\max}=\norm{\uu^h}_{\max}\leq 160\nnz^2(\AA)RX^2(1+\log X)
 			\end{align*}
 	
 			\item problem error (Lemma \ref{lm: fhfa-fphfa}):
 			\begin{align*}
 				\epsilon^p=\frac{\epsilon^h}{|E^h|}\geq\frac{\eps^{lp}}{2\cdot 10^7\nnz^4(\AA)RX^3(3+\log X)^3}
 			\end{align*}

 			\item reduction time (Lemma \ref{lm: fhf-fphf}): $O(|E^h|)=O(\nnz(\AA)\log X)$
 			\item solution mapping time (Lemma \ref{lm: fphfa-sffa}): $O(\abs{E^p})=O(\nnz(\AA)\log X)$
 		\end{itemize}
 		\item \textsc{sffa} 
 		\begin{itemize}
 			\item problem size (Lemma \ref{lm: fphf-sff}):
 			\begin{align*}
 				&|V^s|=|V^{p}|+4p+2\leq 5468\nnz(\AA)(3+\log X)\\
 				&|E^s|=|E^{p}|+7p\leq 9792\nnz(\AA)(1+\log X)\\
 				&|F^s|=|F^p|+2p\leq 2190\nnz(\AA)(2+\log X)\\
 				&|S_1|=|E^{p}|+2p\leq 4352\nnz(\AA)(1+\log X)\\
 				&|S_2|=3p\leq 3264\nnz(\AA)(1+\log X)\\
 				&\norm{\uu^s}_{\max}=\norm{\uu^p}_{\max}\leq 160\nnz^2(\AA)RX^2(1+\log X)
 			\end{align*}

 			\item problem error (Lemma \ref{lm: fphfa-sffa}):
 			\begin{align*}
 				\epsilon^s=\frac{\eps^p}{11|E^p|}=\frac{\eps^{lp}}{5\cdot 10^{11}\nnz^5(\AA)RX^3(3+\log X)^4}
 			\end{align*}

 			\item reduction time (Lemma \ref{lm: fphf-sff}): $O(|E^p|)=O(\nnz(\AA)\log X)$
 			\item solution mapping time (Lemma \ref{lm: sffa-2cffa}): $O(|E^s|)=O(\nnz(\AA)\log X)$
 		\end{itemize}
 		\item \textsc{2cffa} 
 		\begin{itemize}
 			\item problem size (Lemma \ref{lm: sff-2cff}):
 			\begin{align*}
 				&|V^f|=|V^s|+2(|S_1|+|S_2|)\leq 20700\nnz(\AA)(3+\log X)\\
 				&|E^f|=|E^s|+4(|S_1|+|S_2|)\leq 40256\nnz(\AA)(1+\log X)\\
 				&|F^f|\leq4(|F^s|+|S_1|+|S_2|)\leq 39224\nnz(\AA)(2+\log X)\\
 				&\norm{\uu^f}_{\max}= \norm{\uu^s}_{\max}\leq 160\nnz^2(\AA)RX^2(1+\log X)
 			\end{align*}
 		
 			\item problem error (Lemma \ref{lm: sffa-2cffa}):	
 			\begin{align*}
 				\epsilon^f=\frac{\eps^s}{6|E^s|}\geq\frac{\eps^{lp}}{3\cdot 10^{16}\nnz^6(\AA)RX^3(3+\log X)^5}
 			\end{align*}

 			\item reduction time (Lemma \ref{lm: sff-2cff}): $O(|E^s|)=O(\nnz(\AA)\log X)$
 			\item solution mapping time (Lemma \ref{lm: 2cffa-2cfra}): $O(|E^f|)=O(\nnz(\AA)\log X)$
 		\end{itemize}
 		
 		\item \textsc{2cfra} 
 		\begin{itemize}
 			\item problem size (Lemma \ref{lm: 2cff-2cfr}):
 			\begin{align*}
 				&|V^r|=|V^f|+2|E^f|+8\leq 2\cdot 10^5\nnz(\AA)(3+\log X),\\
 				&|E^r|=7|E^f|+10\leq 3\cdot 10^5\nnz(\AA)(3+\log X),\\
 				&\norm{\uu^r}_{\max}= \max\left\{2\norm{\uu^f}_{\max}, M^f\right\}\leq |E^f|\norm{\uu^f}_{\max}\leq 7\cdot 10^6 \nnz^3(\AA)RX^2 (1+\log X)^2,\\
 				&R_1=R_2=2M^f\leq 1.4\cdot 10^7 \nnz^3(\AA)RX^2 (1+\log X)^2
 			\end{align*}
 			\item problem error (Lemma \ref{lm: 2cffa-2cfra}):
 			\begin{align*}
 				\epsilon^r= \frac{\eps^f}{12|E^f|}\geq\frac{\eps^{lp}}{1.1\cdot 10^{23}\nnz^7(\AA)RX^3(3+\log X)^6}
 			\end{align*}
 			\item reduction time (Lemma \ref{lm: 2cff-2cfr}): $O(|E^f|)=O(\nnz(\AA)\log X)$
 			\item solution mapping time (Lemma \ref{lm: 2cfra-2cfa}): $O(|E^r|)=O(\nnz(\AA)\log X)$
 		\end{itemize}
 		\item \textsc{2cfa} 
 		\begin{itemize}
 			\item problem size (Lemma \ref{lm: 2cfr-2cf}):
 			\begin{align*}
	 			|V^{2cf}|, |E^{2cf}|&\leq 10^6\nnz(\AA)(3+\log X),\\
	 			\norm{\uu^{2cf}}_{max}, R^{2cf}&\leq 10^8 \nnz^3(\AA)RX^2 (2+\log X)^2
	 		\end{align*}
 	
 			\item problem error (Lemma \ref{lm: 2cfra-2cfa}):
 			\begin{align*}
 				\epsilon^{2cf}=\frac{\eps^r}{4}\geq\frac{\eps^{lp}}{5\cdot 10^{23}\nnz^7(\AA)RX^3(3+\log X)^6}
 			\end{align*}
 			\item reduction time (Lemma \ref{lm: 2cfr-2cf}): $O(|E^r|)=O(\nnz(\AA)\log X)$
 			\item solution mapping time: $\backslash$
 		\end{itemize}
	\end{enumerate}
\end{proof}

By applying Theorem \ref{thm: main_formal}, we can prove Corollary \ref{col}, 
which states that if there exists a fast high-accuracy \textsc{2cfa} solver, 
then there also exists an almost equally fast high-accuracy \textsc{lpa} solver.


\printbibliography

\end{document}